\documentclass{fundam}
\usepackage{amsmath}
\usepackage{enumitem}
\usepackage{multicol}
\usepackage{hyperref}
\usepackage[T1]{fontenc}
\usepackage{tikz}
\tikzstyle{every node} = [draw, fill=none, circle, inner sep=0pt, minimum size=5pt]
\tikzstyle{d} = [very thick]
\tikzstyle{n} = [draw=none, rectangle, inner sep=6pt] 
\tikzstyle{i} = [draw, fill=black, circle, inner sep=0pt, minimum size=5pt] 
\usepackage{graphicx}
\usepackage{color}

\usepackage{amssymb}

\newcommand{\m}{\mathbf}
\newcommand{\RR}{\circ}
\newcommand{\R}[3]{#3\in #1\circ #2}

\newcommand{\Rb}[3]{#3\in #1\bullet #2}

\begin{document}

\title{Duality theory and representations for  distributive quasi relation algebras and DInFL-algebras}
\runninghead{A. Craig, P. Jipsen, C. Robinson}{Duality and representations for DqRAs and DInFL-algebras}

\author{Andrew Craig \\
 Department of Mathematics and Applied Mathematics\\
 University of Johannesburg, South Africa \\and \\
 National Institute for Theoretical and Computational Sciences (NITheCS) \\ Johannesburg, South Africa \\ acraig@uj.ac.za\thanks{Funding from the National Research Foundation grant 127266 is gratefully acknowledged.}\\
		\and Peter Jipsen \\
        Chapman University \\
        Orange CA 92866, USA\\
        jipsen@chapman.edu \\
\and        Claudette Robinson \\
		Department of Mathematics and Applied Mathematics\\
		University of Johannesburg, South Africa\\
		claudetter@uj.ac.za
	}

\maketitle              

\vspace*{-80pt}

\begin{abstract}
We develop dualities for complete perfect distributive quasi relation algebras and complete perfect distributive involutive FL-algebras. The duals are partially ordered frames with additional structure. These frames are analogous to the atom structures used to study relation algebras. We also extend the duality from complete perfect algebras to all algebras, using so-called doubly-pointed frames with a Priestley topology. 

We then turn to the representability of these algebras as lattices of binary relations. Some algebras can be realised as term subreducts of representable relation algebras and are hence representable. We provide a detailed account of known representations for all algebras up to size six. 

\end{abstract}

\section{Introduction}
The algebraic theory of binary relations and its abstract version in the form of relation algebras have many applications in mathematics and computer science (see e.g., \cite{Mad2006,Mad1993,SS1993}). 
In mathematics these theories are foundational to algebraic logic and have important connections with group theory, combinatorics, and projective geometry. In computer science relation algebras are used for knowledge representation, spatial reasoning about regions, network theory and for formal specification of hardware and software. This is perhaps most evident from the use of binary relations to model the semantics of computer programs, where the underlying frames represent state spaces. In these contexts, the algebraic framework provides a rigorous basis for compositionality, allowing complex systems to be understood via their constituents. 

Both the theory of concrete and abstract binary relations are based on Boolean algebras, hence use classical negation (complementation) in an integral way. Recent generalizations have aimed to weaken these theories to some well-behaved nonclassical versions, such as weakening relation algebras and quasi relation algebras, where the unary operations of complementation and converse are not available individually. Here we concentrate on the variety of distributive quasi relation algebras (DqRAs), which builds upon the more general variety of distributive involutive full Lambek algebras (DInFL-algebras).

The variety of quasi relation algebras was first studied by Galatos and Jipsen~\cite{GJ2013}. One of the features of this variety and its subvariety DqRA is that they have decidable equational theories~\cite[Corollary 5.6]{GJ2013}. Distributive qRAs based on binary relations provide a notion of representability~\cite{RDqRA25} for DqRAs which mimics the intensely studied concept of representability for relation algebras. 

 Studying DInFL-algebras and DqRAs via (partially ordered) frames is motivated by the fact that a powerful tool for the study of finite relation algebras has been the so-called \emph{atom structures}. Any finite relation algebra can be studied simply via an operation table of its atoms. Frames are useful for implementing a decision procedure for DqRA since, e.g., tableaux methods build frame-based counterexamples or prove that none exist.

In applications to computing, relations are often treated in the setting of heterogeneous relation algebras or within category theory as allegories. In the current paper we consider only unisorted DqRAs and their frames, but it is possible to adapt these concepts to the heterogeneous and categorical settings. 

Another reason for wanting to study frames is for the development of a game-based approach to determine whether certain algebras are representable. This has been done very successfully in the relation algebra case \cite{Mad1983}, \cite{HH2002} and also more recently for weakening relation algebras~\cite{JS2023}. 

This paper is an expanded version of~\cite{CJR2024}. In particular, the Priestley duality between DqRAs and DqRA-spaces has been included, and an extensive discussion of representability for DInFL-algebras and DqRAs is 
followed by a catalog of finite algebras up to cardinality $6$ (up to $8$ in \cite{DInFL-list}). 

We begin with some background material on involutive FL-algebras with a De Morgan operation and the definition of qRAs. In Section~\ref{sec:frames}, we describe dual frames for both DInFL-algebras and DqRAs. 
We investigate dual frames in the setting of  complete perfect algebras (i.e., the completely join-irreducible elements are join-dense and the completely meet-irreducible elements are meet-dense). All results obviously apply to the case of finite algebras. In Section~\ref{sec:frame-morphisms}, we define morphisms for frames and show that they are dual to complete homomorphisms for complete perfect DInFL-algebras and DqRAs. We then extend our frame-based approach to a Priestley-style 
representation in Section~\ref{sec:dual-spaces}. 

Representability of DInFL-algebras and DqRAs is discussed in Section~\ref{sec:RDqRAs}.
Applications of frames to finite models and their representability are considered in Section~\ref{sec:fin-models}. 
In the first part we present tables with the number of both DInFL-algebras and DqRAs up to size eight. These were calculated using Prover9/Mace4~\cite{McC}. We verified these numbers by calculating the numbers of corresponding frames. 
In the remainder of 
Section~\ref{sec:fin-models}, we consider cyclic DqRAs which are ${\{\vee,\cdot,1,\sim\}}$-subreducts of 16-element relation algebras. From the representability of these relation algebras, we obtain representations of the term-subreducts. We also identify the dual frames of particular DqRAs and relate them to the atom structures of the relation algebras in which they can be embedded. In Tables~\ref{tab:Rep-DqRA<=5} and~\ref{tab:Rep-DqRA=6} we give a summary of known representations of DqRAs up to size six. 

\section{Background and results on quasi relation algebras}\label{sec:background}

An \emph{involutive full Lambek algebra} (briefly InFL-algebra) $\m A=(A,\wedge,\vee,\cdot,1,\sim,-)$ is a lattice $(A,\wedge,\vee)$ and a monoid $(A,\cdot, 1)$ such that for all $a,b,c\in A$
$$ 
\mathsf{(Res)} \qquad 
a\cdot b\leqslant c\iff a\leqslant -(b\cdot{\sim}c)\iff b\leqslant{\sim}(-c\cdot a).
$$ 
Hence $-(b\cdot{\sim}c)$ and ${\sim}(-c\cdot a)$ are terms for the right residual $c/b$ and the left residual $a\backslash c$ respectively. It follows that $-1={\sim}1$ and this element is usually denoted $0$. 
The operations $-,\sim$ are called \emph{linear negations} and are order-reversing as well as \emph{involutive}: $-{\sim}a=a={\sim}{-}a$. If the residuals and $0$ are included in the signature, then $\sim, -$ can be defined by ${\sim}a=a\backslash 0$ and $-a=0/a$, and we obtain InFL-algebras (in a term-equivalent form) if the identities $-{\sim}a=a={\sim}{-}a$ hold.

The same variety can also be axiomatized by the following idempotent semiring axioms \cite{JV2021}: an InFL-algebra
$\m A=(A,\vee,\cdot,1,\sim,-)$ is a semilattice $(A,\vee)$ and a monoid $(A,\cdot, 1)$ such that $\cdot$ distributes over $\vee$ and
$$a\leqslant b\iff a\cdot{\sim}b\leqslant -1\iff -b\cdot a\leqslant -1.$$
Moreover $-{\sim}a=a={\sim}{-}a$, and if we define $a\wedge b=-({\sim}a\vee{\sim}b)$ then $(A,\wedge,\vee)$ is a lattice. The operation $+$ is defined by $a+b=-({\sim}b\cdot{\sim}a)$, or equivalently by $a+b={\sim}(-b\cdot -a)$ \cite{GJKO2007}.

An InFL-algebra is \emph{cyclic} if ${\sim}a=-a$, \emph{commutative} if $a\cdot b=b\cdot a$ and \emph{distributive} if $a\wedge(b\vee c)=(a\wedge b)\vee(a\wedge c)$.

A \emph{De Morgan lattice} $\m A=(A,\wedge,\vee,\neg)$ is a lattice $(A,\wedge,\vee)$ with a unary operation that satisfies $\neg\neg a=a$ and 
$$\mathsf{(Dm)}\qquad \neg(a\wedge b)=\neg a\vee \neg b.$$
A \emph{De Morgan InFL-algebra} $\m A=(A,\wedge,\vee,\cdot,1,\sim,-,\neg)$ is a De Morgan lattice $(A,\wedge,\vee,\neg)$ and an InFL-algebra $(A,\vee,\cdot,1,\sim,-)$.

A \emph{quasi relation algebra} (qRA for short) is a De Morgan InFL-algebra such that the following identity holds:
$$\mathsf{(Dp)}\qquad \neg (a\cdot b) = \neg a+\neg b.$$

Interesting qRAs that are not relation algebras  
include
the finite Sugihara chains (where ${\sim}a=-a=\neg a$). 
See Table~\ref{tab:smallqRA}
and~\cite[Figure 1]{RDqRA25} for further examples of (non)cyclic qRAs.

In \cite{GJ2013}, the definition of qRAs included a third identity, $\mathsf{(Di)}$:  $\neg{\sim} a={-}\neg a$, 
but we note here that it is implied by the above definition. 
The proof follows from applications of $\mathsf{(Dp)}$ and $\mathsf{(Res)}$ from above. 

\begin{lemma}\label{lem:Di-in-qRA}
The identities ${\sim} 1=\neg 1=-1$ and $\mathsf{(Di)}$ hold in any qRA.
\end{lemma}

\begin{proof} It suffices to show $-1=\neg 1$ since ${\sim}1=-1$ holds in every InFL-algebra. 

From $\mathsf{(Dp)}$ with $b=1$ we obtain $\neg a=\neg a+\neg 1=-({\sim}{\neg}1\cdot {\sim}{\neg}a)$. Applying $\sim$ on both sides and letting $a=\neg{-}1$ we get ${\sim}\neg \neg{-}1={\sim}{\neg}1\cdot {\sim}{\neg}\neg{-}1$, which simplifies to $1={\sim}\neg 1\cdot 1={\sim}\neg 1$. Hence $-1=\neg 1$.

For $\mathsf{(Di)}$, we first show that $-\neg a \leqslant \neg{\sim}a$. We have ${\sim}a\cdot 1\leqslant {\sim}a$, and therefore, by $\mathsf{(Res)}$, $1 \leqslant {\sim}\left(-{\sim}a \cdot {\sim}a\right) = {\sim}\left(a \cdot {\sim}a\right)$. Hence, $a\cdot {\sim}a\leqslant -1 = \neg 1$, and so $1 \leqslant \neg\left(a\cdot {\sim}a\right)$. Applying $\mathsf{(Dp)}$ to the right-hand side gives $1 \leqslant \neg a +\neg{\sim} a = {\sim}\left(-\neg{\sim} a\cdot -\neg a\right)$. Thus, $-\neg{\sim} a\cdot -\neg a \leqslant -1 = \neg 1$, and applying $\mathsf{(Res)}$ to this yields $-\neg a \leqslant {\sim}\left(-\neg 1\cdot -\neg {\sim}a\right) = \neg{\sim}a + \neg 1$. Another application of $\mathsf{(Dp)}$ gives $-\neg a \leqslant \neg\left(1 \cdot {\sim}a\right) = \neg{\sim}a$.

We also have $a \leqslant \neg\neg a= \neg\left(\neg a \cdot 1\right)$, and so, by $\mathsf{(Dp)}$, $a \leqslant \neg\neg a + \neg 1 = a +\neg 1 = {\sim}\left(-\neg 1\cdot - a\right)$. Hence, by $\mathsf{(Res)}$, $-a\cdot a \leqslant \neg 1 = -1$, which implies that $1 \leqslant {\sim}\left(-a\cdot a\right) = {\sim}\left(-\neg \neg a\cdot -\neg\neg{\sim} a\right) = \neg\neg{\sim}a + \neg\neg a$. Applying $\mathsf{(Dp)}$ to the right-hand side gives $1 \leqslant \neg\left(\neg{\sim}a \cdot \neg a\right)$. Thus $\neg{\sim}a\cdot \neg a\leqslant \neg 1$, and hence, by $\mathsf{(Res)}$, $\neg{\sim}a\leqslant -\left(\neg a\cdot {\sim}\neg 1\right) = -\left(\neg a\cdot 1\right) = -\neg a$. 
\end{proof}

The next result shows that qRAs can be obtained from commutative InFL-algebras.

\begin{lemma}
Commutative InFL-algebras are quasi relation algebras.
\end{lemma}
\begin{proof} Commutative InFL-algebras satisfy $a\backslash b=b/a$, hence they are cyclic. Defining $\neg a={\sim}a$, it is clear that 
$\mathsf{(Dm)}$ holds, and $\mathsf{(Dp)}$ follows from the definition of $a+b=-({\sim}b\cdot{\sim}a)$ and commutativity.
\end{proof}

Distributive InFL-algebras and distributive qRAs are called DInFL-algebras and DqRAs for short. By the preceding lemma, DqRAs that satisfy the identity $\neg a={\sim} a$ can be obtained from commutative DInFL-algebras.

A relation algebra is a cyclic DqRA if the linear negations $\sim,-$ are defined as \emph{complement-converse} $(\neg a)^\smallsmile=\neg (a^\smallsmile)$ and a cyclic DqRA is a relation algebra if $\neg$ is complementation, i.e., $a\vee \neg a=\top$ and $a\wedge \neg a=\bot$. 

DqRAs can also be obtained from term-subreducts of relation algebras with respect to the signature $\{\vee,\cdot,1,\sim\}$ where, as before, ${\sim}a=-a=(\neg a)^\smallsmile$.
In the case of representable relation algebras, these subreducts are known as \emph{(representable) weakening relation algebras} \cite{GJ2020}, \cite{GJ2020a}, \cite{JS2023}. If the subreducts are commutative, they are DqRAs  by the preceding lemma.
However, the next result implies that one should start with nonsymmetric relation algebras if the aim is to construct DqRAs that are not relation algebras.

An element $a$ in a DqRA is said to be \emph{symmetric} if $\neg a={\sim}a=-a$ and the whole algebra is \emph{symmetric} if this property holds for all elements. 

\begin{lemma}
Suppose $\m A$ is a symmetric relation algebra. Then all InFL-algebra subreducts of $\m A$ are also relation algebras with $a^\smallsmile=a$.
\end{lemma}
\begin{proof}
In a symmetric relation algebra the classical negation $\neg$ and the linear negations $\sim,-$ all coincide, hence subalgebras with respect to the InFL operations are in fact relation algebras.
\end{proof}

\section{Frames for distributive quasi relation algebras and distributive involutive FL-algebras}\label{sec:frames}

In this section, we begin by recalling 
the definition of perfect distributive lattices. For a lattice $\mathbf{A}= (A, \wedge, \vee)$ we denote by $J^\infty(\mathbf{A})$ and $M^\infty(\mathbf{A})$ the completely join-irreducible and completely meet-irreducible elements, respectively. 
A lattice $\mathbf{A}$ is said to be \emph{perfect} if for every $a \in A$, we have 
$a = \bigvee \{\, j \in J^\infty(\mathbf{A}) \mid j \leqslant a \,\} = \bigwedge \{\, m \in M^\infty(\mathbf{A}) \mid a \leqslant m\,\}$, i.e., $J^\infty(\mathbf{A})$ 
 is join-dense and $M^\infty(\mathbf{A})$ is meet-dense.
Some authors include completeness as part of the definition of   perfect, but we state it separately. 
Here we will 
give a representation for DInFL-algebras and DqRAs whose underlying lattices are complete perfect distributive lattices. The 
theorem below follows from Gehrke and J\'{o}nsson~\cite[Theorem 2.2]{GJ94}.

\begin{theorem} For  a complete distributive lattice $\mathbf{A}$, the following are equivalent:
\begin{enumerate}[label=\textup{(\roman*)}]
    \item $\mathbf{A}$ is completely distributive and $J^\infty(\mathbf{A}) $ is completely join-dense in $\mathbf{A}$;
    \item $J^\infty(\mathbf{A})$ is completely join-dense in $\mathbf{A}$ and $M^\infty(\mathbf{A})$ is completely meet-dense in $\mathbf{A}$;
    \item $\mathbf{A}$ is isomorphic to the lattice of upsets of some poset $P$.
\end{enumerate}
\end{theorem}

We remark that in a completely distributive lattice, all completely join-irreducible elements are completely join-prime and all completely meet-irreducible elements are completely meet-prime. 

Next we present complete perfect  DqRAs and complete perfect DInFL-algebras by frames, similar to the frames of finite representable weakening relation algebras \cite{JS2023} and atom structures of atomic relation algebras \cite{Mad1982}.

The frames for DInFL-algebras are essentially Routley--Meyer style frames from relevance logic~\cite{BDM09}, but with the relevant negation replaced by two linear negations that are inverses of each other. 

Note that given a binary set operation $\circ: W \times W \to \mathcal{P}(W)$, it can be extended to
a binary operation on $\mathcal P(W)$ via $U\circ V=\bigcup\{x\circ y\mid x\in U,y\in V\}$. We will  also make use of the notation $x\circ V=\{x\}\circ V$ and $U\circ y=U\circ \{y\}$.

\begin{definition}\label{def:DInFL-frames}
A \emph{DInFL-frame} is a tuple $\mathbb{W}= \left(W, I, \preccurlyeq, \circ, {^{\sim}}, {^{-}}\right)$ with 
$I\subseteq W$, a partial order $\preccurlyeq$ on $W$, a binary set-operation $\circ:W\times W\to \mathcal P(W)$ 
and functions $^{\sim} : W \to W$ and $^{-} : W\to W$ such that the following conditions hold for all $u,v, x, y, z \in W$:
\begin{enumerate}[label=\textup{(\arabic*)}]
\item $x \preccurlyeq y \Longleftrightarrow \R{I}{x}{y}\Longleftrightarrow \R{x}{I}{y}$ 
\item $x \preccurlyeq y \text{ and } x\in I \implies y\in I$
\item $x \preccurlyeq y \text{ and } \R{u}{v}{x} \implies \R{u}{v}{y}$
\item $(x\circ y)\circ z=x\circ (y\circ z)$ 
\item $\R{x}{y}{z^\sim} \Longleftrightarrow \R{z}{x}{y^{-}}$
\item $x^{\sim -} \preccurlyeq x$ and $  x^{-\sim} \preccurlyeq x$
\end{enumerate}
\end{definition}

The expression $z\in x\circ y$ is equivalent to a ternary relation that is usually denoted by $Rxyz$ (e.g., in relevance logic). Hence conditions (1)--(6) are equivalent to first-order formulas.

The following lemma shows that $^-$ and $^\sim$ are order-reversing inverses of each other, and that $\circ$ is downward persistent in both arguments. 

\begin{lemma}\label{lem:twiddle_minus_involutive}
In a DInFL-frame $\mathbb{W}= \left(W, I,\preccurlyeq, \RR, {^{\sim}}, {^{-}}\right)$ the following hold for all $x,y,z,u,v\in W$:
\begin{enumerate}[label=\textup{(\roman*)}]
\item $x^{\sim -} =x = x^{-\sim}$. 
\item $x \preccurlyeq y\implies y^- \preccurlyeq x^-$ and 
$y^\sim \preccurlyeq x^\sim$. 
\item $x \preccurlyeq y$ and $\R{y}{w}{z}\implies \R{x}{w}{z}$. 
\item $x \preccurlyeq y$ and $\R{w}{y}{z}\implies \R{w}{x}{z}$. 
\end{enumerate}
\end{lemma}

\begin{proof}
(i) We only have to prove that $x \preccurlyeq x^{\sim -}$ and $x \preccurlyeq x^{-\sim}$. 
We have $x^\sim \preccurlyeq x^{\sim}$, so by Definition \ref{def:DInFL-frames}(1) there is some $i \in I$ such that $i \in I$ and 
$\R{i}{x^\sim}{x^\sim}$.
An application of (5) gives $\R{x}{i}{x^{\sim -}}$, and using (1) again, $x \preccurlyeq x^{\sim -}$. The second equality follows similarly.

(ii) Assume $x \preccurlyeq y$. Then by 
(1) there is some $i \in I$ and $\R{i}{x}{y}$. The latter is equivalent to $\R{i}{x}{y^{-\sim}}$ by (i). Applying (5) to this gives $\R{y^-}{i}{x^-}$. Therefore, using (1) again, $y^- \preccurlyeq x^-$.  
The proof of the second inequality is similar.

For (iii), assume $x \preccurlyeq y$ and $\R{y}{u}{v}$. Applying 
(ii) to the first part gives $y^\sim \preccurlyeq x^\sim$. The second part is equivalent to $\R{y}{u}{v^{\sim -}}$ by (i), which means $\R{u}{v^\sim}{y^\sim}$ by (5). 
Hence, since $y^\sim \preccurlyeq x^\sim$ and $\R{u}{v^\sim}{y^\sim}$, by (3) 
$\R{u}{v^\sim}{x^\sim}$. This gives $\R{x}{u}{v^{\sim -}}$, which is equivalent to $\R{x}{u}{v}$. The proof of (iv) is similar.
\end{proof}

From any DInFL-frame, it is possible to construct a complete perfect DInFL-algebra. 

\begin{proposition}\label{prop:complex_algebra_DInFL-frame}
Let $\mathbb{W}= \left(W, I,\preccurlyeq, \RR, {^{\sim}}, {^{-}}\right)$ be a DInFL-frame and let $\mathsf{Up}\left(W, \preccurlyeq\right)$ be  the set of all upsets
of $\left(W,\preccurlyeq\right)$.  For all $U \in \mathsf{Up}\left(W, \preccurlyeq\right)$, define
${\sim} U  = \left\{w \in W \mid w^- \notin U\right\}$ and 
$-U  = \left\{w \in W \mid w^{\sim} \notin U\right\}$.
Then $\mathbb{W}^+ = \left( \mathsf{Up}\left(W, \preccurlyeq\right), \cap, \cup, \circ, I, \sim, -\right)$ is a DInFL-algebra. 
\end{proposition}

\begin{proof}
We first show that $\mathsf{Up}\left(W, \preccurlyeq\right)$ is closed under the above operations. It is well-known (and easy to check) that $\mathsf{Up}\left(W, \preccurlyeq\right)$ is closed under taking unions and intersections. It follows that $I \in \mathsf{Up}\left(W, \preccurlyeq\right)$ by (2). Now let $x \in U\circ V$ and $x \preccurlyeq y$. Since $x \in U\circ V$, there are $u \in U$ and $v \in V$ such that $\R{u}{v}{x}$. Hence, by (3) 
we get $\R{u}{v}{y}$, and so $y \in U\circ V$. This shows that $U\circ V \in \mathsf{Up}\left(W, \preccurlyeq\right)$.

Next let $x \in {\sim}U$ and  $x \preccurlyeq y$. Since $x \in {\sim}U$, we have $x^- \notin U$. Applying Lemma~\ref{lem:twiddle_minus_involutive}(ii)
to $x\preccurlyeq y$, we obtain $y^- \preccurlyeq x^-$. Hence, since $U \in \mathsf{Up}\left(W, \preccurlyeq\right)$, we have $y^- \notin U$. This gives $y \in {\sim}U$. Thus, ${\sim}U \in \mathsf{Up}\left(W, \preccurlyeq\right)$ for all $U \in \mathsf{Up}\left(W, \preccurlyeq\right)$. 
Using Lemma~\ref{lem:twiddle_minus_involutive}(ii)
we can show in a similar way that $-U \in \mathsf{Up}\left(W, \preccurlyeq\right)$ for all $U \in \mathsf{Up}\left(W, \preccurlyeq\right)$.

Next we show that $I\circ U = U \circ I = U$ for all $U \in \mathsf{Up}\left(W, \preccurlyeq\right)$.
Let $x \in I \circ U$. Then there is some $i \in I$ and $u \in U$ such that $\R{i}{u}{x}$. Hence, by (1), 
$u \preccurlyeq x$. But $U \in \mathsf{Up}\left(W, \preccurlyeq\right)$, so $x \in U$. 

Now let $x \in U$. We have $x \preccurlyeq x$, so by (1) there is some $i \in I$ such that $\R{i}{x}{x}$. Hence, $x \in I\circ U$. 
Using the other equivalence of (1) 
we can show in a similar way that $U\circ I = U$ for all $U \in \mathsf{Up}\left(W, \preccurlyeq\right)$. 

The associativity of $\circ$ follows from (4) in Definition \ref{def:DInFL-frames}. 

Finally, we prove that $T \circ U \subseteq V$ iff $U \subseteq {\sim}\left(-V\circ T\right)$ iff $T \subseteq -\left(U\circ {\sim}V\right)$ for all upsets $T, U, V$ of $\left(W, \preccurlyeq\right)$. First, assume that $T \circ U \subseteq V$. Let $u \in U$ and suppose for the sake of a contradiction that $u \notin {\sim}\left(-V\circ T\right)$. The latter gives $u^{-}\in \left(-V\circ T\right)$. Hence, there exist $v\in -V$ and $t \in T$ such that $\R{v}{t}{u^{-}}$, and so by (5), $\R{t}{u}{v^{\sim}}$.  Since $t \in T$ and $u \in U$, we have $v^{\sim} \in T \circ U$, which means $v^{\sim} \in V$. This gives $v \notin -V$, a contradiction. 

Conversely, assume $U \subseteq {\sim}\left(-V\circ T\right)$. 
Let $w \in T\circ U$. Then there exist $t \in T$ and $u \in U$ such that $\R{t}{u}{w}$. Since $u \in U$, we have $u \in {\sim}\left(-V\circ T\right)$. Hence, $u^{-} \notin -V\circ T$. This implies that for all $x, y \in W$, if $y \in T$ and $\R{x}{y}{u^{-}}$ then $x \notin -V$. Now $\R{t}{u}{w}$ is equivalent to $\R{t}{u}{w^{-\sim}}$ by Lemma \ref{lem:twiddle_minus_involutive}, which is equivalent to $\R{w^-}{t}{u^-}$ by (5). Since $t \in T$, we have $w^- \notin -V$. Hence, $w^{-\sim} \in V$, and so $w \in V$. This shows that $T \circ U \subseteq V$. 

Now assume  $T \circ U \subseteq V$ and let $t \in T$. Suppose towards contradiction that $t \notin -\left(U\circ {\sim}V\right)$. Then $t^{\sim} \in U \circ {\sim}V$. Hence, there exists $u \in U$ and $v \in {\sim}V$ such that $\R{u}{v}{t^{\sim}}$. We thus have $t\in T$, $u \in U$ and $\R{t}{u}{v^{-}}$. This gives $v^{-} \in T\circ U$, and so $v^{-}\in V$. Thus, $v \notin {\sim}V$, a contradiction. 

Conversely, assume $T \subseteq -\left(U\circ {\sim}V\right)$ and let $w \in T\circ U$. This implies there are $t \in T$ and $u \in U$ such that $\R{t}{u}{w}$. Since $t \in T$, we get $t \in -\left(U\circ {\sim}V\right)$, and so $t^{\sim} \notin U\circ {\sim}V$. Hence, for all $x , y \in W$, if $x \in U$ and $\R{x}{y}{t^{\sim}}$, then $y \notin {\sim}V$. Now $\R{t}{u}{w}$ is equivalent to $\R{t}{u}{w^{\sim -}}$, which in turn is equivalent to $\R{u}{w^{\sim}}{t^{\sim}}$. Therefore, since $u \in U$, we get $w^{\sim} \notin {\sim}V$. This gives $w^{\sim -} \in V$, thus $w \in V$. 
\end{proof}

The algebra $\mathbb{W}^+$ in Proposition \ref{prop:complex_algebra_DInFL-frame} is called the $\emph{complex algebra}$ of $\mathbb{W}$. It is easy to check that $\mathbb{W}^+$ is complete and perfect. 

Recall that if $\mathbf{A}$ is a complete perfect distributive lattice,   
then the  map $\kappa : J^{\infty}\left(\mathbf{A}\right) \to M^{\infty}\left(\mathbf{A}\right)$ defined by $\kappa\left(j\right) = \bigvee\left\{a \in A \mid j \not\leqslant a\right\}$ is an order isomorphism (cf.~\cite[Section II.5]{FJN1995}). 
We now define two useful unary operations on the completely join-irreducibles of a complete perfect distributive InFL-algebra. 

\begin{definition}\label{def:superscrip_minus/tilde}
For every completely join-irreducible $a$ of a complete perfect DInFL-algebra $\mathbf{A}$, define $a^{\sim} = {\sim}\kappa\left(a\right)$ and $a^- = -\kappa\left(a\right)$. 
\end{definition}

Using the above definitions, we obtain the following important lemma. 

\begin{lemma}\label{lem:minus/tilde_completely_join_irreducible}
Let $\mathbf{A}=\left(A,\wedge, \vee, \cdot, 1, \sim, -\right)$ be a complete  perfect DInFL-algebra. If $a$ is a completely join-irreducible, then so are $a^{\sim}$ and $a^{-}$. 
\end{lemma}
\begin{proof}
Let $a$ be a completely join-irreducible element of a  complete perfect DInFL-algebra $\mathbf{A}$. By the definition of $\kappa$, we have $\kappa(a) \in M^\infty(\mathbf{A})$. By~\cite[Section 2]{GJ2013}, in an InFL-algebra, both $\sim$ and $-$ are dual lattice isomorphisms and hence ${\sim}\kappa(a) \in J^\infty(\mathbf{A})$ and $-\kappa(a) \in J^\infty(\mathbf{A})$. 
\end{proof}

The following proposition adapts a well-known method (see \cite{DGP2005}) for obtaining dual frames from complete perfect algebras:

\begin{proposition}\label{prop:DInFL-frame_from_perfect_InFL-algebra}
Let $\mathbf{A}=\left(A,\wedge, \vee, \cdot, 1, \sim, -\right)$ be a complete perfect DInFL-algebra. 
Set $I_1=\left\{i \in J^\infty(\mathbf{A}) \mid i \leqslant 1\right\}$ and, for all $a, b, c$ in $J^\infty(\mathbf{A})$, define $\preccurlyeq {=} \geqslant$, $\R{a}{b}{c}$ iff $c \leqslant a\cdot b$, $a^{\sim} = {\sim}\kappa\left(a\right)$ and $a^- = -\kappa\left(a\right)$.
Then the structure $\mathbf{A}_{+} =\left(J^{\infty}\left(\mathbf A\right), I_1, \preccurlyeq, \circ, {^\sim}, {^-}\right)$ is a DInFL-frame. 
\end{proposition}

\begin{proof}
(1): Let $a, b \in J^{\infty}\left(\mathbf{A}\right)$ and assume $a \preccurlyeq b$. Then $b \leqslant a$. Since $\mathbf{A}$ is completely join-generated by $J^{\infty}\left(\mathbf{A}\right)$, we have $$b \leqslant 1 \cdot a  = \bigvee\left\{j \in J^{\infty}\left(\mathbf{A}\right) \mid j \leqslant 1\right\} \cdot a 
 = \bigvee\left\{j \cdot a \mid j \in J^{\infty}\left(\mathbf{A}\right) \textnormal{ and } j \leqslant 1\right\}.$$
Now $b$ is completely join-prime, so we have $b \leqslant i\cdot a$ for some $i \in J^\infty\left(\mathbf{A}\right)$ such that $i \leqslant 1$. Hence, there is some $i \in J^\infty\left(\mathbf{A}\right)$ such that $i \in I_1$ and $\R{i}{a}{b}$.  Consequently, $b \in I_1\circ a$. 

Conversely, let $b \in I_1\circ a$. Then there is some $i \in J^\infty\left(\mathbf{A}\right)$ such that $i \in I_1$ and $\R{i}{a}{b}$. Hence,  $i \leqslant 1$ and $b \leqslant i\cdot a$. From $i \leqslant 1$, we obtain $i\cdot a \leqslant 1\cdot a = a$, and so $b \leqslant i \cdot a \leqslant  a$. This shows that $a \preccurlyeq b$. 
The proof of the other equivalence in (1) is similar. 

(2): Let $i \in I_1$ and $j \in J^{\infty}\left(\mathbf{A}\right)$, and assume $i \preccurlyeq j$. Then $i \leqslant 1$ and $j \leqslant i$. Hence, $j \leqslant 1$, and so $j \in I_1$.

(3): Assume $a \preccurlyeq b$ and $\R{c}{d}{a}$. Then $b \leqslant a$ and $a \leqslant c\cdot d$, and so $b \leqslant c\cdot d$, which shows that $\R{c}{d}{b}$. 

(4): This follows from the associativity of the monoid operation. 

(5): 
First assume $\R{a}{b}{c^{\sim}}$. Then $c^\sim \leqslant a \cdot b$, and so  ${\sim}\kappa(c)
\leqslant a \cdot b$. Since $-$ is order-reversing and $-$ and $\sim$ are inverses of each other, we have  $-\left(a\cdot b\right) \leqslant \kappa(c)$.
To prove that $\R{c}{a}{b^{-}}$, we have to show that $b^-\leqslant c\cdot a$, i.e., $-\kappa(b)
\leqslant c \cdot a$. This is equivalent to showing that ${\sim}\left(c\cdot a\right) \leqslant \kappa(b)$.
Hence, 
it suffices to 
show that $b \not\leqslant {\sim}\left(c \cdot a\right)$.
Suppose for the sake of a contradiction that $b \leqslant {\sim}\left(c\cdot a\right)$. Then $b \leqslant {\sim}\left(-{\sim}c \cdot a\right)$, and so $a\cdot b \leqslant {\sim}c$ by ($\mathsf{Res}$) in Section~\ref{sec:background}, which means $c \leqslant -\left(a \cdot b\right)$. 
We thus obtain $c \leqslant \kappa(c)$.
Therefore, since $c$ is completely join-prime, there is some $s \in A$ such that $c \leqslant s$ and  $c \not\leqslant s$, a contradiction. 
The converse implication can be proved in a similar way. 

(6): 
We will show $a^{\sim -} \preccurlyeq a$ for all $a \in J^{\infty}\left(\mathbf{A}\right)$ and leave the other case for the reader.  Proving that $a^{\sim -} \preccurlyeq a$ is equivalent to showing that $a \leqslant a^{-\sim}$, which is equivalent to $a \leqslant {\sim}\kappa\left(a^-\right)$.
Since $-$ is order-reversing and $\sim$ and $-$ are inverses of each other, this is equivalent to showing that $\kappa\left(a^{-}\right) 
\leqslant -a$. That is, we have to show that $b \leqslant -a$ for all $b \in A$ such that $a^- \not\leqslant b$. Let $b$ be an arbitrary element of $A$ such that $a^{-} \not\leqslant b$. For the sake of a contradiction, suppose $b \not\leqslant -a$. Then $a = {\sim}{-}a \not\leqslant {\sim} b$.  Since  $a^- \not\leqslant b$, we have ${\sim}b \not\leqslant \kappa\left(a\right)$.
Hence, ${\sim}b \not\leqslant c$ for all $c \in A$ such that $a \not\leqslant c$. But $a \not\leqslant {\sim}b$, so in particular, ${\sim}b\not\leqslant {\sim} b$, a contradiction.
\end{proof}

The next theorem shows
that every complete perfect DInFL-algebra is isomorphic to the complex algebra of its DInFL-frame of completely join-irreducibles.

\begin{theorem}\label{thm:DInFl-algebras_duality}
If $\mathbf{A} = \left(A,\wedge, \vee, \cdot, 1, \sim, -\right)$ is a complete perfect DInFL-algebra, then $\mathbf{A} \cong  \left(\mathbf{A}_+\right)^+$. 
\end{theorem}

\begin{proof}
The fact that the map $\psi : A \to \mathsf{Up}\left(J^{\infty}\left(\mathbf A\right), \preccurlyeq\right)$ defined by $\psi\left(a\right) = \left\{j \in J^{\infty}\left(\mathbf{A}\right) \mid j \leqslant a\right\}$ is 
a lattice isomorphism is well known (cf.~\cite{GNV2005}). 

First, we note that $\psi\left(1\right) = \left\{j \in J^{\infty}\left(\mathbf{A}\right) \mid j \leqslant 1\right\} = I_1$. 
To show that $\psi$ preserves the monoid operation, first let $i \in \psi\left(a\cdot b\right)$. Then $i \in J^{\infty}\left(\mathbf{A}\right)$ and $i \leqslant a \cdot b$. Since $\mathbf{A}$ is completely join-generated by $J^{\infty}\left(\mathbf{A}\right)$, we have 
\begin{align*}
i \leqslant a\cdot b & = \bigvee\left\{j \in J^{\infty}\left(\mathbf{A}\right) \mid j \leqslant a \right\} \cdot \bigvee\left\{k \in J^{\infty}\left(\mathbf{A}\right) \mid k \leqslant b\right\}\\
& = \bigvee\left\{j \cdot k \mid j, k \in J^{\infty}\left(\mathbf{A}\right), j \leqslant a  \textnormal{ and } k \leqslant b\right\}.
\end{align*}
Since $i$ is completely join-prime, it follows that $i \leqslant j \cdot k$ for some $j, k \in J^{\infty}\left(\mathbf{A}\right)$ such that $j \leqslant a$ and $k \leqslant b$. Hence, $j \in \psi\left(a\right)$, $k \in \psi\left(b\right)$ and $\R{j}{k}{i}$, which shows that $i \in \psi\left(a\right) \circ \psi\left(b\right)$. 

For the other inclusion, let $i \in J^\infty(\mathbf{A})$ such that  $i \in \psi\left(a\right) \circ \psi\left(b\right)$. Then there are $j \in \psi\left(a\right)$ and $k \in \psi\left(b\right)$
 such that $\R{j}{k}{i}$. Hence, $j \leqslant a$, $k \leqslant b$ and $i \leqslant j\cdot k$. From the first part we get $j\cdot k \leqslant a\cdot k$ and from the second part we get $a \cdot k \leqslant a \cdot b$. It follows that $i \leqslant j\cdot k \leqslant a\cdot k \leqslant a\cdot b$, and so $i \in \psi\left(a \cdot b\right)$. 

 Next we show that $\psi\left({\sim}a\right) = {\sim}\psi\left(a\right)$. First, let $j \in \psi\left({\sim}a\right)$. Then $j \in J^{\infty}\left(\mathbf{A}\right)$ and $j \leqslant {\sim}a$. For the sake of a contradiction, suppose $j \notin {\sim}\psi\left(a\right)$. Then we have $j^-\in \psi\left(a\right)$, and so $j^-\leqslant a$. Since $j \leqslant {\sim}a$ and $-$ is order-reversing, we obtain $a = -{\sim}a \leqslant - j$. Hence, $-\kappa\left(j\right) = j^- \leqslant -j$, which means $j \leqslant \kappa\left(j\right)$. Since $j$ is completely join-prime, there is some $s \in A$ such that $j \leqslant s$ and $j \not\leqslant s$, a contradiction. It must therefore be the case that $j \in {\sim}\psi\left(a\right)$.

For the other inclusion, let $j \in {\sim}\psi\left(a\right)$. Then $j^- \notin \psi\left(a\right)$, and so $j^- = -\kappa\left(j\right) \not\leqslant a$, which gives ${\sim}a \not\leqslant \kappa\left(j\right)$. Hence, ${\sim}a\not\leqslant s$ for all $s \in A$ such that $j \not\leqslant s$. Therefore it must be the case that $j \leqslant {\sim}a$; for otherwise, ${\sim}a \not\leqslant {\sim}a$, a contradiction. 
Thus $j \in \psi\left({\sim}a\right)$. 
The proof that $\psi\left(-a\right) = -\psi\left(a\right)$ is similar. 
\end{proof}

As mentioned after Proposition~\ref{prop:complex_algebra_DInFL-frame}, the complex algebra of a DInFL-frame $\mathbb{W} = \left(W, I, \preccurlyeq, \RR, {^\sim}, {^-}\right)$ is a complete perfect DInFL-algebra. 
We will define frame morphisms in Section~\ref{sec:frame-morphisms}. For the result below, the claimed isomorphic relationship between the two frames should be viewed simply in terms of the usual concept of isomorphism between first-order structures. 

\begin{theorem}\label{thm:DInFL-frames_duality}
If $\mathbb{W}= \left(W, I,\preccurlyeq, \RR, {^{\sim}}, {^{-}}\right)$ is a DInFL- frame, then  $\mathbb{W} \cong \left(\mathbb{W}^+\right)_+$.     
\end{theorem}

\begin{proof}
It is well known that $J^{\infty}\left(\mathbb{W}^+\right) = \left\{{\uparrow}x\mid x \in W\right\}$ where ${\uparrow}x = \left\{y \in W\mid x \preccurlyeq y\right\}$. Hence the standard map $x \mapsto {\uparrow}x$
gives us an order-isomorphism from $(W,\preccurlyeq) $ to $(J^{\infty}(\mathbb{W}^+),\supseteq)$. 

Since $I$ is an upset, we have  
 $x \in I$ if and only if ${\uparrow}x \subseteq I$. 
Now $I_I \subseteq J^{\infty}(\mathbb{W}^+)$ is defined by 
$I_I= \{ U \in \mathsf{Up}(W,\preccurlyeq) \mid U \in J^\infty (\mathbb{W}^+), U \subseteq I\}=\{ {\uparrow}x \mid x \in W, {\uparrow}x \subseteq I\}$. Hence $x \in I$ if and only if ${\uparrow}x \in I_I$. 

Next, we need to show, for any $x,y,z \in W$, that $z \in x \circ y$ if and only if ${\uparrow}z \subseteq {\uparrow}x \circ {\uparrow} y$. We have 
$${\uparrow}x\circ {\uparrow} y =\bigcup \{ u \circ v \mid u \in {\uparrow}x, v \in {\uparrow} y \}=\bigcup \{u \circ v \mid x \preccurlyeq u, y \preccurlyeq v \}.$$
Assume $z \in x \circ y$ and $z \preccurlyeq w$. By the reflexivity of $\preccurlyeq$ we get $x \circ y\subseteq {\uparrow}x \circ {\uparrow}y$. Hence $z$ is in the upset ${\uparrow}x \circ {\uparrow} y$ and so  $w \in {\uparrow} x \circ {\uparrow} y$. For the converse, assume ${\uparrow}z \subseteq {\uparrow}x \circ {\uparrow} y$. So there exist $u,v$ with $x \preccurlyeq u$ and $y \preccurlyeq v$ such that $z \in u \circ v$. By items (iii) and (iv) of 
Lemma~\ref{lem:twiddle_minus_involutive} we get $z\in x \circ y$. 

Finally we will show that  ${\uparrow}(x^{\sim})=({\uparrow}x)^{\sim}$. The proof that ${\uparrow}(x^{-})=({\uparrow}x)^{-}$ will follow similarly. Using the definition of $^{\sim}$ on $J^{\infty}(\mathbb{W}^+)$ used in Proposition~\ref{prop:DInFL-frame_from_perfect_InFL-algebra}, we get the following:
\begin{align*}
({\uparrow}x)^{\sim} = {\sim} \kappa ({\uparrow}x) &={\sim} \bigcup \{ U \in \mathsf{Up}(W,\preccurlyeq) \mid {\uparrow}x \not\subseteq U \} \\ 
&= \left\{ w \in W \mid w^- \notin \bigcup \{ U \in \mathsf{Up}(W,\preccurlyeq) \mid {\uparrow}x \not\subseteq U \} \right\} 
\\
&= \{w \in W \mid \text{for all } U \in \mathsf{Up}(W,\preccurlyeq),{\uparrow}x \not\subseteq U \Longrightarrow  w^- \notin U \} \\
&= \{w \in W \mid \text{for all } U \in \mathsf{Up}(W,\preccurlyeq), {\uparrow}w^- \subseteq U \Longrightarrow  {\uparrow}x \subseteq U  \} \\
&= \{w \in W \mid {\uparrow}x \subseteq {\uparrow}w^-  \} \\
&= \{w \in W \mid w^- \preccurlyeq x   \}\\
&= \{w \in W \mid x^{\sim} \preccurlyeq
w\} = {\uparrow}(x^{\sim}). \end{align*}
\end{proof}

We now equip DInFL-frames with additional structure to dually represent complete perfect DqRAs.

\begin{definition} \label{def:DqRA_frames}
A \emph{DqRA-frame} is a tuple $\mathbb{W}= \left(W, I,\preccurlyeq, \RR, {^{\sim}}, {^{-}}, {^\neg}\right)$ such that $\left(W, I,\preccurlyeq, \RR, {^{\sim}}, {^{-}}\right)$ is a DInFL-frame and 
the following additional conditions hold for all $x, y, z\in W$:
\begin{enumerate}[label=\textup{(\arabic*)}]
\stepcounter{enumi}
\stepcounter{enumi}
\stepcounter{enumi}
\stepcounter{enumi}
\stepcounter{enumi}
\stepcounter{enumi}
\item $x^{\neg\neg} = x$
\item $x \preccurlyeq y \implies y^\neg \preccurlyeq x^\neg$
\item $\R{x}{y}{z^-}  \Longleftrightarrow \R{y^{\sim\neg}}{x^{\sim\neg}}{z^\neg}$
\end{enumerate}
\end{definition}

To illustrate the concept of DqRA-frame, Table~\ref{tab:smallqRA} shows all frames for nontrivial algebras up to size four. The names of these frames are chosen to match the list of algebras in the first two rows of Figure 2 in Section~\ref{sec:list-of-algebras}. Hence the superscript shows the number of upward closed sets of the frame (rather than the number of elements of the frame). The operation $^-$ is not shown since it is the inverse of $^\sim$ (and for the frames in the table $x^-=x^\sim$).

Note that $\mathbb W_{1,1}^2$ is the frame of the 2-element Boolean algebra (as a DqRA), $\mathbb W_{1,1}^3$ and $\mathbb W_{1,2}^4$ are frames of MV-chains usually denoted by \L$_3$ and \L$_4$, while $\mathbb W_{1,2}^3$ and $\mathbb W_{1,3}^4$ are frames of the Sugihara chains $\m S_3$ and $\m S_4$. The DqRA-frames $\mathbb W_{2,1a}^4$ and $\mathbb W_{2,1b}^4$ are isomorphic when considered as DInFL-frames (in Figure 7.3 they correspond to the InFL-algebra $D_{2,1,2}^4$).
The empty frame $\mathbb W_{1,1}^1$, which corresponds to the one-element DqRA, is not shown in Table~\ref{tab:smallqRA}.

\begin{table}
\begin{center}
\small
\tikzstyle{every picture} = [scale=0.6]
\tikzstyle{every node} = [draw, fill=white, circle, inner sep=0pt, minimum size=5pt]
\begin{tikzpicture}[baseline=0pt]
\node(0) at (0,0)[label=right:$e$]{};
\end{tikzpicture}
\quad
$\begin{array}{|c|c||c|c|}\hline
\!\mathbb W_{1,1}^2\!\!&\ e&^\sim&^\neg\\\hline
e   &e	&e&e	\\\hline
\end{array}$
\qquad\quad
\begin{tikzpicture}[baseline=9pt]
\node(1) at (0,1)[label=right:$u$]{};
\node(0) at (0,0)[label=right:$e$]{} edge (1);
\end{tikzpicture}
\quad
$\begin{array}{|c|cc||c|c|}\hline
\!\mathbb W_{1,1}^3\!\!&\ u&\ e&^\sim&^\neg\\\hline
u   &\emptyset	&u	&e&e\\
e   &u&	eu&u	&u\\\hline
\end{array}$
\quad
$\begin{array}{|c|cc||c|c|}\hline
\!\mathbb W_{1,2}^3\!\!&e&\ u&^\sim&^\neg\\\hline
e&e&eu	&u&u	\\
u&eu&eu &e&e	\\\hline
\end{array}$\\[8pt]

\begin{tikzpicture}[baseline=15pt]
\node(2) at (0,2)[label=right:$v$]{};
\node(1) at (0,1)[label=right:$u$]{} edge (2);
\node(0) at (0,0)[label=right:$e$]{} edge (1);
\end{tikzpicture}
\quad
$\begin{array}{|c|ccc||c|c|}\hline
\!\mathbb W_{1,1}^4\!\!&v\ &u\ &e&^\sim&^\neg\\\hline
v&\emptyset	&\emptyset&v&e&e\\
u&\emptyset  &uv	&uv	&u&u\\
e&v&uv	&euv	&v&v\\\hline
\end{array}$
\quad
$\begin{array}{|c|ccc||c|c|}\hline
\!\mathbb W_{1,2}^4\!\!&v\ &u&\ e&^\sim&^\neg\\\hline
v&\emptyset	&\emptyset&v&e&e\\
u&\emptyset  &v	&uv	&u&u\\
e&v&uv	&euv	&v&v\\\hline
\end{array}$\\[8pt]

\begin{tikzpicture}[baseline=16pt]
\node(2) at (0,2)[label=right:$v$]{};
\node(1) at (0,1)[label=right:$e$]{} edge (2);
\node(0) at (0,0)[label=right:$u$]{} edge (1);
\end{tikzpicture}
\quad
$\begin{array}{|c|ccc||c|c|}\hline
\!\mathbb W_{1,3}^4\!\!&v\ &e&\ u&^\sim&^\neg\\\hline
v&v	&v	&euv&u&u	\\
e&v  &ev&euv&e&e	\\
u&euv&euv&euv&v&v\\\hline
\end{array}$
\qquad\quad
\begin{tikzpicture}[baseline=16pt]
\node(2) at (0,2)[label=right:$e$]{};
\node(1) at (0,1)[label=right:$v$]{} edge (2);
\node(0) at (0,0)[label=right:$u$]{} edge (1);
\end{tikzpicture}
\quad
$\begin{array}{|c|ccc||c|c|}\hline
\!\mathbb W_{1,4}^4\!\!&e\ &v\ &u&^\sim&^\neg\\\hline
e&e&ev&euv &u&u\\
v&ev&euv&euv&v&v\\
u&euv&euv&euv&e&e\\\hline
\end{array}$\\[8pt]

\!\!\!\!\begin{tikzpicture}[baseline=0pt]
\node(2) at (0,0)[label=below:$e$]{};
\node(0) at (0.7,0)[label=below:$f$]{};
\end{tikzpicture}
\quad
$\begin{array}{|c|cc||c|c|}\hline
\!\mathbb W_{2,1a}^4\!\!&e&f&^\sim&^\neg\\\hline
e	&e&\emptyset	&e&e	\\
f   &\emptyset&f &f&f	\\\hline
\end{array}$
\quad
$\begin{array}{|c|cc||c|c|}\hline
\!\mathbb W_{2,1b}^4\!\!&e&f&^\sim&^\neg\\\hline
e	&e&\emptyset	&e&f	\\
f   &\emptyset&f &f&e	\\\hline
\end{array}$\\[8pt]

\begin{tikzpicture}[baseline=0pt]
\node(2) at (0,0)[label=below:$e$]{};
\node(0) at (0.7,0)[label=below:$u$]{};
\end{tikzpicture}
\ 
$\begin{array}{|c|cc||c|c|}\hline
\!\mathbb W_{2,2}^4\!\!&e&u&^\sim&^\neg\\\hline
e	&e	&u	&e&e	\\
u   &u&e	&u&u	\\\hline
\end{array}$
\quad 
$\begin{array}{|c|cc||c|c|}\hline
\!\mathbb W_{2,3}^4\!\!&e&u&^\sim&^\neg\\\hline
e	&e	&u	&e&e	\\
u   &u&eu	&u&u	\\\hline
\end{array}$
\quad
$\begin{array}{|c|cc||c|c|}\hline
\!\mathbb W_{3,1}^4\!\!&e&u&^\sim&^\neg\\\hline
e	&e	&u	&u&u	\\
u   &u&e	&e&e	\\\hline
\end{array}$
\quad
$\begin{array}{|c|cc||c|c|}\hline
\!\mathbb W_{3,2}^4\!\!&e&u&^\sim&^\neg\\\hline
e	&e	&u	&u&u	\\
u   &u  &\emptyset 	&e&e	\\\hline
\end{array}$
\end{center}

\caption{DqRA-frames for distributive quasi relation algebras up to cardinality four. The identity set $I$ is either $\{e\}$ or $\{e,f\}$, and in the table for the binary $\circ$, a single element $u$ or a sequence of elements such as $euv$ denotes the set $\{u\}$ or $\{e,u,v\}$ respectively. The corresponding DqRAs are shown in the first two rows of Figure~2.}
\label{tab:smallqRA}

\vspace{-0.6cm}
\end{table}

As in the case of qRAs (see Lemma~\ref{lem:Di-in-qRA}), we can deduce that $x^{\sim\neg} = x^{\neg -}$ for all $x \in W$. 

\begin{lemma}\label{lem:proving_Di_frames}
Let $\mathbb{W}= \left(W, I,\preccurlyeq, \RR, {^{\sim}}, {^{-}}, {^\neg}\right)$ be a DqRA-frame. Then $x^{\sim\neg} = x^{\neg -}$ for all $x \in W$. 
\end{lemma}

\begin{proof}
We first show that $x^{\sim \neg} \preccurlyeq x^{\neg -}$. Since $x^{-\sim} \preccurlyeq x^{-\sim}$, there exists $i \in W$ such that $i \in I$ and $\R{x^{-\sim}}{i}{x^{-\sim}}$ by (1). Applying (5) to the second part gives $\R{x^-}{x^{-\sim}}{i^-}$, and so, by (9), $\R{x^{-\sim\sim\neg}}{x^{-\sim \neg}}{i^\neg}$. Hence, using Lemma~\ref{lem:twiddle_minus_involutive}(i), we obtain $\R{x^{\sim\neg}}{x^{\neg}}{i^{\neg - \sim}}$. Another application of (5) yields $\R{i^{\neg -}}{x^{\sim \neg}}{x^{\neg -}}$. Therefore, by (9), $\R{x^{\sim\neg\sim\neg}}{i^{\neg -\sim\neg}}{x^{\neg\neg}}$, and so, by (7) and Lemma~\ref{lem:twiddle_minus_involutive}(i),  $\R{x^{\sim\neg\sim\neg}}{i}{x}$. Since $i \in I$, we can apply (1) again to get $x^{\sim\neg\sim\neg} \preccurlyeq x$. Hence, by (7) and (8), $x^\neg \preccurlyeq x^{\sim\neg \sim\neg\neg} = x^{\sim \neg\sim}$. Using Lemma~\ref{lem:twiddle_minus_involutive}(i) and (ii) we obtain $x^{\sim\neg} = x^{\sim\neg\sim -} \preccurlyeq x^{\neg -}$. 

Now $x \preccurlyeq x^{\neg\neg}$ by (7), so there is some $i \in W$ such that $i \in I$ and $\R{x}{i}{x^{\neg\neg}}$. Hence, by (7) and Lemma~\ref{lem:twiddle_minus_involutive}, we have $\R{x^{\neg -\sim\neg}}{i^{\neg -\sim \neg}}{x^{\neg\neg}}$. Applying (9) to this gives $\R{i^{\neg -}}{x^{\neg -}}{x^{\neg -}}$, and so $\R{x^{\neg -}}{x^\neg}{i^{\neg - \sim}}$ by (5). Thus, by (7) and Lemma~\ref{lem:twiddle_minus_involutive}, $\R{x^{\neg - \neg -\sim\neg}}{x^{\neg\neg - \sim \neg}}{i^{\neg}}$. Applying (9) again gives $\R{x^{\neg\neg -}}{x^{\neg -\neg -}}{i^{-}}$, and therefore, $\R{x^-}{x^{\neg -\neg -}}{i^-}$. By (5), $\R{x^{\neg -\neg -}}{i}{x^{-\sim}}$, which means $\R{x^{\neg -\neg -}}{i}{x}$. We thus get $x^{\neg -\neg -} \preccurlyeq x$.  Hence, $x^\sim \preccurlyeq x^{\neg -\neg -\sim} = x^{\neg - \neg}$, and so $x^{\neg -} = x^{\neg -\neg \neg} \preccurlyeq x^{\sim\neg}$. 
\end{proof}

The following lemma shows that $x^{\sim\neg} = x^{\neg -}$ is equivalent to $x^{-\neg} = x^{\neg\sim}$, as expected. 

\begin{lemma}\label{lem:equivalent_Di}
Let $\mathbb{W}= \left(W, I,\preccurlyeq, \RR, {^{\sim}}, {^{-}}, {^\neg}\right)$ be a DqRA-frame. Then $x^{\sim\neg} = x^{\neg -}$ for all $x \in W$ iff $x^{-\neg} = x^{\neg\sim}$ for all $x \in W$.
\end{lemma}
\begin{proof} 
Assume $x^{\sim\neg} = x^{\neg -}$ for all $x \in W$. Then 
$(x^-)^{\sim\neg}=(x^-)^{\neg -}$. Applying 
Lemma~\ref{lem:twiddle_minus_involutive}(i) gives us $x^\neg = (x^-)^{\neg -}$. Hence $x^{\neg\sim}=(x^-)^{\neg -\sim}$ and applying Lemma~\ref{lem:twiddle_minus_involutive}(i) again gives $x^{\neg\sim}=x^{-\neg}$. The proof of the other implication is similar.
\end{proof}

The above lemma allows us to prove the following result. 

\begin{proposition}\label{prop:complex_algebra_DqRA-frame}
Let $\mathbb{W}= \left( W, I, \preccurlyeq, \RR, {^\sim}, {^-}, {^\neg}\right)$ be a DqRA-frame.  For all $U, V \in \mathsf{Up}\left(W, \preccurlyeq\right)$, define $\circ$, $\sim$ and $-$ as in Proposition \ref{prop:complex_algebra_DInFL-frame} and $\neg$ by
$\neg U  = \left\{x \in W \mid x^\neg \notin U\right\}$. 
Then the structure $\mathbb{W}^+ = \left( \mathsf{Up}\left(W, \preccurlyeq\right), \cap, \cup, \circ, I, \sim, -, \neg\right)$ is a DqRA.
\end{proposition}

\begin{proof}
Since $\left( \mathsf{Up}\left(W, \preccurlyeq\right), \cap, \cup, \circ, I, \sim, -\right)$ is a DInFL-algebra (Proposition~\ref{prop:complex_algebra_DInFL-frame}), it will follow that $\mathbb{W}^+$ is a 
DqRA
if we can show that $\neg U \in \mathsf{Up}\left(W, \preccurlyeq\right)$ for all $U \in \mathsf{Up}\left(W, \preccurlyeq\right)$, $\neg$ is involutive, and $\mathsf{(Dm)}$ and $\mathsf{(Dp)}$ hold.  
Using the fact that $^\neg$ is order-reversing, we can show that if $U \in \mathsf{Up}\left(W, \preccurlyeq\right)$, then $\neg U \in \mathsf{Up}\left(W, \preccurlyeq\right)$. 

Next we show that $\neg$ is involutive. Let $U \in \mathsf{Up}\left(W, \preccurlyeq\right)$. Then $x \in \neg\neg U$ iff $x^\neg \notin \neg U$ iff $x^{\neg\neg} \in U$ iff $x \in U$. Here the last equivalence follows from (7). 

For $\mathsf{(Dm)}$, let $U, V \in \mathsf{Up}\left(W, \preccurlyeq\right)$. Then $x \in \neg\left(U \cap V\right)$ iff $x^\neg \notin U \cap V$ iff $x^\neg \notin U$ or $x^{\neg} \notin V$ iff $x \in \neg U$ or $x \in \neg V$ iff $x \in \neg U \cup \neg V$. 

Finally, for $\mathsf{(Dp)}$, let $U, V \in \mathsf{Up}\left(W, \preccurlyeq\right)$. We first show that $\neg\left(U \circ V\right) \subseteq {\sim}\left(-\neg V \circ -\neg U\right)$. Suppose  $x \notin {\sim}\left(-\neg V \circ -\neg U\right)$. Then we have $x^- \in -\neg V \circ -\neg U$. This means there exist $v \in -\neg V$ and $u \in -\neg U$ such that $\R{v}{u}{x^{-}}$. Applying condition (9) to $\R{v}{u}{x^{-}}$ gives $\R{u^{\sim\neg}}{v^{\sim\neg}}{x^\neg}$. From $v \in -\neg V$ and $u \in -\neg U$ we get $v^{\sim\neg} \in V$ and $u^{\sim\neg} \in U$. Hence, $x^\neg \in U\circ V$, and so $x \notin \neg\left(U\circ V\right)$. 

For the other inclusion, suppose $x \notin \neg\left(U\circ V\right)$. Then $x^\neg \in U\circ V$. This means there exist $u \in U$ and $v \in V$ such that $\R{u}{v}{x^\neg}$. Now $u = u^{\sim - \neg \neg}$ by (7) and Lemma~\ref{lem:twiddle_minus_involutive}. Furthermore, by  Lemma~\ref{lem:equivalent_Di}, $u^{\sim - \neg \neg} = u^{\sim\neg\sim\neg}$, and hence  $u = u^{\sim\neg\sim\neg}$. Likewise, $v = v^{\sim - \neg \neg} = v^{\sim\neg\sim\neg}$. We thus have  $u^{\sim\neg\sim\neg} \in U$, $v^{\sim\neg\sim\neg} \in V$ and $\R{u^{\sim\neg\sim\neg}}{v^{\sim\neg\sim\neg}}{x^\neg}$. Applying (9) to this gives $\R{v^{\sim\neg}}{u^{\sim\neg}}{x^-}$. From $u^{\sim\neg\sim\neg} \in U$ and $v^{\sim\neg\sim\neg} \in V$ we get $u^{\sim\neg}\in -\neg U$ and $v^{\sim\neg}\in -\neg V$. It thus follows that $x^- \in -\neg V \circ -\neg U$, and so $x \notin {\sim}\left(-\neg V \circ -\neg U\right)$. 
\end{proof}

\begin{definition}
Let $\mathbf{A}=\left(A,\wedge, \vee, \cdot, 1, \sim, -, \neg\right)$ be a complete perfect DqRA. For every completely join-irreducible $a$ of $\mathbf{A}$, define $a^{\neg} = {\neg}\kappa\left(a\right)$. 
\end{definition}

\begin{lemma}\label{lem:neg_completely_join-irreducible}
Let $\mathbf{A}=\left(A,\wedge, \vee, \cdot, 1, \sim, -, \neg\right)$ be a complete  perfect DqRA. If $a$ is a completely join-irreducible, then so is $a^{\neg}$. 
\end{lemma}

\begin{proposition}\label{prop:DqRA-frame_from_perfect_qRA}
Let $\mathbf{A}=(A,\wedge, \vee, \cdot, 1, \sim, -, \neg)$ be a complete perfect DqRA. 
Define $I_1$,
$\preccurlyeq$, $\RR$, $^\sim$, $^-$ as in Proposition~\ref{prop:DInFL-frame_from_perfect_InFL-algebra} and for all $a \in J^{\infty}\left(\mathbf A\right)$, define 
$a^{\neg} = \neg \kappa(a)$.
Then the structure $\mathbf{A}_+ = \left(J^{\infty}\left(\mathbf A\right), I_1, \preccurlyeq, \RR, {^\sim}, {^-}, {^\neg}\right)$ is a DqRA-frame. 
\end{proposition}

\begin{proof}
We only have to prove that $\mathbf{A}_+$ satisfies conditions (7) to (9) of Definition~\ref{def:DqRA_frames}.
We first show that $a^{\neg\neg} = a$. To prove that $a \preccurlyeq a^{\neg\neg}$, we have to show that $\neg\kappa\left(a^\neg\right)
\leqslant a$, i.e., $\neg a \leqslant \kappa\left(a^\neg\right)$.
If we can show that $a^\neg \not\leqslant \neg a$, we would be done, so suppose towards contradiction that $a^\neg \leqslant \neg a$. Then $\neg\kappa(a)
\leqslant \neg a$, and so $a \leqslant \kappa\left(a\right)$
Since $a$ is completely join-prime, there is some $s \in A$ such that $a \leqslant s$ and $a \not\leqslant s$, a contradiction. 

To prove that $a^{\neg\neg} \preccurlyeq a$, we have to show that $a \leqslant \neg \kappa\left(a^\neg\right)$,
which is equivalent to showing that 
$\kappa\left(a^\neg\right) \leqslant \neg a$. Let $b$ be an arbitrary element of $A$ such that $a^\neg \not\leqslant b$. Then $\neg\kappa(a)
\not\leqslant b$, and therefore $\neg b \not\leqslant \kappa(a)$
This means $\neg b \not\leqslant s$ for all $s \in A$ such that $a \not\leqslant s$. It follows that it must be the case that $b \leqslant \neg a$; for otherwise, $a \not\leqslant \neg b$, which means $\neg b\not\leqslant \neg b$, a contradiction. Since $b$ was an arbitrary element of $A$ we have $b \leqslant \neg a$ for all $b \in A$ such that 
$a^\neg \not\leqslant b$, and hence $\kappa\left(a^\neg\right)
\leqslant \neg a$, as required. 

To see that $^\neg$ is order-reversing, assume $a \preccurlyeq b$. Then we have $b \leqslant a$. We have to show that $a^\neg \leqslant b^{\neg}$; that is, we have to show that $\neg\kappa(a) \leqslant \neg\kappa(b)$. 
This is equivalent to showing that $\kappa(b) \leqslant \kappa(a)$.
Let $s$ be an arbitrary element of $A$ such that $b \not\leqslant s$. Then it must be the case that $a \not\leqslant s$; for otherwise, $b \leqslant a \leqslant s$, a contradiction. It follows that $s \leqslant \kappa(a)$.
Hence, since $s$ was arbitrary, $s \leqslant \kappa(a)$
for all $s \in A$ such that $b \not\leqslant s$. This proves that $\kappa(b) \leqslant \kappa(a)$, as required. 

Finally, we show that condition (9) holds 
for all $a, b, c \in J^{\infty}\left(\mathbf{A}\right)$ by the following chain of logical equivalences: 
\begin{align*}
\R{b^{\sim\neg}}{a^{\sim\neg}}{c^{\neg}} &\iff
c^\neg \leqslant b^{\sim\neg}\cdot a^{\sim\neg}\\
&\iff \neg \kappa(c)
\leqslant b^{\sim\neg}\cdot a^{\sim\neg}\\
&\iff \neg\left(b^{\sim\neg}\cdot a^{\sim\neg}\right) \leqslant \kappa(c)\\
&\iff
\neg\left(b^{\sim\neg}\right)+ \neg \left(a^{\sim\neg}\right) \leqslant \kappa(c)\quad\text{by }\mathsf{(Dp)}\\
&\iff 
\neg\neg \kappa\left(b^{\sim}\right) 
+ \neg\neg \kappa\left(a^{\sim}\right)
\leqslant \kappa(c)\\
&\iff
\kappa\left(b^{\sim}\right) + \kappa\left(a^{\sim}\right) 
\leqslant 
\kappa(c)\\
&\iff
{\sim}\left(-\kappa\left(a^\sim\right)
\cdot - \kappa\left(b^{\sim}\right)\right)
\leqslant \kappa(c)\\
&\iff
- \kappa\left(c\right)
\leqslant -\kappa\left(a^\sim\right)
\cdot -\kappa\left(b^\sim\right)\\
&\iff c^- \leqslant a^{\sim -}\cdot b^{\sim -}\\
&\iff \R{a}{b}{c^-}. 
\end{align*}
\end{proof}

\begin{theorem}\label{thm:DqRA-algebras_duality}
If $\mathbf{A}$
is a complete perfect DqRA, then $\mathbf{A} \cong  \left(\mathbf{A}_+\right)^+$. 
\end{theorem}
\begin{proof}
All that is left here to do is to show that the map $\psi: A \to \mathsf{Up}\left(J^{\infty}\left(\mathbf{A}\right), \preccurlyeq\right)$ defined by $\psi\left(a\right) = \left\{j \in J^\infty\left(\mathbf{A}\right) \mid j \leqslant a\right\}$ preserves $\neg$, and the proof of this is analogous to the proof in Theorem~\ref{thm:DInFl-algebras_duality} that $\psi\left({\sim}a\right) = {\sim}\psi\left(a\right)$. 
\end{proof}

The theorem below follows using the same map as in Theorem~\ref{thm:DInFL-frames_duality}. 

\begin{theorem}\label{thm:DqRA-frames_duality}
If $\mathbb{W}= \left(W, I,\preccurlyeq, \RR, {^{\sim}}, {^{-}}, {^\neg}\right)$ is a DqRA-frame, then $\mathbb{W} \cong \left(\mathbb{W}^+\right)_+$.     
\end{theorem}

\begin{proof}
All that needs to be checked is that ${\uparrow}(x^\neg)=(\uparrow x)^\neg$. This follows using an argument similar to that in Theorem~\ref{thm:DInFL-frames_duality} and  using (7) and (8) from Definition~\ref{def:DqRA_frames}. 
\end{proof}

\section{Morphisms for DInFL- and DqRA-frames}\label{sec:frame-morphisms}

In this section, we define morphisms for DInFL- and DqRA-frames and outline the duality between complete homomorphisms on complete perfect DInFL-algebras (resp. complete perfect DqRAs) and morphisms on DInFL-frames (resp. DqRA-frames).

\begin{definition}\label{def:DInFL-frame-morphism}
Let $\mathbb{W}_1= \left(W_1, I_1,\preccurlyeq_1, \RR_1, {^{\sim_1}}, {^{-_1}}\right)$
and 
$\mathbb{W}_2= \left(W_2, I_2,\preccurlyeq_2, \RR_2, {^{\sim_2}}, {^{-_2}}\right)$ be DInFL-frames. 
A function $f: W_1 \to W_2$ is a \emph{DInFL-frame morphism} if the following conditions hold: 
\begin{enumerate}
[label=\textup{(\arabic*)}]
\item $x \preccurlyeq_1 y \Longrightarrow f(x) \preccurlyeq_2 f(y)$ 
\item  $z \in x\circ_1 y \Longrightarrow f(z) \in f(x)\circ_2 f(y)$
\item $f(z) \in u \circ_2 v \Longrightarrow  (\exists x, y \in W_1)(u \preccurlyeq_2 f(x), v \preccurlyeq_2 f(y)$ and $z \in  x\circ_1 y)$
\item $f(x^{\sim_1})=f(x)^{\sim_2}$
\item $f(x^{-_1})=f(x)^{-_2}$
\item $I_1 = f^{-1}[I_2]=\{x \in W_1\mid f(x) \in I_2\} $
\end{enumerate}
\end{definition}

A DInFL-frame morphism from a DInFL-frame $\mathbb{W}_1$ to a DInFL-frame $\mathbb{W}_2$ gives rise to a complete homomorphism from the complex algebra $\mathbb{W}^+_2$ of $\mathbb{W}_2$ to the complex algebra $\mathbb{W}^+_1$ of $\mathbb{W}_1$. The lemma below  says that the preimage of an upset of $(W_2, \preccurlyeq_2)$ under this morphism is an upset of $(W_1, \preccurlyeq_1)$, suggesting that the homomorphism we are looking for is the map sending an upset of $(W_2, \preccurlyeq_2)$ to its preimage under this morphism. 

\begin{lemma}\label{inv(f)_upsets_to_upsets}
Let $\mathbb{W}_1= \left(W_1, I_1,\preccurlyeq_1, \RR_1, {^{\sim_1}}, {^{-_1}}\right)$
and 
$\mathbb{W}_2= \left(W_2, I_2,\preccurlyeq_2, \RR_2, {^{\sim_2}}, {^{-_2}}\right)$ be DInFL-frames. If $f: W_1 \to W_2$ is a DInFL-frame morphism and $U_2$ is an upset of $(W_2, \preccurlyeq_2)$, then the preimage $f^{-1}[U_2] = \{u_1 \in W_1\mid f(u_1) \in U_2\}$ is an upset of $(W_1, \preccurlyeq_1)$. 
\end{lemma} 

\begin{proof}
Let $u_1 \in f^{-1}[U_2]$  and  $u_1 \preccurlyeq_1 v_1$. The first part implies $f(u_1) \in U_2$, and the second part implies $f(u_1) \preccurlyeq_2 f(v_1)$ (by (1) of Definition~\ref{def:DInFL-frame-morphism}). Hence, since $U_2$ is an upset of $(W_2, \preccurlyeq_2)$, it follows that $f(v_1) \in U_2$, which means $v_1 \in f^{-1}[U_2]$. 
\end{proof}

\begin{definition}\label{def:dual_morphism_compl_alg}
Let $\mathbb{W}_1= \left(W_1, I_1,\preccurlyeq_1, \RR_1, {^{\sim_1}}, {^{-_1}}\right)$
and 
$\mathbb{W}_2= \left(W_2, I_2,\preccurlyeq_2, \RR_2, {^{\sim_2}}, {^{-_2}}\right)$ be DInFL-frames, and let $f: W_1 \to W_2$ be  a DInFL-frame morphism. The \emph{dual} of $f$ is the map $f^{+}: \mathsf{Up}\left(W_2, \preccurlyeq_2\right) \to \mathsf{Up}\left(W_1, \preccurlyeq_1\right)$ defined by $f^{+}(U_2)=f^{-1}[U_2]$, 
for all $U_2 \in \mathsf{Up}\left(W_2, \preccurlyeq_2\right)$.
\end{definition}

Recall that a map $f$ between posets $(P,\leqslant_P)$ and $(Q,\leqslant_Q)$ is an \emph{order-embedding} if $x \leqslant_P y$ if and only if $f(x) \leqslant_Q f(y)$. 

\begin{proposition}  
\label{prop:DInFL-frame-morphisms-duality}
Let $\mathbb{W}_1= \left(W_1, I_1,\preccurlyeq_1, \RR_1, {^{\sim_1}}, {^{-_1}}\right)$
and 
$\mathbb{W}_2= \left(W_2, I_2,\preccurlyeq_2, \RR_2, {^{\sim_2}}, {^{-_2}}\right)$ be DInFL-frames, and let $f: W_1\to W_2$ be a DInFL-frame morphism.
\begin{enumerate}
[label=\textup{(\roman*)}]
\item The map $f^{+}$ is a complete homomorphism from $\mathbb{W}^+_2$ to $\mathbb{W}^+_1$.
\item If $f$ is surjective, then $f^{+}$ is injective.
\item If $f$ is 
an order-embedding, 
then $f^{+}$ is surjective. 
\end{enumerate}
\end{proposition}

\begin{proof}
(i) Proving that $f^+$ preserves arbitrary intersections and unions is routine. By (6) of Definition \ref{def:DInFL-frame-morphism}, we have $f^+(I_2) = f^{-1}[I_2] = I_1$. To see that $f^+$ preserves $\circ_2$, let $w_1 \in f^{+}(U_2 \circ_2V_2)$. Then $f(w_1) \in U_2 \circ_2 V_2$. Hence, there are $u_2 \in U_2$ and $v_2 \in V_2$ such that $f(w_1) \in u_2\circ_2 v_2$. By item (3) of Definition~\ref{def:DInFL-frame-morphism}, there are $u_1, v_1\in W_1$ such that $u_2 \preccurlyeq_2 f(u_1)$, $v_2\preccurlyeq_2 f(v_1)$ and $w_1 \in u_1\circ_1 v_1$. Since $U_2$ is an upset of $(W_2, \preccurlyeq_2)$, it follows  that $f(u_1) \in U_2$, and so $u_1 \in f^{+}(U_2)$. Likewise, $f(v_1) \in V_2$, so $v_1 \in f^{+}(V_2)$. Hence, since $w_1 \in u_1\circ_1 v_1$, it follows that $w_1 \in f^{+}(U_2)\circ_1 f^{+}(V_2)$. This shows that $f^{+}(U_2\circ_2 V_2) \subseteq f^{+}(U_2)\circ_1 f^{+}(V_2)$. 

For the other containment, let $w_1 \in f^{+}(U_2)\circ_1 f^{+}(V_2)$. Then there exist $u_1 \in f^{+}(U_2)$ and $v_1 \in f^{+}(V_2)$ such that $w_1 \in  u_1 \circ_1 v_1$. Thus, $f(u_1) \in U_2$, $f(v_1) \in V_2$ and $w_1 \in u_1\circ_1 v_1$. Applying (2) of Definition~\ref{def:DInFL-frame-morphism} to the last part gives $f(w_1) \in f(u_1)\circ_2 f(v_1)$. Hence, $f(w_1) \in U_2 \circ_2 V_2$, and so $w_1 \in f^{+}(U_2 \circ V_2)$. 

Next we show that $f^{+}$ preserves ${\sim}_2$. Using (4) of Definition~\ref{def:DInFL-frame-morphism}, we get $w_1 \in f^{+}({\sim}_2U_2)$ iff $f(w_1) \in {\sim}_2U_2$ iff $f(w_1)^{-_{2}} \notin U_2$ iff $f(w_1^{-_1}) \notin U_2$ iff $w_1^{-_1} \notin f^{+}(U_2)$ iff $w_1 \in {\sim}_1f^{+}(U_2)$.

Likewise, using (5) of Definition~\ref{def:DInFL-frame-morphism}, we can show that $f^{+}$ preserves $-_2$.

(ii) Let $U_2, V_2 \in \mathsf{Up}\left(W_2, \preccurlyeq_2\right)$ such that $U_2 \neq V_2$. Then, without loss of generality, there is some $w_2\in W_2$ such that $w_2 \in U_2$ while $w_2 \notin V_2$. Since $f$ is surjective, there is some $w_1 \in W_1$ such that $f(w_1)= w_2$. Hence, $f(w_1) \in U_2$ while $f(w_1) \notin V_2$, and so $w_1 \in f^{+}(U_2)$ while $w_1 \notin f^{+}(V_2)$. This shows $f^{+}(U_2) \neq f^{+}(V_2)$. 


(iii) Let $U_1 \in  \mathsf{Up}\left(W_1, \preccurlyeq_1\right)$. We have to find some $U_2 \in \mathsf{Up}\left(W_2, \preccurlyeq_2\right)$ such that $f^{+}(U_2) = U_1$. Set $U_2 = {\uparrow}f[U_1]= {\uparrow}\{f(u_1) \mid u_1 \in U_1\}$. To see that $f^{+}(U_2) = U_1$, let $u_1 \in U_1$. Then $f(u_1) \in f[U_1]$, and so, since $f(u_1) \preccurlyeq_2 f(u_1)$, we have $f(u_1) \in U_2$. This shows $u_1 \in f^{-1}[U_2] = f^+(U_2)$, and hence $U_1 \subseteq f^{+}(U_2)$. Now let $u_1 \in f^{+}(U_2) = f^{-1}[U_2]$. Then we have $f(u_1) \in U_2 = {\uparrow}f[U_1]$. Hence, $f(v_1) \preccurlyeq_2 f(u_1)$ for some $v_1 \in U_1$. But $f$ is an order-embedding, so $v_1 \preccurlyeq_1 u_1$. Since $U_1 \in \mathsf{Up}\left(W_1, \preccurlyeq_1\right)$, it follows that $u_1 \in U_1$, and therefore $f^{+}(U_2) \subseteq U_1$. 
\end{proof}

In the opposite direction, a complete homomorphism from a complete perfect DInFL-algebra $\mathbf A$ to a complete perfect DInFL-algebra $\mathbf{B}$ gives rise to a DInFL-frame morphism from the DInFL-frame of completely join-irreducibles of $\mathbf{B}$ to the DInFL-frame of completely join-irreducibles of $\mathbf{A}$. We will need the following lemma:

\begin{lemma}\label{lem:interaction-h_and_invh}
Let $\mathbf{A}=(A,\wedge^\mathbf{A}, \vee^\mathbf{A}, \cdot^\mathbf{A}, 1^\mathbf{A}, \sim^\mathbf{A}, -^\mathbf{A})$ and $\mathbf{B}=(B,\wedge^\mathbf{B}, \vee^\mathbf{B}, \cdot^\mathbf{B}, 1^\mathbf{B}, \sim^\mathbf{B}, -^\mathbf{B})$ be complete perfect DInFL-algebras, and let $h: \mathbf A \to \mathbf B$ be a complete homomorphism. 
\begin{enumerate}
[label=\textup{(\roman*)}]
\item If $b \in B$, then $b \leqslant_\mathbf{B}h\left(\bigwedge h^{-1}[{\uparrow}b]\right) = h\left(\bigwedge \{a\in A\mid b \leqslant_{\mathbf{B}} h(a)\}\right)$.
\item If $a \in A$, $b \in B$ and $b\leqslant_\mathbf{B} h(a)$, then $\bigwedge h^{-1}[{\uparrow}b] \leqslant_\mathbf{A} a$.
\item If $b \in J^\infty(\mathbf{B})$, then $\bigwedge h^{-1}[{\uparrow}b] \in J^\infty(\mathbf{A})$.
\item If $a \in A$, $b \in J^{\infty}(\mathbf{B})$ and $\bigwedge h^{-1}[{\uparrow}b] \leqslant_\mathbf{A} a$, then $b\leqslant_\mathbf{B} h(a)$. 
\end{enumerate}
\end{lemma}

\begin{proof}
(i) Since $h$ is a complete homomorphism, proving (i) is equivalent to showing that $b \leqslant_\mathbf{B} \bigwedge h[h^{-1}[{\uparrow}b]]$. To this end, let $c$ be an arbitrary element of $B$ such that $c \in h[h^{-1}[{\uparrow}b]]$. Then there is a $d \in A$ such that $c = h(d)$ and  $d \in h^{-1}[{\uparrow} b]$. Hence, $b\leqslant_\mathbf{B} h(d)$, which means $b \leqslant_\mathbf{B} c$. Since $c$ was arbitrary, we obtain $b \leqslant \bigwedge h[h^{-1}[{\uparrow}b]]$.

(ii) Assume $b \leqslant_\mathbf{B} h(a)$. Then $a \in h^{-1}[{\uparrow}b]$, and so $\bigwedge h^{-1}[{\uparrow}b] \leqslant_\mathbf{A} a$. 

(iii) Let $S \subseteq A$ and assume $\bigwedge h^{-1}[{\uparrow} b]=\bigvee S$. By part (i), we have $b \leqslant_\mathbf{B} h\left(\bigwedge h^{-1}[{\uparrow}b]\right)$. Hence, $b \leqslant_\mathbf{B} h\left(\bigvee S\right) = \bigvee h[S]$. Since $b$ is completely join-prime, there is some $s \in S$ such that $b \leqslant_\mathbf{B} h(s)$. This gives $\bigwedge h^{-1}[{\uparrow}b] \leqslant_\mathbf{A} s$ by part (ii). Therefore, since $s \leqslant_\mathbf{A} \bigvee S$, we have $s \leqslant_\mathbf{A} \bigwedge h^{-1}[{\uparrow}b]$. It thus follows that $s = \bigwedge h^{-1}[{\uparrow}b]$, which means $\bigwedge h^{-1}[{\uparrow}b] \in S$. 

(iv) Assume $\bigwedge h^{-1}[{\uparrow}b] \leqslant_\mathbf{A} a$. Now $ \bigwedge h^{-1}[{\uparrow}b] \in J^\infty(\mathbf{A})$ by (iii), so it follows from item (i) that $
b \leqslant_\mathbf{B} h(\bigwedge h^{-1}[{\uparrow}b]) \leqslant_\mathbf{B} \bigvee\{h(j) \mid j \in J^\infty(\mathbf{A}) \textnormal{ and } j \leqslant a\} = h(\bigvee\{j \in J^\infty(\mathbf{A}) \mid j \leqslant a\} = h(a)$. 
\end{proof}

In the following definition, the \emph{dual} of a homomorphism $h$ from a complete perfect DInFL-algebra $\mathbf{A}$ to a complete perfect DInFL-algebra $\mathbf{B}$ is given. Item (iii) of Lemma~\ref{lem:interaction-h_and_invh} guarantees that this map is well-defined. 

\begin{definition}\label{def:dual_homo_compl_alg}
Let $\mathbf{A}=(A,\wedge^\mathbf{A}, \vee^\mathbf{A}, \cdot^\mathbf{A}, 1^\mathbf{A}, \sim^\mathbf{A}, -^\mathbf{A})$ and $\mathbf{B}=(B,\wedge^\mathbf{B}, \vee^\mathbf{B}, \cdot^\mathbf{B}, 1^\mathbf{B}, \sim^\mathbf{B}, -^\mathbf{B})$ be complete perfect DInFL-algebras, and let $h: \mathbf A \to \mathbf B$ be a complete homomorphism. The  
\emph{dual} of $h$ is the map $h_{+}: J^\infty\left(\mathbf{B}\right) \to J^\infty\left(\mathbf{A}\right)$ defined by $h_+(b)=\bigwedge h^{-1}[{\uparrow} b]$, for all $b \in J^\infty\left(\mathbf{B}\right)$.   \end{definition}

The following proposition asserts that the duals of complete DInFL-homomorphisms are DInFL-frame morphisms.

\begin{proposition} \label{prop:DInFL-homomorphisms-duality}
Let $\mathbf{A}=(A,\wedge^\mathbf{A}, \vee^\mathbf{A}, \cdot^\mathbf{A}, 1^\mathbf{A}, \sim^\mathbf{A}, -^\mathbf{A})$ and $\mathbf{B}=(B,\wedge^\mathbf{B}, \vee^\mathbf{B}, \cdot^\mathbf{B}, 1^\mathbf{B}, \sim^\mathbf{B}, -^\mathbf{B})$ be complete perfect DInFL-algebras, and let $h: \mathbf A \to \mathbf B$ be a complete homomorphism. 
\begin{enumerate}
[label=\textup{(\roman*)}]
\item The map $h_+$ is a DInFL-frame morphism from $\mathbf{B}_+$ to $\mathbf{A}_+$.
\item If $h$ is injective, then $h_+$ is surjective. 
\item If $h$ is surjective, then $h_+$ is 
an order-embedding.
\end{enumerate}
\end{proposition}

\begin{proof}
(i) To prove that condition (1) of Definition~\ref{def:DInFL-frame-morphism} holds, let $b_1, b_2 \in J^\infty(\mathbf B)$ and assume $b_1 \preccurlyeq_\mathbf{B} b_2$. Then $b_2 \leqslant_\mathbf{B} b_1$. We have to show that $h_{+}(b_1) \preccurlyeq_\mathbf{A} h_{+}(b_2)$, i.e., that $h_{+}(b_2) \leqslant_\mathbf{A} h_{+}(b_1)$. It is therefore enough to show that $h^{-1}[{\uparrow} b_1] \subseteq h^{-1}[{\uparrow}b_2]$. To this end, let $a \in h^{-1}[{\uparrow} b_1]$. Then $b_1 \leqslant_\mathbf{B} h(a)$, and so, since $b_2 \leqslant_\mathbf{B} b_1$, we have $b_2 \leqslant_\mathbf{B} h(a)$. Hence, $a \in h^{-1}[{\uparrow}b_2]$, as required.

For (2), let $b_1, b_2, b_3 \in J^\infty(\mathbf{B})$ and assume $b_1 \in b_2 \circ^\mathbf{B} b_3$. Then $b_1 \leqslant_\mathbf{B} b_2 \cdot^{\mathbf B} b_3$. Now, by part (i) of Lemma~\ref{lem:interaction-h_and_invh}, $b_2 \leqslant_\mathbf{B} h\left(\bigwedge h^{-1}[{\uparrow}b_2]\right)$ and $b_3 \leqslant_\mathbf{B} h\left(\bigwedge h^{-1}[{\uparrow}b_3]\right)$. Hence, 
\[
b_1 \leqslant_\mathbf{B} b_2 \cdot^{\mathbf{B}} b_3 \leqslant_\mathbf{B} h\left(\bigwedge h^{-1}[{\uparrow}b_2]\right)\cdot^\mathbf{B} h\left(\bigwedge h^{-1}[{\uparrow}b_3]\right) = h\left(h_+(b_2)\right) \cdot^{\mathbf{B}} h\left(h_{+}(b_3)\right).
\]
Since $h(h_+(b_2)) \cdot^\mathbf{B} h(h_+(b_3)) = h(h_+(b_2)\cdot^\mathbf{A} h_+(b_3))$, we have $b_1 \leqslant_\mathbf{B} h(h_+(b_2)\cdot^\mathbf{A} h_+(b_3))$. By part (ii) of Lemma~\ref{lem:interaction-h_and_invh}, $h_+(b_1) \leqslant_{\mathbf{B}} h_+(b_2)\cdot^\mathbf{A} h_+(b_3))$. This gives $h_+(b_1) \in h_+(b_2)\cdot^\mathbf{A} h_+(b_3)$. 

For the proof of (3), let $c\in J^\infty(\mathbf B), a_1,a_2\in J^\infty(\mathbf A)$ and assume $h_+(c)\in a_1\circ^\mathbf{A} a_2$, or equivalently, $h_+(c)\leqslant_\mathbf{A} a_1\cdot^\mathbf{A} a_2$. Then $c\leqslant_\mathbf{B} h(a_1\cdot^\mathbf{A} a_2) = h(a_1)\cdot^\mathbf{B} h(a_2)$, and since $\mathbf B$ is perfect,  
\begin{align*}
h(a_1)\cdot^\mathbf{B} h(a_2)&=\bigvee\{i\in J^\infty(\mathbf B)\mid i\leqslant_\mathbf{B} h(a_1)\}\cdot^\mathbf{B} h(a_2)\\
&=\bigvee\{i\cdot^\mathbf{B} h(a_2)\mid i \in J^\infty(\mathbf{B}) \textnormal{ and } i\leqslant_\mathbf{B} h(a_1)\}.
\end{align*}
Because $c$ is completely join-prime, it follows that $c\leqslant_\mathbf{B} b_1\cdot^\mathbf{B} h(a_2)$ for some $b_1 \in J^\infty(\mathbf{B})$ such that $b_1\leqslant_\mathbf{B} h(a_1)$. Similarly we find $b_2 \in J^\infty({\mathbf{B}})$ such that $b_2\leqslant_\mathbf{B} h(a_2)$ and $c\leqslant_\mathbf{B} b_1\cdot^\mathbf{B} b_2$. Therefore, $h_+(b_1)\leqslant_\mathbf{A} a_1$, $h_+(b_2)\leqslant_\mathbf{A} a_2$, or equivalently $a_1\preccurlyeq_\mathbf{A} h_+(b_1)$, $a_2\preccurlyeq_\mathbf{A} h_+(b_2)$, and $c\in b_1\circ^\mathbf{B} b_2$.

To show that (4) holds, let $b \in J^\infty(\mathbf{B})$. We have to show that $\bigwedge h^{-1}[{\uparrow}b^{{\sim}^{\mathbf{B}}}] = {\sim}^\mathbf{A}\kappa(h_{+}(b))$. First, let $c \in A$ such that $c \leqslant_\mathbf{A} \bigwedge h^{-1}[{\uparrow}b^{{\sim}^{\mathbf{B}}}]$. We have to show that $c \leqslant_\mathbf{A} {\sim}^\mathbf{A}\kappa(h_+(b))$. This is equivalent to proving that $\kappa(h_{+}(b)) \leqslant_\mathbf{A} {-}^\mathbf{A}c$. To this end, let $a$ be an arbitrary element of $A$ such that $h_{+}(b) \not\leqslant_{\mathbf{A}} a$. This implies $d \not\leqslant_\mathbf{A} a$ for all $d \in A$ such that $b \leqslant_\mathbf{B} h(d)$. Consequently, $b\not\leqslant_\mathbf{B} h(a)$; for otherwise, $a \not\leqslant_\mathbf{A} a$, a contradiction. It thus follows that $h(a) \leqslant_\mathbf{B} \kappa(b)$, which means $b^{{\sim}^\mathbf{B}} = {\sim}^\mathbf{B}\kappa(b) \leqslant_\mathbf{B} {\sim}^\mathbf{B}h(a)$. Since $h({\sim}^\mathbf{A}a) = {\sim}^\mathbf{B}h(a)$, we have  $b^{{\sim}^\mathbf{B}} \leqslant_\mathbf{B} h({\sim}^\mathbf{A}a)$. Hence, it must be the case that $c \leqslant_\mathbf{A} {\sim}^\mathbf{A}a$. Therefore, $a \leqslant_\mathbf{A} {-}^\mathbf{A}c$, and so, since $a$ was arbitrary, we have $\kappa(h_{+}(b)) \leqslant_\mathbf{A} {-}^\mathbf{A}c$.

Now let $c \in A$ such that $c \leqslant_\mathbf{A} {\sim}^\mathbf{A}\kappa(h_{+}(b))$. We have to show that $c \leqslant_\mathbf{A} \bigwedge h^{-1}[{\uparrow}b^{{\sim}^{\mathbf{B}}}]$, so let $a$ be an arbitrary element of $A$ such that $b^{{\sim}^{\mathbf{B}}} \leqslant_\mathbf{B} h(a)$. We have to show $c \leqslant_\mathbf{A} a$. We have  ${\sim}^\mathbf{B}\kappa(b) \leqslant_\mathbf{B} h(a)$, and so, since $h(-^\mathbf{A}a) = -^\mathbf{B}h(a)$, it follows that 
$h(-^\mathbf{A}a) = -^{\mathbf{B}}h(a) \leqslant_\mathbf{B} \kappa(b)$. Now $c \leqslant_\mathbf{A} {\sim}^\mathbf{A}\kappa(h_{+}(b))$ implies that $\kappa(h_+(b)) \leqslant_\mathbf{A} -^\mathbf{A} c$, which means $d \leqslant_\mathbf{A} -^\mathbf{A}c$ for all $d \in A$ such that $h_+(b) \not\leqslant_\mathbf{A} d$. If we can thus show that $h_+(b) \not\leqslant_\mathbf{A} -^\mathbf{A}a$, we would have $-^{\mathbf{A}}a \leqslant_\mathbf{A} -^{\mathbf{A}}c$, which gives $c \leqslant_\mathbf{A} a$, as required. To see why $h_+(b) \not\leqslant_\mathbf{A} -^\mathbf{A}a$, suppose $h_+(b) \leqslant_\mathbf{A} -^\mathbf{A}a$. Thus, by part by part (iv) of Lemma~\ref{lem:interaction-h_and_invh}, $b \leqslant_\mathbf{B} h(-^{\mathbf{A}}a)$, and therefore $b \leqslant_\mathbf{B} \kappa(b)$. Since $b$ is completely join-prime, $b \leqslant_\mathbf{B} b'$ for some $b' \in B$ such that $b \not\leqslant_\mathbf{B} b'$, a contradiction. 

In a similar way we can show that (5) holds. 

For (6), let $b \in J^{\infty}(\mathbf{B})$. Using Lemma~\ref{lem:interaction-h_and_invh} and the fact that $h(1^{\mathbf{A}}) = 1^\mathbf{B}$, we have $b \in I_{1^\mathbf{B}}$ iff $b \leqslant_\mathbf{B} 1^{\mathbf{B}}$ iff $b \leqslant_{\mathbf{B}} h(1^{\mathbf{A}})$ 
iff $\bigwedge h^{-1}[{\uparrow} b] \leqslant_\mathbf{A} 1^{\mathbf{A}}$ iff $\bigwedge h^{-1}[{\uparrow} b] \in I_{1^{\mathbf{A}}}$ iff $h_+(b) \in I_{1^{\mathbf{A}}}$ iff $b \in h_+^{-1}[I_{1^\mathbf{A}}]$.

(ii) Assume $h:\m A\to \m B$ is a complete injective homomorphism. Let $a \in J^{\infty}(\mathbf{A})$. Since $\m B$ is perfect, $h(a)=\bigvee\{c\in J^\infty(\m B)\mid c\leqslant_\mathbf{B} h(a)\}$. Moreover, $h_+(c)\leqslant_\mathbf{A} a$ for all $c\in J^\infty(\m B)$ such that $c\leqslant_\mathbf{B} h(a)$. Hence, $\bigvee\{h_+(c)\mid c\in J^\infty(\m B) \textnormal{ and } c\leqslant_\mathbf{B} h(a)\}\leqslant_\mathbf{A} a$. Applying $h$ we get $h(\bigvee\{h_+(c)\mid c\in J^\infty(\m B) \textnormal{ and } c\leqslant_\mathbf{B} h(a)\})= \bigvee \{h(h_+(c))\mid c\in J^\infty(\m B) \textnormal{ and } c\leqslant_\mathbf{B} h(a)\}\leqslant_{\mathbf{B}} h(a)$. Since $c \leqslant_\mathbf{A} h(h_+(c))$ for all $c \in J^\infty(\mathbf{B})$ such that $c \leqslant_\mathbf{B} h(a)$ by item (i) of Lemma~\ref{lem:interaction-h_and_invh}, we have $$h(a) =\bigvee\{c\in J^\infty(\m B)\mid c\leqslant_\mathbf{B} h(a)\} \leqslant_\mathbf{B} \bigvee \{h(h_+(c))\mid c\in J^\infty(\m B) \textnormal{ and } c\leqslant_\mathbf{B} h(a)\}.$$
Therefore, $h(a) = h(\bigvee\{h_+(c)\mid c\in J^\infty(\m B) \textnormal{ and } c\leqslant_\mathbf{B} h(a)\})$. By the injectivity of $h$ it follows that $\bigvee\{h_+(c)\mid c\in J^\infty(\m B) \textnormal{ and } c\leqslant_\mathbf{B} h(a)\}=a$. Now, since $a$ is a completely join-irreducible element of $\mathbf{A}$, $a=h_+(c)$ for some $c\in J^\infty(\m B)$, hence $h_+$ is surjective.

(iii) Let $b_1, b_2 \in J^{\infty}(\mathbf{B})$ such that $b_1 \not\leqslant_\mathbf{B} b_2$. We have to show that $h_{+}(b_1) \not\leqslant_\mathbf{B} h_+(b_2)$.  Since $h$ is surjective, there are $a_1, a_2 \in A$ such that $h(a_1) = b_1$ and $h(a_2) = b_2$. Hence,  $b_1 \not\leqslant_\mathbf{B} h(a_2)$, and so, by item (iv) of Lemma~\ref{lem:interaction-h_and_invh}, $\bigwedge h^{-1}[{\uparrow}b_1] \not\leqslant_\mathbf{B} a_2$.  Now $b_2 = h(a_2)$, so $a_2 \in h^{-1}[{\uparrow}b_2]$.  Consequently, $\bigwedge h^{-1}[{\uparrow}b_1] \not\leqslant_\mathbf{B} \bigwedge h^{-1}[{\uparrow}b_1]$, i.e., $h_{+}(b_1) \not\leqslant_\mathbf{B} h_+(b_2)$.
\end{proof}

We now extend the definition of a DInFL-frame morphism to the definition of a DqRA-frame morphism.

\begin{definition}\label{def:DqRA-frame-morphism}
Let $\mathbb{W}_1= \left(W_1, I_1,\preccurlyeq_1, \RR_1, {^{\sim_1}}, {^{-_1}}, {^{\neg_1}}\right)$
and 
$\mathbb{W}_2= \left(W_2, I_2,\preccurlyeq_2, \RR_2, {^{\sim_2}}, {^{-_2}}, {^{\neg_2}}\right)$ be DqRA-frames. A function $f: W_1 \to W_2$ is a \emph{DqRA-frame morphism} if 
$f$ is a DInFL-frame morphism and the following condition holds:
\begin{enumerate}
[label=\textup{(\arabic*)}]
\setcounter{enumi}{6}
\item $f(x^{\neg_1})=f(x)^{\neg_2}$
\end{enumerate}
\end{definition}

If $f: W_1 \to W_2$ is a DqRA-frame morphism from a DqRA-frame $\mathbb{W}_1$ to a DqRA-frame $\mathbb{W}_2$, then we define the dual $f^{+}: \mathsf{Up}\left(W_2, \preccurlyeq_2\right) \to \mathsf{Up}\left(W_1, \preccurlyeq_1\right)$ of $f$ by $f^{+}(U_2)=f^{-1}[U_2]$, for all $U_2 \in \mathsf{Up}\left(W_2, \preccurlyeq_2\right)$.

\begin{proposition}  
\label{prop:DqRA-frame-morphisms-duality}
Let $\mathbb{W}_1= \left(W_1, I_1,\preccurlyeq_1, \RR_1, {^{\sim_1}}, {^{-_1}}, {^{\neg_1}}\right)$
and 
$\mathbb{W}_2= \left(W_2, I_2,\preccurlyeq_2, \RR_2, {^{\sim_2}}, {^{-_2}}, {^{\neg_2}}\right)$ be DqRA-frames, and let $f: W_1\to W_2$ be a DqRA-frame morphism.
\begin{enumerate}
[label=\textup{(\roman*)}]
\item The map $f^{+}$ is a complete homomorphism from $\mathbb{W}^+_2$ to $\mathbb{W}^+_1$.
\item If $f$ is surjective, then $f^{+}$ is injective.
\item If $f$ is 
an order-embedding, 
then $f^{+}$ is surjective.
\end{enumerate}
\end{proposition}

\begin{proof}
For (i), we just need to check that $f^{+}$ preserves $\neg_2$. Using (7) of Definition~\ref{def:DqRA-frame-morphism}, we have 
\begin{align*}
w_1 \in f^{+}({\neg}_2U_2) & \iff f(w_1) \in {\neg}_2U_2\\
& \iff f(w_1)^{\neg_{2}} \notin U_2\\
& \iff f(w_1^{\neg_1}) \notin U_2 \\
& \iff w_1^{\neg_1} \notin f^{+}(U_2) \\
& \iff w_1 \in {\neg}_1f^{+}(U_2).
\end{align*}
Items (ii) and (iii) follow from Proposition~\ref{prop:DInFL-frame-morphisms-duality}.
\end{proof}

If $\mathbf{A}=(A,\wedge^\mathbf{A}, \vee^\mathbf{A}, \cdot^\mathbf{A}, 1^\mathbf{A}, \sim^\mathbf{A}, -^\mathbf{A}, \neg ^\mathbf{A})$ and $\mathbf{B}=(B,\wedge^\mathbf{B}, \vee^\mathbf{B}, \cdot^\mathbf{B}, 1^\mathbf{B}, \sim^\mathbf{B}, -^\mathbf{B}, \neg^\mathbf{B})$ are complete perfect DqRAs and $h: \mathbf A \to 
\mathbf B$ is a complete homomorphism, then we define the  
\emph{dual} $h_{+}: J^\infty\left(\mathbf{B}\right) \to J^\infty\left(\mathbf{A}\right)$ of $h$ by $h_+(b)=\bigwedge h^{-1}[{\uparrow} b]$, for all $b \in J^\infty\left(\mathbf{B}\right)$.   

\begin{proposition} \label{prop:DqRA-homomorphisms-duality}
Let $\mathbf{A}$ and $\mathbf{B}$ 
be complete perfect DqRAs, and let $h: \mathbf A \to \mathbf B$ be a complete homomorphism. 
\begin{enumerate}
[label=\textup{(\roman*)}]
\item The map $h_+$ is a DqRA-frame morphism from $\mathbf{B}_+$ to $\mathbf{A}_+$.
\item If $h$ is surjective, then $h_+$ is 
an order-embedding. 
\item If $h$ is injective, then $h_+$ is surjective. 
\end{enumerate}
\end{proposition}

\begin{proof}
For (i) we only need to check that (7) of Definition~\ref{def:DqRA-frame-morphism} holds, and the proof of this is analogous to the proof that $h_{+}(b^{\sim^\mathbf{B}}) = \left(h_+(b)\right)^{\sim^\mathbf{A}}$ for all $b \in J^\infty(\mathbf{B})$ (see Proposition~\ref{prop:DInFL-homomorphisms-duality}, also  for (ii), (iii)).
\end{proof}

\section{Priestley-style duality 
for DInFL-algebras and DqRAs}\label{sec:dual-spaces}

In this section, we will use the results of Section~\ref{sec:frames}
to define Priestley spaces with additional structure that will be dual to DInFL-algebras and DqRAs. As the signatures of  DInFL-algebras and DqRAs do not include lattice bounds, our dual spaces will in fact be \emph{doubly-pointed} Priestley spaces (i.e. the poset will have both a  least and greatest element). Such spaces have also been called \emph{bounded Priestley spaces} (cf. Section 1.2 and Theorem~4.3.2 of Clark and Davey~\cite{CD98}). More recent papers by Cabrer and Priestley (cf.~\cite{CP14}) refer to the spaces as doubly-pointed.  

We recall that a partially ordered topological space $(X,\leqslant ,\tau)$ is \emph{totally order-disconnected} if whenever $x \nleqslant y$ there exists a clopen upset $U$ of $X$ such that $x \in U$ and $y \notin U$. A \emph{doubly-pointed Priestley space} is a compact totally order-disconnected space with bounds $\bot \neq\top$. 
When recovering an unbounded distributive lattice from a doubly-pointed Priestley space, the proper, non-empty, clopen upsets form a lattice~\cite[Theorem 1.2.4]{CD98}. 

For $\mathbf{A}$ a DInFL-algebra or DqRA, we call  $F \subseteq A$ a \emph{generalised prime filter} if $F$ is a prime filter, $F=A$ or $F=\varnothing$. This terminology follows, for instance, Fussner and Galatos~\cite{FG19}, although we note that there they only allow $F=A$ as their algebras have a top element, but not a bottom element. Hence, their dual spaces are \emph{pointed} Priestley spaces rather than doubly-pointed Priestley spaces. Our dual spaces will use the set of generalised prime filters as their underlying set. 

We give two examples that show why  we need this version of Priestley duality for our setting. The integers (with $+$ as the monoid operation) are a commutative distributive residuated lattice and hence a DInFL-algebra. As the underlying lattice is unbounded, we cannot represent it as usual as the clopen upsets of a Priestley space $(X,\leqslant, \tau)$ as this would introduce bounds $\varnothing$ and $X$. 
Secondly, we would like to represent the dual of a homomorphism $h$ as $h^{-1}$. Consider the two-element Sugihara monoid $\mathbf{S}_2$ and the homomorphism $h(a_1)=a_1$ and $h(a_{-1})=a_{-1}$ into the four-element Sugihara monoid $\mathbf{S}_4$ (whose elements are $a_{-2}<a_{-1}<a_1<a_2$). Now ${\uparrow}a_2$ is a prime filter of $\mathbf{S}_4$, but $h^{-1}({\uparrow}a_2)=\emptyset$ is not a prime filter of  $\mathbf{S}_2$.

First, we develop  some 
results about generalised prime filters that will be needed later.
\begin{lemma}\label{lem:F^-_and_F^sim_prime_filters}
Let $\mathbf{A}=(A,\wedge, \vee, \cdot, 1, \sim, -)$ be a DInFL-algebra. If $F$ is a generalised prime filter of the lattice reduct of $\mathbf A$, then so are
$F^{\sim} = \left\{{\sim}a \mid a \notin F\right\}$
and $F^- = \left\{-a\mid a \notin F\right\}$.
\end{lemma}
\begin{proof}
Let $F$ be a generalised prime filter. We show only the case of $F^{\sim}$, with the case of $F^{-}$ being similar. If $F=A$, then $F^{\sim}=\varnothing$ and if $F=\varnothing$, then $F^{\sim}=A$. Now assume $F\neq A$ and $F \neq \varnothing$. Let $a \in F^{\sim}$ and $a \leqslant b$. Now $a={\sim}-a$ with $-a\notin F$. Since $-$ is order-reversing, $-b \leqslant -a$
so $-b \notin F$ and hence $b={\sim}-b \in F^{\sim}$. Now assume $a \vee b \in F^{\sim}$. We have $a\vee b = \sim -(a\vee b)=\sim(-a \wedge -b)$. Since $A{\setminus}F$ is a prime ideal, we have $-a \notin F$ or $-b \notin F$. Hence we get $a \in F^{\sim}$ or $b \in F^{\sim}$.
\end{proof}

The following lemma is a restatement of \cite[Lemma 3.5]{Gal2000}, but for generalised prime filters (see also \cite[Lemma 2.2]{Urq1996}).  We denote by $F\cdot G$ the set $\{\, a\cdot b \mid a \in F, b\in G\,\}$. 

\begin{lemma}\label{lem:Urq2.2}
Let $F,G,H$ be filters (possibly empty or total) of a distributive residuated lattice such that $H$ is a generalised prime filter and $F \cdot G \subseteq H$. Then there exist generalised prime filters $F'$ and $G'$ such that $F'\cdot G \subseteq H$ and $F \cdot G' \subseteq H$.     
\end{lemma}
\begin{proof}
This follows from \cite[Lemma 3.5]{Gal2000} with straightforward adaptations to account for the empty set and the whole algebra. 
\end{proof}
In view of the above lemma, we define $F\bullet G=\{H\mid F\cdot G\subseteq H\}$. Note that this is automatically an upward-closed set of generalised prime filters.

\begin{definition} \label{def:dbly-pnted-DInFL-frame} A \emph{doubly-pointed DInFL-frame} is  a DInFL-frame where the poset is bounded, and the set $I$ is a proper, non-empty upset. 
\end{definition} 

We use Lemmas~\ref{lem:F^-_and_F^sim_prime_filters} and~\ref{lem:Urq2.2} to show that that the set of generalised prime filters can be equipped with the necessary structure to be a doubly-pointed DInFL-frame. 

\begin{proposition}\label{prop:prime-filter_DInFL-frame}
Let $\mathbf{A}=(A,\wedge, \vee, \cdot, 1, \sim, -)$ be a DInFL-algebra. Let $W_\mathbf{A}$ be the set of generalised prime filters of the lattice reduct of $\mathbf{A}$. For all $F, G, H$ in $W_\mathbf{A}$, define $F \in \mathcal{I}$ iff $1 \in F$, $F \preccurlyeq G$ iff $F \subseteq G$, 
$F^\sim = \left\{{\sim}a \mid a \notin F\right\}$ and $F^- = \left\{-a\mid a \notin F\right\} $. Then the structure $\mathfrak{F}(\mathbf{A})=\left(W_\mathbf{A}, \mathcal{I}, \preccurlyeq, \bullet, ^{\sim}, ^{-}, \varnothing, A\right)$ is a doubly-pointed DInFL-frame. 
\end{proposition}

\begin{proof}
The relation $\preccurlyeq$ is clearly a partial order. To see that $\mathcal{I}$ is a upset of $\left(W, \preccurlyeq\right)$, let $F \in \mathcal{I}$ and assume $F \preccurlyeq G$. Then $1 \in F \subseteq G$, so $G \in \mathcal{I}$. Since $\varnothing \notin \mathcal{I}$ and $A \in \mathcal{I}$, it is a proper, non-empty subset of $W_\mathbf{A}$.

Conditions (1) and (2) of Definition~\ref{def:DInFL-frames} can be proven using Lemma~\ref{lem:Urq2.2} and the filter ${\uparrow}1$. Condition (3) follows easily from the fact that the order on $W_\mathbf{A}$ is set containment. For (4), see \cite[Theorem 3.7(2)]{Gal2000}, but apply Lemma~\ref{lem:Urq2.2} above. 

Next we show that (5) holds, i.e., that $\Rb{F}{G}{H^{\sim}}$ iff $\Rb{H}{F}{G^{-}}$. Assume $\Rb{F}{G}{H^{\sim}}$. If at least one of $F,G$ or $H$ is empty, then $\Rb{H}{F}{G^{-}}$ is trivially satisfied. Now assume they are all non-empty and let $a \in H$ and $b \in F$. We must show that $a\cdot b \in G^{-}$.
Suppose for the sake of a contradiction that $a \cdot b \notin G^{-}$. Then ${\sim}\left(a \cdot b\right) \in G$. Hence, since $b \in F$ and $\Rb{F}{G}{H^{\sim}}$, we get $b \cdot {\sim}\left(a \cdot b\right) \in H^{\sim}$. This means $-\left(b \cdot {\sim}\left(a \cdot b\right)\right) \notin H$. Now $a \cdot b \leqslant a \cdot b$, so we have $a \leqslant -\left(b \cdot {\sim}\left(a \cdot b\right)\right)$. Thus, since $a \in H$ and $H$ is upward closed, we get $-\left(b \cdot {\sim}\left(a \cdot b\right)\right) \in H$, a contradiction.  

Conversely, assume $\Rb{H}{F}{G^{-}}$. Again, if any of $F,G$ or $H$ are empty, $\Rb{F}{G}{H^{\sim}}$ is trivially satisfied. Let $a \in F$ and $b \in G$. We have to show that $a \cdot b \in H^\sim$. Suppose for the sake of a contradiction that $a\cdot b \notin H^{\sim}$. Then $-\left(a\cdot b\right) \in H$, and so, since $a\in F$ and $\Rb{H}{F}{G^{-}}$, we have $-\left(a\cdot b\right) \cdot a\in G^{-}$. This means ${\sim}\left(-\left(a\cdot b\right)\cdot a\right) \notin G$. Now $a\cdot b\leqslant a\cdot b$, so we have $b \leqslant {\sim}\left(-\left(a\cdot b\right)\cdot a\right)$, and therefore, since $G$ is upward closed and $b \in G$, we have ${\sim}\left(-\left(a\cdot b\right)\cdot a\right) \in G$, a contradiction. 

For (6), to see that $F^{\sim -} \subseteq F$, let $a \in F^{\sim -}$. Then $a = -b$ for some $b \notin F^\sim$. Hence, ${\sim}a = {\sim}{-}b = b$, and so ${\sim}a \notin F^\sim$. This means $-{\sim}a =a \in F$. 
A similar proof gives $F^{-\sim} \subseteq F$. The cases $F=\varnothing$ and $F=A$ are trivial. 
\end{proof}

Recall that 
after Definition~\ref{def:DInFL-frames} we defined, for $U,V \subseteq W$, the underlying set of a DInFL-frame, the set $U \circ V = \bigcup \{\, x \circ y \mid x \in U, y \in V\,\}$.  
In Proposition~\ref{prop:complex_algebra_DInFL-frame} we gave  definitions of ${\sim}U$ and $-U$. We remark that (3) below is equivalent to both maps $x \mapsto x^{\sim}$ and $x \mapsto x^{-}$ being  continuous.

\begin{definition}\label{def:DInFL-spaces}
A  doubly-pointed DInFL-frame with topology 
$\left(W, I, \preccurlyeq, \RR, ^\sim, ^-, \bot, \top, \tau\right)$ is a
 \emph{DInFL-space} if $\tau$ is a compact totally order-disconnected topology and the following conditions are satisfied: 
\begin{enumerate}[label=\textup{(\arabic*)}]
\item $I$ is clopen. 
\item If $U$ and $V$ are clopen proper non-empty upsets, then $U \circ V$ is clopen.
\item If $U$ is a clopen proper non-empty upset, then  ${\sim} U $ and $-U$ are clopen.
\end{enumerate}
\end{definition}

For a DInFL-algebra $\mathbf{A}$, we consider the structure $\mathfrak{W}(\mathbf{A})=(\mathfrak{F}(\mathbf{A}),\tau_P)$ where $\tau_P$ is the topology on the set of generalised prime filters with subbasic open sets of the form $X_a=\{F \in W_\mathbf{A} \mid a \in F\}$ and $X_a^c=\{F \in W_\mathbf{A} \mid a \notin F\}$. 
The operation $\bullet : W_\mathbf{A} \times W_\mathbf{A} \to \mathcal{P} (W_\mathbf{A})$ extends to subsets $U,V \subseteq W_\mathbf{A}$ by 
$U \bullet V = \bigcup \{\, F \bullet G \mid F \in U, G \in V \,\}$. 
For a DInFL-space  $\mathbb{W}$, we denote by $K_\mathbb{W}$ the set of clopen proper non-empty upsets of $\mathbb{W}$ and define $\mathfrak{A}(\mathbb{W})$ to be the  algebra 
$\left(K_\mathbb{W},\cap,\cup,\circ, I, {\sim},-\right)$.

\begin{proposition}\label{prop:W(A)-DInFL-space}
If $\mathbf{A}$ is a DInFL-algebra, then $\mathfrak{W}(\mathbf{A})$ is a DInFL-space and if $
\mathbb{W}$ is a DInFL-space then $\mathfrak{A}(\mathbb{W})$ is a DInFL-algebra.
\end{proposition}
\begin{proof} The fact that $\mathfrak{W}(\mathbf{A})$ has an underlying DInFL-frame structure is the result of Proposition~\ref{prop:prime-filter_DInFL-frame} and the compact totally order-disconnectedness follows from Priestley duality.  By definition,  $\mathcal{I}=X_1$ so it is clopen.  
For (2), we consider $U,V$ clopen proper non-empty upsets of $\mathfrak{W}(\mathbf{A})$, i.e. $U$ and $V$ are sets of generalised prime filters. The fact that $U \bullet  V$ is clopen follows from, for instance, \cite[Theorem 6.3]{JL2022}, noting that for any $a \in A$, we have $\varnothing \notin X_a$ and $A \in X_a$. 
Let $U$ be a clopen proper non-empty upset. From the duality we have that $U=X_a=\{F \in W_\mathbf{A} \mid a \in F\}$ for some $a \in A$. Now $-U = \{F \mid F^{\sim} \notin U \} = \{F \mid a \notin F^{\sim}\}=\{F \mid {\sim}{-}a \notin F^{\sim}\}$. But ${\sim}{-}a\notin F^{\sim}$ iff $-a\in F$. Hence $-U=X_{-a}$, which is clopen. 
A similar proof shows that ${\sim} U$ is clopen.

For a DInFL-space $\mathbb{W}$,  the lattice structure of $\mathfrak{A}(\mathbb{W})$ follows from Priestley duality. The algebra structure follows from the definition of a DInFL-space,  Proposition~\ref{prop:complex_algebra_DInFL-frame} and the fact that the elements of  $K_\mathbb{W}$ are special upsets.  
\end{proof}

\begin{theorem}\label{thm:DInFL-spaces-iso}
Let $\mathbf{A}$ be a DInFL-algebra and $\mathbb{W}$ a DInFL-space. Then we have $\mathbf{A} \cong \mathfrak{A}(\mathfrak{W}(\mathbf{A}))$ and $\mathbb{W}\cong \mathfrak{W}(\mathfrak{A}(\mathbb{W}))$.
\end{theorem}

\begin{proof}
The standard maps $a \mapsto X_a$ and $x \mapsto \{\,U \in K_{\mathbb{W}} \mid x \in U \,\}$ give us the required isomorphisms. 
\end{proof}

Here we will define the dual spaces of DqRAs. They will be topologised versions of the DqRA frames from Section~\ref{sec:frames}.

\begin{lemma}\label{lem:F^neg_prime_filters}
Let $\mathbf{A}=(A,\wedge, \vee, \cdot, 1, \sim, -, \neg)$ be a DqRA. If $F$ is a generalised prime filter of the lattice reduct of $\mathbf A$, then so is $F^{\neg}= 
 \left\{{\neg}a \mid a \notin F\right\}$.
\end{lemma}
\begin{proof} This follows from a similar  argument to that in the proof of Lemma~\ref{lem:F^-_and_F^sim_prime_filters}. 
\end{proof}

As for DInFL-frames, we will consider doubly-pointed DqRA-frames, which are bounded and have the additional constraint that $I$ must be proper and non-empty. 

\begin{proposition}\label{prop:prime-filter_DqRA-frame}
Let $\mathbf{A}=(A,\wedge, \vee, \cdot, 1, \sim, -, \neg)$ be a DqRA. Define $\mathcal{I}$, $\preccurlyeq$, $\bullet$, $^-$ and $^\sim$ as in Proposition \ref{prop:prime-filter_DInFL-frame} and for all $F \in W_{\mathbf{A}}$, define 
$F^\neg = \left\{\neg a \mid a \notin F\right\}$. Then  $\left(W_\mathbf{A}, \mathcal{I}, \preccurlyeq, \bullet, ^{\sim}, ^{-}, ^\neg, \varnothing, A\right)$ is a
doubly-pointed
DqRA-frame. 
\end{proposition}

\begin{proof}
We must show that (7), (8) and (9) from Definition~\ref{def:DqRA_frames} hold. For (7), since $\neg \neg a = a$, we have $a \in F$ iff $\neg a \notin F^{\neg}$ iff $\neg \neg a \in F^{\neg\neg}$ iff $a \in F^{\neg\neg}$. 

To see that $F \preccurlyeq G$ implies $G^\neg \preccurlyeq F^\neg$, assume $F \subseteq G$ and let $a \in G^{\neg}$. The latter implies that $a = \neg b$ for some $b \notin G$. Hence, $\neg a = \neg \neg b = b$, and therefore $\neg a \notin G$. Since $F \subseteq G$, we have $\neg a \notin F$, and so $\neg\neg a =a \in F^\neg$. 

Finally, we show (9), i.e. 
$\Rb FG{H^{-}}$ iff $\Rb{G^{\sim\neg}}{F^{\sim\neg}}{H^\neg}$. Assume that $\Rb{F}{G}{H^-}$. Notice that if $G^{\sim\neg}=\varnothing$, $F^{\sim\neg}=\varnothing$, or 
$H=\varnothing$
then $\Rb{G^{\sim\neg}}{ F^{\sim\neg}}{H^\neg}$ is trivially satisfied. Hence let $a \in G^{\sim \neg}$ and $b \in F^{\sim\neg}$. We have to show that $a\cdot b \in H^\neg$. Since $a \in G^{\sim \neg}$, there is some $c \in A$ such that $c \notin G^\sim$ and $a = \neg c$. Hence, $\neg a = \neg \neg c = c$, and so $\neg a \notin G^\sim$. This implies that $-\neg a \in G$. Likewise, since $b \in F^{\sim\neg}$, we have $-\neg b \in F$. Therefore, by our assumption that $\Rb{F}{G}{H^-}$, we get $-\neg b \cdot -\neg a \in H^{-}$. It must therefore be the case that $\sim\left(-\neg b \cdot -\neg a\right) \notin H$. Applying $\mathsf{(Dp)}$, we get $\neg\left(a\cdot b\right) \notin H$, which means $a\cdot b \in H^\neg$, as required. 

Conversely, assume $\Rb{G^{\sim\neg}}{F^{\sim \neg}}{H^\neg}$. If $F=\varnothing$, $G=\varnothing$, or $H =\varnothing$ then $\Rb{F}{G}{H^{-}}$ is trivially true. So, let $a \in F$ and $b \in G$. We must show that $a \cdot b \in H^{-}$. Since $a \in F$, we have ${\sim}a \notin F^{\sim}$, and so $\neg{\sim}a \in F^{\sim\neg}$. Hence, by $\mathsf{(Di)}$, $-\neg a \in F^{\sim\neg}$. Likewise, since $b \in G$, we can show that $-\neg b \in G^{\sim\neg}$. It thus follows from our assumption that $-\neg b \cdot -\neg a \in H^\neg$, which means $\neg\left(-\neg b \cdot -\neg a\right) \notin H$. This gives $-\neg\left(-\neg b \cdot - \neg a\right) \in H^-$, and therefore, by $\mathsf{(Di)}$, $\neg{\sim}\left(-\neg b \cdot -\neg a\right) \in H^{-}$. Applying $\mathsf{(Dp)}$ to this, we get $\neg\neg\left(a \cdot b\right)=a \cdot b \in H^{-}$.  
\end{proof}

We now define the Priestley-style dual objects of DqRAs. As before, for $U$ an upset of a DqRA-frame, we have $\neg U = \{x \in W \mid x^\neg \notin W \}$. 

\begin{definition}
A doubly-pointed DqRA-frame with topology   $(W,I, \preccurlyeq, \RR, ^{\sim}, ^{-}, ^{\neg},\bot, \top,\tau)$
is a \emph{DqRA-space} if $\tau$ is a compact totally order-disconnected topology and the following conditions are satisfied: 
\begin{enumerate}[label=\textup{(\arabic*)}]
\item $I$ is clopen. 
\item If $U$ and $V$ are clopen proper non-empty upsets, then $U \circ V$ is clopen.
\item If $U$ is a clopen proper non-empty upset, then 
${\sim} U$, $-U$ and $\neg U$ are clopen. 
\end{enumerate}
\end{definition}

We extend the maps $\mathfrak{A}$ and $\mathfrak{W}$ to DqRAs and DqRA-spaces, and denote these extensions by $\mathfrak{A}_q$ and $\mathfrak{W}_q$. 

\begin{proposition}
For a DqRA $\mathbf{A}$, let $\mathfrak{W}_q(\mathbf{A})=(W_\mathbf{A}, \mathcal{I},\preccurlyeq ,\bullet , ^{\sim},^{-}, ^{\neg},\varnothing,A,\tau_P)$ and for a DqRA-space $\mathbb{W}$, let  $\mathfrak{A}_q\left(\mathbb{W}\right)=\left(K_{\mathbb{W}}, \cap, \cup, \circ, I, \sim, -, \neg\right)$.
Then $\mathfrak{W}_q\left(\mathbf{A}\right)$ is a DqRA-space and $\mathfrak{A}_q\left(\mathbb{W}\right)$ is a DqRA. 
\end{proposition}

\begin{proof}
Most of the work has been done by Proposition~\ref{prop:W(A)-DInFL-space}. If $U$ is a clopen proper non-empty upset of $\mathfrak{W}(\mathbf{A})$, then $U=X_a$ for some $a \in A$. Then $\neg U =\{F \mid F^{\neg} \notin U\}= \{F \mid a \notin F^\neg \} =\{F \mid \neg \neg a \notin F^\neg\} =\{ F \mid \neg a \in F\}=X_{\neg a}$. To show 
$\mathfrak{A}_q\left(\mathbb{W}\right)$ is a DqRA, we 
use Proposition~\ref{prop:complex_algebra_DqRA-frame} to show that $\neg U$ has the required properties. 
\end{proof}

The same maps from Theorem~\ref{thm:DInFL-spaces-iso} are used to prove the theorem below. 

\begin{theorem}
For any DqRA $\mathbf{A}$ and DqRA-space $\mathbb{W}$, we have that $\mathbf{A} \cong \mathfrak{A}_q(\mathfrak{W}_q(\mathbf{A}))$ and $\mathbb{W} \cong \mathfrak{W}_q(\mathfrak{A}_q(\mathbb{W}))$. 
\end{theorem}

In the remainder of this section, we build on the definitions from Section~\ref{sec:frame-morphisms} and define morphisms for the two classes of spaces. 

\begin{definition}\label{def:DInFL-space-morphism} 
Consider two DInFL-spaces $\mathbb{W}_1=(W_1, \mathcal{I}_1,\preccurlyeq_1,\circ_1,^{{\sim}_1},^{-_1},\bot_1,\top_{\!1} ,\tau_1)$ and $\mathbb{W}_2=
(W_2,\mathcal{I}_2,\preccurlyeq_2, \circ_2,^{{\sim}_2},^{-_2},\bot_2,\top_{\!2} ,\tau_2)$. A function $f : W_1 \to W_2$ is a \emph{DInFL-space morphism} if 
\begin{enumerate}[label=\textup{(\arabic*)}]
\item $f$ is a DInFL-frame morphism
\item $f$ is continuous 
\item $f$ preserves the bounds. 
\end{enumerate}

\end{definition}

There are two ways to define morphisms between DqRA spaces. 

\begin{definition}\label{def:DqRA-space-morphism}
Consider two DqRA-spaces $\mathbb{W}_1=(W_1, \mathcal{I}_1,\circ_1,^{{\sim}_1},^{-_1},^{\neg_1},\bot_1,\top_{\!1} ,\tau_1)$ and $\mathbb{W}_2=
(W_2,\mathcal{I}_2,\circ_2,^{{\sim}_2},^{-_2},^{\neg_2},\bot_2,\top_{\!2} ,\tau_2)$. 
A function $f : W_1 \to W_2$ is a \emph{DqRA-space morphism} if it satisfies either of the following: 
\begin{enumerate}
[label=\textup{(\arabic*)}]
\item $f$ is a DqRA-frame morphism that is continuous and preserves the bounds; 
\item $f$ is a DInFL-space morphism such that $f(x^{\neg_1})=f(x)^{\neg_2}$.
\end{enumerate}
\end{definition}

If $f : \mathbb{W}_1 \to \mathbb{W}_2$ is a DInFL-space morphism, we define $\mathfrak{A}(f): \mathfrak{A}(\mathbb{W}_2) \to \mathfrak{A}(\mathbb{W}_1)$ by $\mathfrak{A}(f)(U)=f^{-1}(U)$ for $U$ a clopen proper non-empty upset of $\mathbb{W}_2$. When $\mathbb{W}_1$ and $\mathbb{W}_2$ are DqRA-spaces and $f: \mathbb{W}_1 \to \mathbb{W}_2$ a DqRA-space morphism, we define $\mathfrak{A}_q(f): \mathfrak{A}_q(\mathbb{W}_2)\to \mathfrak{A}_q(\mathbb{W}_1)$ the same way. If $h :\mathbf{A}_1 \to \mathbf{A}_2$ is a DInFL-algebra homomorphism, then $\mathfrak{W}(h): \mathfrak{W}(\mathbf{A}_2)\to \mathfrak{W}(\mathbf{A}_1)$ is defined by $\mathfrak{W}(h)(F)=h^{-1}(F)$ for $F$  a generalised prime filter of $\mathbf{A}_2$. When $\mathbf{A}_1$ and $\mathbf{A}_2$ are DqRAs and $h: \mathbf{A}_1\to \mathbf{A}_2$ a DqRA homomorphism, we apply the same definition for $\mathfrak{W}_q(h) : \mathfrak{W}_q(\mathbf{A}_2) \to \mathfrak{W}_q(\mathbf{A}_1)$. 

In light of the definitions above we have the following theorem, the proof of which follows from the results of Section~\ref{sec:frame-morphisms} and standard results from Priestley duality. 

\begin{theorem}\label{thm:duality-morphisms-spaces}
\,
\begin{enumerate}[label=\textup{(\roman*)}]
\item Let $h: \mathbf{A}_1 \to \mathbf{A}_2$ be a DInFL-algebra (DqRA) homomorphism. Then $\mathfrak{W}(h): \mathfrak{W}(\mathbf{A}_2) \to \mathfrak{W}(\mathbf{A}_1)$ is a DInFL-space morphism ($\mathfrak{W}_q(h): \mathfrak{W}_q(\mathbf{A}_2)\to \mathfrak{W}_q(\mathbf{A}_1)$ is a DqRA-space morphism).
\item Let $f: \mathbb{W}_1 \to \mathbb{W}_2$ be a DInFL-space (DqRA-space) morphism. Then $\mathfrak{A}(f): \mathfrak{A}(\mathbb{W}_2) \to \mathfrak{A}(\mathbb{W}_1)$ is a DInFL-algebra homomorphism  ($\mathfrak{A}_q(h): \mathfrak{A}_q(\mathbb{W}_2)\to \mathfrak{A}_q(\mathbb{W}_1)$ is a DqRA homomorphism).    
\end{enumerate}
\end{theorem}

\section{Representable DInFL-algebras and DqRAs}
\label{sec:RDqRAs}

One of the most captivating problems within the study of relation algebras is determining which abstract algebras are \emph{representable}~\cite{JT1948}. Recent work by two of the current  authors~\cite{RDqRA25} developed a definition of representability for DqRAs (and DInFL-algebras). Here we briefly recall the basic set-up for constructing concrete DqRAs and DInFL-algebras. While representable relation algebras are embeddable in relation algebras built from equivalence relations (or sets), for DInFL-algebras and DqRAs, the concrete algebras are built from partially ordered equivalence relations (or posets). 

Below we set up some notation and recall some basic facts about binary relations. 
For 
$R \subseteq X^2$ 
its converse is 
$R^\smallsmile = \left\{\left(x, y\right) \mid \left(y, x\right) \in R\right\}$. When considering a binary relation, $R$, it will usually be contained in an equivalence relation $E$ and therefore
the complement of $R$, denoted  $R^c$, will mean $R^c=\{\, (x,y) \in E \mid (x,y) \notin R\,\}$. If an equivalence relation is not specified, it should be taken as $X^2$. 
 For binary relations $R$ and $S$, their composition is $R \mathbin{;} S = \left\{\left(x, y\right) \mid \left(\exists z \in X\right)\left(\left(x, z\right) \in R \textnormal{ and } \left(z, y\right) \in S\right)\right\}$.

We define $\mathrm{id}_X=\{\,(x,x)\mid x \in X\,\}$.
The following familiar equivalences will be used often. 
If $E$ is an equivalence relation on $X$ and $R,S,T \subseteq E$ then we have $(R^{\smallsmile})^{\smallsmile}=R$, $(R^{\smallsmile})^c=(R^c)^{\smallsmile}$, and $\mathrm{id}_X\mathbin{;} R = R\mathbin{;} \mathrm{id}_X = R$. Further, $\left(R\, ; S\right)\mathbin{;} T = R\mathbin{;} \left(S\mathbin{;} T\right)$ and 
$\left(R\mathbin{;} S\right)^\smallsmile = S^\smallsmile\mathbin{;} R^\smallsmile$.

Given a function $\gamma: X \to X$, we use  $\gamma$ to denote either the function or the binary relation that is its graph. We emphasize that the lemma below applies when $\gamma$ is a bijective function $\gamma : X \to X$.
We remind the reader that $R^c$ is the complement of $R$ in $E$.

\begin{lemma}
[{\cite[Lemma 3.4]{RDqRA25}}]
\label{lem:important_eq_injective_map}
Let $E$ be an equivalence relation on a set $X$, and let $R, \gamma \subseteq E$. If 
$\gamma$ satisfies 
$\gamma^{\smallsmile}\mathbin{;} \gamma = \mathrm{id}_X$ and $\gamma \mathbin{;} \gamma^{\smallsmile}=\mathrm{id}_X$
then the following hold:
\begin{enumerate}[label=(\roman*)]
\item $\left(\gamma\mathbin{;} R\right)^c = \gamma \mathbin{;} R^c$
\item $\left(R\mathbin{;} \gamma\right)^c = R^c \mathbin{;} \gamma$.
\end{enumerate}	
\end{lemma}

In order to construct a DqRA of binary relations, consider a poset $\mathbf X = \left(X, \leqslant\right)$ and   $E$ an equivalence relation on $X$ with ${\leqslant} \subseteq E$. The equivalence relation $E$ is then partially ordered for all $(u, v), (x, y) \in E$ as follows: $(u, v) \preccurlyeq (x, y)  \textnormal{ iff } x \leqslant u \textnormal{ and } v\leqslant y$. 
From the poset 
$\mathbf E = \left(E, \preccurlyeq\right)$, we  consider the set of upsets, $\mathsf{Up}\left(\mathbf E\right)$, which forms a distributive lattice under the inclusion order. 

If $R,S \in \mathsf{Up}(\mathbf E)$, then $R\mathbin{;}S \in \mathsf{Up}(\mathbf E)$. The set of downsets of $(E,\preccurlyeq)$ is denoted $\mathsf{Down}(\mathbf{E})$ and is also closed under composition. Importantly, $R \in \mathsf{Up} (\mathbf{E})$ iff $R^c \in \mathsf{Down}(\mathbf{E})$ iff $R^{\smallsmile} \in \mathsf{Down}(\mathbf{E})$.
We have 
${\leqslant}\in \mathsf{Up}(\mathbf{E})$ and it is the identity with respect to~$\mathbin{;}$. 
Further, the operation of relational composition $\mathbin{;}$ is residuated with  
$R \backslash_{\mathsf{Up}(\mathbf E)} S = (R^{\smallsmile}\mathbin{;}S^c)^c$ and 
$R/_{\mathsf{Up}(\mathbf E)}S=(R^c\mathbin{;}S^{\smallsmile})^c$.
Hence $\left( \mathsf{Up}(\mathbf{E}), \cap, \cup, \mathbin{;},\backslash_{\mathsf{Up}(\mathbf E)},/_{\mathsf{Up}(\mathbf E)}, \leqslant \right)$ is a distributive residuated lattice.

The remaining DqRA operations in the construction (see  Theorem~\ref{thm:Dq(E)} below)  require an order automorphism $\alpha : X \to X$ and a dual order automorphism $\beta : X \to X$. 
The underlying lattice structure of the  construction from~\cite{RDqRA25} 
is the same distributive lattice obtained by Galatos and Jipsen~\cite{GJ2020, GJ2020a}, and  by Jipsen and \v{S}emrl~\cite{JS2023}. 
The additional requirement of an order automorphism $\alpha$ allows for the representation of \emph{odd} DInFL-algebras (i.e. $1=0$). The dual order automorphism $\beta$ is used to define $\neg$, and forces the poset to be self-dual.

We note briefly some facts regarding set-theoretic operations on $\mathsf{Up}(\mathbf{E})$; details can be found in~\cite[Lemma 3.5]{RDqRA25}.
If $R, S  \in \mathsf{Up}\left(\mathbf E\right)$, then 
$\alpha\mathbin{;} R$, and  $R\mathbin{;} \alpha$ are elements of $\mathsf{Up}\left(\mathbf E\right)$.
If $R  \in \mathsf{Down}\left(\mathbf E\right)$, then 
$\beta\mathbin{;}R\mathbin{;}\beta \in \mathsf{Up}(\mathbf E)$.

For $\alpha : X \to X$ an order automorphism,
we let $0= \alpha \mathbin{;} {\leqslant^{c\smallsmile}}$. With this we can define ${\sim}$ and $-$ 
on $\mathsf{Up}(\mathbf{E})$. 
By \cite[Lemma 3.10]{RDqRA25} we can in fact define the linear negations \emph{without} using the residuals. We get 
${\sim} R = R^{c\smallsmile}\mathbin{;} \alpha$ and $-R = \alpha \mathbin{;} R^{c\smallsmile}$ for all $R \in \mathsf{Up}\left(\mathbf{E}\right)$. For calculations, it is much easier to use these definitions of ${\sim}R$ and $-R$ than those involving the residuals. 

Lastly, $\neg R$ is defined using the dual order automorphism $\beta$, as stated below.\\ 

\begin{theorem}[{\cite[Theorems 3.12 and 3.15]{RDqRA25}}]\label{thm:Dq(E)}
Let $\mathbf{X}=\left(X,\leqslant\right)$ be a poset and $E$ an equivalence relation on $X$ such that ${\leqslant} \subseteq E$.  Let $\alpha: X \to X$ be an order automorphism of $\mathbf X$
with $\alpha \subseteq E$. 
For $R \in \mathsf{Up}(\mathbf E)$, define
${\sim} R = R^{c\smallsmile}\mathbin{;} \alpha$ and $-R = \alpha \mathbin{;} R^{c\smallsmile}$. 
\begin{enumerate}[label=(\roman*)]
\item The algebra $\left( \mathsf{Up}\left(\mathbf{E}\right),\cap, \cup, \mathbin{;}, \leqslant, {\sim}, {-} \right)  $ is a DInFL-algebra.
\item The algebra in (i) is cyclic if and only if $\alpha$ is the identity map.
\item For $\beta : X \to X$ a self-inverse dual order automorphism of $\mathbf{X}$ with $\beta \subseteq E$ and $\beta = \alpha \mathbin{;} \beta \mathbin{;} \alpha$, define  $\neg R=\alpha \mathbin{;} \beta \mathbin{;} R^c \mathbin{;} \beta$ for $R \in \mathsf{Up}(\mathbf{E})$. Then 
the algebra $\mathbf{Dq}(\mathbf E) = \left( \mathsf{Up}\left(\mathbf{E}\right),\cap, \cup, \mathbin{;}, \leqslant, {\sim}, {-}, {\neg}  \right)  $ 
is a distributive quasi relation algebra. 
\end{enumerate} 
\end{theorem}	

An algebra of the form described by Theorem~\ref{thm:Dq(E)} is called an \emph{equivalence distributive quasi relation algebra} and the class of such algebras is  denoted $\mathsf{EDqRA}$. If $E=X^2$, then we refer to the algebra $\mathbf{Dq}(\mathbf{E})$ as a \emph{full distributive quasi relation algebra}, with the class of such algebras denoted by $\mathsf{FDqRA}$. Analogous to the case for relation algebras (cf.~\cite[Chapter 3]{Mad2006}), it was shown  that 
$\mathbb{IP}(\mathsf{FDqRA})=\mathbb{I}(\mathsf{EDqRA})$~\cite[Theorem 4.4]{RDqRA25}. This gives rise to the definition below: 
\begin{definition}[{\cite[Definition 4.5]{RDqRA25}}]\label{def:RDqRA}
A DqRA $\mathbf{A} = \left(  A, \wedge, \vee, \cdot, 1, {\sim}, {-}, {\neg}\right)$ 
is \emph{representable} if 
$\mathbf{A} \in \mathbb{ISP}\left(\mathsf{FDqRA}\right)$
or, equivalently, $\mathbf{A} \in \mathbb{IS}\left(\mathsf{EDqRA}\right)$. 
\end{definition}

We will denote the class of representable DqRAs by $\mathsf{RDqRA}$ and the class of representable DInFL-algebras by $\mathsf{RDInFL}$. 
We say that a DqRA $\mathbf{A}$ is \emph{finitely} representable if the poset $( X ,\leqslant)$ used in the representation of $\mathbf{A}$ is finite. In Table~\ref{tab:Rep-DqRA<=5} and
Table~\ref{tab:Rep-DqRA=6}, we record what is known about the representability of finite algebras of size less than or equal to 6. We make frequent use of a result on finite representability~\cite[Theorem 5.12]{RDqRA25} which says that if a DInFL-algebra or DqRA $\mathbf{A}$ has an element $a$ with $0<a<1$ and $a^2\leqslant 0$, then $\mathbf{A}$ is not finitely representable. Notice that the result in~\cite{RDqRA25} is stated for DqRAs but only relies on properties of $\alpha$  and the fact that $0=\alpha\mathbin{;} (\leqslant^c)^\smallsmile$ and hence is also applicable to DInFL-algebras.

Another important representability result that we make use of in our tables involves so-called generalised ordinal sums of DqRAs~\cite{CMR2025}. Roughly stated, the construction $\mathbf{K}[\mathbf{L}]$ described there inserts the algebra $\mathbf{L}$ in place of the element $1_\mathbf{K} \in K$. An important consequence~\cite[Theorem 4.1]{CMR2025} of the construction states that if a DqRA is representable, then the same DqRA with a new top and bottom element will also be representable.

\section{Applications to  finite models and representability}\label{sec:fin-models}

In this section we combine our work on frames (Section~\ref{sec:frames}) with the concept of representability recalled in the previous section.

\subsection{Counting DInFL-frames and DqRA-frames}

The DInFL-frames and DqRA-frames defined in Section~\ref{sec:frames} can be used to determine the number of algebras with eight elements or less for both DInFL-algebras and DqRAs. In addition, using the frames allows us to classify the algebras by the structure of their underlying lattices.  

Amongst the DInFL-algebras, there is only one seven-element algebra that is non-cyclic, and only one eight-element algebra that is non-cyclic. Their dual frames are $\bowtie$ and $\mathbf{X}$, respectively. For DqRAs, there are two seven-element non-cyclic algebras (both with $\bowtie$ as their dual frame), and two eight-element non-cyclic algebras (in both cases $\mathbf{X}$ is the dual frame).

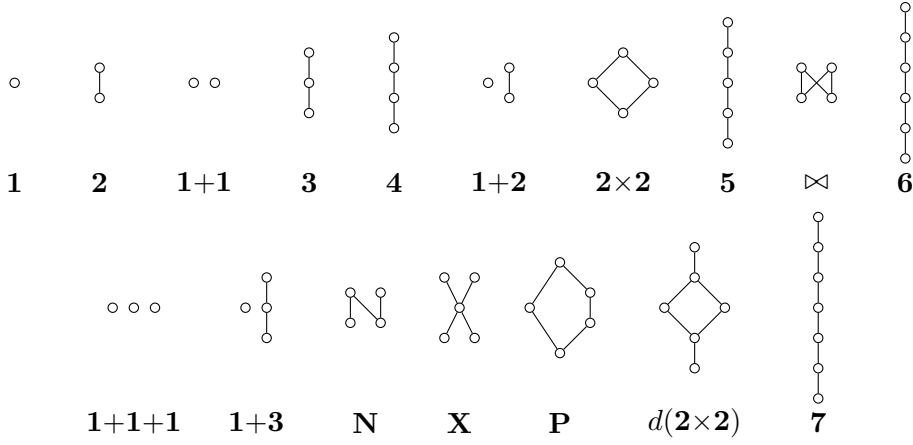
\begin{figure}[ht]
\tikzstyle{every picture} = [scale=.4]
\tikzstyle{every node} = [draw, fill=white, circle, inner sep=0pt, minimum size=3.5pt]
\begin{center}
\begin{tikzpicture}[baseline=0pt]
\node at (0,-.8)[n]{$\m 1$};
\node(0) at (0,2.5){};
\end{tikzpicture}
\quad
\begin{tikzpicture}[baseline=0pt]
\node at (0,-.8)[n]{$\m 2$};
\node(1) at (0,3){};
\node(0) at (0,2){} edge (1);
\end{tikzpicture}
\quad
\begin{tikzpicture}[baseline=0pt]
\node at (0.35,-.8)[n]{$\m{1{+}1}$};
\node(2) at (0,2.5){};
\node(0) at (0.7,2.5){};
\end{tikzpicture}
\quad
\begin{tikzpicture}[baseline=0pt]
\node at (0,-.8)[n]{$\m 3$};
\node(2) at (0,3.5){};
\node(1) at (0,2.5){} edge (2);
\node(0) at (0,1.5){} edge (1);
\end{tikzpicture}
\quad
\begin{tikzpicture}[baseline=0pt]
\node at (0,-.8)[n]{$\m 4$};
\node(3) at (0,4){};
\node(2) at (0,3){} edge (3);
\node(1) at (0,2){} edge (2);
\node(0) at (0,1){} edge (1);
\end{tikzpicture}
\quad
\begin{tikzpicture}[baseline=0pt]
\node at (0.35,-.8)[n]{$\m{1{+}2}$};
\node(2) at (0,2.5){};
\node(1) at (0.7,3){};
\node(0) at (0.7,2){} edge (1);
\end{tikzpicture}
\quad
\begin{tikzpicture}[baseline=0pt]
\node at (0,-.8)[n]{$\m{2{\times}2}$};
\node(3) at (0,3.5){};
\node(2) at (1,2.5){} edge (3);
\node(1) at (-1,2.5){} edge (3);
\node(0) at (0,1.5){} edge (1) edge (2);
\end{tikzpicture}
\quad
\begin{tikzpicture}[baseline=0pt]
\node at (0,-.8)[n]{$\m 5$};
\node(4) at (0,4.5){};
\node(3) at (0,3.5){} edge (4);
\node(2) at (0,2.5){} edge (3);
\node(1) at (0,1.5){} edge (2);
\node(0) at (0,0.5){} edge (1);
\end{tikzpicture}
\quad
\begin{tikzpicture}[baseline=0pt]
\node at (0.5,-.8)[n]{$\bowtie$};
\node(3) at (1,3){};
\node(2) at (0,3){};
\node(1) at (1,2){} edge (2) edge (3);
\node(0) at (0,2){} edge (2) edge (3);
\end{tikzpicture}
\quad
\begin{tikzpicture}[baseline=0pt]
\node at (0,-.8)[n]{$\m 6$};
\node(5) at (0,5){};
\node(4) at (0,4){} edge (5);
\node(3) at (0,3){} edge (4);
\node(2) at (0,2){} edge (3);
\node(1) at (0,1){} edge (2);
\node(0) at (0,0){} edge (1);
\end{tikzpicture}

\begin{tikzpicture}[baseline=0pt]
\node at (0.7,-.8)[n]{$\m{1{+}1{+}1}$};
\node(3) at (0,3){};
\node(2) at (0.7,3){};
\node(0) at (1.4,3){};
\end{tikzpicture}
\
\begin{tikzpicture}[baseline=0pt]
\node at (0.35,-.8)[n]{$\m{1{+}3}$};
\node(3) at (0,3){};
\node(2) at (0.7,4){};
\node(1) at (0.7,3){} edge (2);
\node(0) at (0.7,2){} edge (1);
\end{tikzpicture}
\quad
\begin{tikzpicture}[baseline=0pt]
\node at (0.5,-.8)[n]{$\m N$};
\node(3) at (1,3.5){};
\node(2) at (0,3.5){};
\node(1) at (1,2.5){} edge (2) edge (3);
\node(0) at (0,2.5){} edge (2);
\end{tikzpicture}
\quad
\begin{tikzpicture}[baseline=0pt]
\node at (0.5,-.8)[n]{$\m X$};
\node(4) at (1,4){};
\node(3) at (0,4){};
\node(2) at (0.5,3){} edge (3) edge (4);
\node(1) at (1,2){} edge (2);
\node(0) at (0,2){} edge (2);
\end{tikzpicture}
\quad
\begin{tikzpicture}[baseline=0pt]
\node at (0,-.8)[n]{$\m P$};
\node(4) at (0,4.5){};
\node(3) at (1,3.5){} edge (4);
\node(2) at (-1,3){} edge (4);
\node(1) at (1,2.5){} edge (3);
\node(0) at (0,1.5){} edge (1) edge (2);
\end{tikzpicture}
\quad
\begin{tikzpicture}[baseline=0pt]
\node at (0,-.8)[n]{$d(\m{2{\times}2})$};
\node(5) at (0,5){};
\node(4) at (0,4){} edge (5);
\node(3) at (1,3){} edge (4);
\node(2) at (-1,3){} edge (4);
\node(1) at (0,2){} edge (2) edge (3);
\node(0) at (0,1){} edge (1);
\end{tikzpicture}
\quad
\begin{tikzpicture}[baseline=0pt]
\node at (0,-.8)[n]{$\m 7$};
\node(6) at (0,6){};
\node(5) at (0,5){} edge (6);
\node(4) at (0,4){} edge (5);
\node(3) at (0,3){} edge (4);
\node(2) at (0,2){} edge (3);
\node(1) at (0,1){} edge (2);
\node(0) at (0,0){} edge (1);
\end{tikzpicture}
\vspace{-10pt}

\caption{Posets of join-irreducibles for self-dual distributive lattices of cardinality $\le 8$}
\end{center}
\end{figure}
\begin{table}[h]
    \centering
    \tabcolsep3pt
\begin{tabular}{|c|c|c|cc|c|ccc|cc|ccccccc|}\hline
Poset&$\m 1$&$\m 2$&$\m{1{+}1}$&$\m 3$&$\m 4$&$\m{1{+}2}$&$\m{2{\times}2}$&$\m 5$&$\bowtie$&$\m 6$&$\m{1{+}1{+}1}$&$\m{1{+}3}$&$\m N$&$\m X$&$\m P$&$d(\m{2{\times}2})$&$\m 7$\\ \hline
DInFL-frames& 1& 2&5 &4 &8 &10 &16 &17 &11 &38 &25 &25 &22 & 21& 28&70 &91 \\
DqRA-frames& 1& 2& 6& 4& 8& 10& 23& 17& 12& 36& 31& 25& 22& 23& 26& 106& 81\\
\hline
\end{tabular}

\ 

\ 
\tabcolsep9pt
\begin{tabular}{|c|c|c|c|c|c|c|c|c|}\hline
        Size of algebra &1&2&3&4&5&6&7&8  \\ \hline 
        Number of DInFL-algebras &1& 1& 2&9 &8 & 43& 49&282 \\
        Number of DqRAs &1& 1& 2& 10& 8& 50& 48& 314 \\\hline
\end{tabular}

\

\caption{Number of nonisomorphic quasi relation algebras of cardinality up to 8}
\label{tab:nqRA}
\end{table}

The Python code used to generate these models can be found in the Jupyter notebook at \url{https://github.com/jipsen/Distributive-quasi-relation-algebras-and-DInFL}. There we have also posted an indexed list of all DqRAs and DInFL-algebras up to size 8. We have included algebras up to size 6 later in this section.

\subsection{Representations for some small quasi relation algebras}\label{sec:reps-term-subreducts}

In this subsection we use DqRA-frames to describe the DqRAs that are ${\{\vee,\cdot,1,\sim\}}$-subreducts of small relation algebras. We refer to these algebras as DqRA-subreducts.
Moreover, we only present information on DqRAs that are not relation algebras.

\begin{table} 
\scriptsize\tabcolsep8pt\arraycolsep3pt
\begin{tabular}{lllll}
$\begin{array}{|c|ccc|}\hline
\cdot_{1}&\ \;a\;\ &\ \;r\;\ &\ \;s\;\ \\\hline
\;a\;&1	&r	&s	\\
r&r&r	&\top	\\
s&s&\top	&s	\\\hline
\end{array}$
&
$\begin{array}{|c|ccc|}\hline
\cdot_{2}&\ \;a\;\ &\ \;r\;\ &\ \;s\;\ \\\hline
\;a\;&1a	&r	&s	\\
r &r	&r	&\top	\\
s &s	&\top	&s	\\\hline
\end{array}$
&
$\begin{array}{|c|ccc|}\hline
\cdot_{3}&\ \;a\;\ &\ \;r\;\ &\ \;s\;\ \\\hline
\;a\;&1	&r	&s	\\
r &r	&s	&1a	\\
s &s	&1a	&r	\\\hline
\end{array}$
&
$\begin{array}{|c|ccc|}\hline
\cdot_{4}&\ \;a\;\ &\ \;r\;\ &\ \;s\;\ \\\hline
\;a\;&1a	&r	&s	\\
r &r	&s	&1a	\\
s &s	&1a	&r	\\\hline
\end{array}$
&
$\begin{array}{|c|ccc|}\hline
\cdot_{5}&\ \;a\;\ &\ \;r\;\ &\ \;s\;\ \\\hline
\;a\;&1	&r	&s	\\
r &r	&rs	&\top	\\
s &s	&\top	&rs	\\\hline
\end{array}$

\\[20pt]

$\begin{array}{|c|ccc|}\hline
\cdot_{6}&\ \;a\;\ &\ \;r\;\ &\ \;s\;\ \\\hline
\;a\;&1a	&r	&s	\\
r &r	&rs	&\top	\\
s &s	&\top	&rs	\\\hline
\end{array}$
&
$\begin{array}{|c|ccc|}\hline
\cdot_{7}&\ \;a\;\ &\ \;r\;\ &\ \;s\;\ \\\hline
\;a\;&1rs	&a	&a	\\
r &a	&r	&1rs	\\
s &a	&1rs	&s	\\\hline
\end{array}$
&
$\begin{array}{|c|ccc|}\hline
\cdot_{8}&\ \;a\;\ &\ \;r\;\ &\ \;s\;\ \\\hline
\;a\;&\top	&a	&a	\\
r &a	&r	&1rs	\\
s &a	&1rs	&s	\\\hline
\end{array}$
&
$\begin{array}{|c|ccc|}\hline
\cdot_{9}&\ \;a\;\ &\ \;r\;\ &\ \;s\;\ \\\hline
\;a\;&1rs	&a	&a	\\
r &a	&s	&1	\\
s &a	&1	&r	\\\hline
\end{array}$
&
$\begin{array}{|c|ccc|}\hline
\cdot_{10}&\ \;a\;\ &\ \;r\;\ &\ \;s\;\ \\\hline
a &\top	&a	&a	\\
r &a	&s	&1	\\
s &a	&1	&r	\\\hline
\end{array}$

\\[20pt]

$\begin{array}{|c|ccc|}\hline
\cdot_{11}&\ \;a\;\ &\ \;r\;\ &\ \;s\;\ \\\hline
a &1rs	&a	&a	\\
r &a	&rs	&1rs	\\
s &a	&1rs	&rs	\\\hline
\end{array}$
&
$\begin{array}{|c|ccc|}\hline
\cdot_{12}&\ \;a\;\ &\ \;r\;\ &\ \;s\;\ \\\hline
a &\top	&a	&a	\\
r &a	&rs	&1rs	\\
s &a	&1rs	&rs	\\\hline
\end{array}$
&
$\begin{array}{|c|ccc|}\hline
\cdot_{13}&\ \;a\;\ &\ \;r\;\ &\ \;s\;\ \\\hline
a &\top	&ar	&a	\\
r &a	&r	&\top	\\
s &as	&1rs	&s	\\\hline
\end{array}$
&
$\begin{array}{|c|ccc|}\hline
\cdot_{14}&\ \;a\;\ &\ \;r\;\ &\ \;s\;\ \\\hline
a &1rs	&ar	&as	\\
r &ar	&r	&\top	\\
s &as	&\top	&s	\\\hline
\end{array}$
&
$\begin{array}{|c|ccc|}\hline
\cdot_{15}&\ \;a\;\ &\ \;r\;\ &\ \;s\;\ \\\hline
a &\top	&ar	&as	\\
r &ar	&r	&\top	\\
s &as	&\top	&s	\\\hline
\end{array}$

\\[20pt]

$\begin{array}{|c|ccc|}\hline
\cdot_{16}&\ \;a\;\ &\ \;r\;\ &\ \;s\;\ \\\hline
a &1rs	&ar	&as	\\
r &ar	&rs	&\top	\\
s &as	&\top	&rs	\\\hline
\end{array}$
&
$\begin{array}{|c|ccc|}\hline
\cdot_{17}&\ \;a\;\ &\ \;r\;\ &\ \;s\;\ \\\hline
a &\top	&ar	&as	\\
r &ar	&rs	&\top	\\
s &as	&\top	&rs	\\\hline
\end{array}$
&
$\begin{array}{|c|ccc|}\hline
\cdot_{18}&\ \;a\;\ &\ \;r\;\ &\ \;s\;\ \\\hline
a &1	&s	&r	\\
r &s	&a	&1	\\
s &r	&1	&a	\\\hline
\end{array}$
&
$\begin{array}{|c|ccc|}\hline
\cdot_{19}&\ \;a\;\ &\ \;r\;\ &\ \;s\;\ \\\hline
a &1	&s	&r	\\
r &s	&ars	&1rs	\\
s &r	&1rs	&ars	\\\hline
\end{array}$
&
$\begin{array}{|c|ccc|}\hline
\cdot_{20}&\ \;a\;\ &\ \;r\;\ &\ \;s\;\ \\\hline
a &1a	&rs	&rs	\\
r &rs	&a	&1a	\\
s &rs	&1a	&a	\\\hline
\end{array}$

\\[20pt]

$\begin{array}{|c|ccc|}\hline
\cdot_{21}&\ \;a\;\ &\ \;r\;\ &\ \;s\;\ \\\hline
a &1a	&rs	&rs	\\
r &rs	&ar	&\top	\\
s &rs	&\top	&as	\\\hline
\end{array}$
&
$\begin{array}{|c|ccc|}\hline
\cdot_{22}&\ \;a\;\ &\ \;r\;\ &\ \;s\;\ \\\hline
a &1a	&rs	&rs	\\
r &rs	&ars	&\top	\\
s &rs	&\top	&ars	\\\hline
\end{array}$
&
$\begin{array}{|c|ccc|}\hline
\cdot_{23}&\ \;a\;\ &\ \;r\;\ &\ \;s\;\ \\\hline
a &1rs	&as	&ar	\\
r &as	&ar	&1rs	\\
s &ar	&1rs	&as	\\\hline
\end{array}$
&
$\begin{array}{|c|ccc|}\hline
\cdot_{24}&\ \;a\;\ &\ \;r\;\ &\ \;s\;\ \\\hline
a &\top	&as	&ar	\\
r &as	&ar	&1rs	\\
s &ar	&1rs	&as	\\\hline
\end{array}$
&
$\begin{array}{|c|ccc|}\hline
\cdot_{25}&\ \;a\;\ &\ \;r\;\ &\ \;s\;\ \\\hline
a &1rs	&as	&ar	\\
r &as	&ars	&1rs	\\
s &ar	&1rs	&ars	\\\hline
\end{array}$

\\[20pt]

$\begin{array}{|c|ccc|}\hline
\cdot_{26}&\ \;a\;\ &\ \;r\;\ &\ \;s\;\ \\\hline
a &\top	&as	&ar	\\
r &as	&ars	&1rs	\\
s &ar	&1rs	&ars	\\\hline
\end{array}$
&
$\begin{array}{|c|ccc|}\hline
\cdot_{27}&\ \;a\;\ &\ \;r\;\ &\ \;s\;\ \\\hline
a &1rs	&ars	&ar	\\
r &as	&ar	&\top	\\
s &ars	&1rs	&as	\\\hline
\end{array}$
&
$\begin{array}{|c|ccc|}\hline
\cdot_{28}&\ \;a\;\ &\ \;r\;\ &\ \;s\;\ \\\hline
a &\top	&ars	&ar	\\
r &as	&ar	&\top	\\
s &ars	&1rs	&as	\\\hline
\end{array}$
&
$\begin{array}{|c|ccc|}\hline
\cdot_{29}&\ \;a\;\ &\ \;r\;\ &\ \;s\;\ \\\hline
a &1rs	&ars	&ar	\\
r &as	&ars	&\top	\\
s &ars	&1rs	&ars	\\\hline
\end{array}$
&
$\begin{array}{|c|ccc|}\hline
\cdot_{30}&\ \;a\;\ &\ \;r\;\ &\ \;s\;\ \\\hline
a &\top	&ars	&ar	\\
r &as	&ars	&\top	\\
s &ars	&1rs	&ars	\\\hline
\end{array}$

\\[20pt]

$\begin{array}{|c|ccc|}\hline
\cdot_{31}&\ \;a\;\ &\ \;r\;\ &\ \;s\;\ \\\hline
a &\top	&ars	&ars	\\
r &ars	&a	&1a	\\
s &ars	&1a	&a	\\\hline
\end{array}$
&
$\begin{array}{|c|ccc|}\hline
\cdot_{32}&\ \;a\;\ &\ \;r\;\ &\ \;s\;\ \\\hline
a &1rs	&ars	&ars	\\
r &ars	&ar	&\top	\\
s &ars	&\top	&as	\\\hline
\end{array}$
&
$\begin{array}{|c|ccc|}\hline
\cdot_{33}&\ \;a\;\ &\ \;r\;\ &\ \;s\;\ \\\hline
a &\top	&ars	&ars	\\
r &ars	&ar	&\top	\\
s &ars	&\top	&as	\\\hline
\end{array}$
&
$\begin{array}{|c|ccc|}\hline
\cdot_{34}&\ \;a\;\ &\ \;r\;\ &\ \;s\;\ \\\hline
a &1rs	&ars	&ars	\\
r &ars	&as	&1a	\\
s &ars	&1a	&ar	\\\hline
\end{array}$
&
$\begin{array}{|c|ccc|}\hline
\cdot_{35}&\ \;a\;\ &\ \;r\;\ &\ \;s\;\ \\\hline
a &\top	&ars	&ars	\\
r &ars	&as	&1a	\\
s &ars	&1a	&ar	\\\hline
\end{array}$

\\[20pt]

$\begin{array}{|c|ccc|}\hline
\cdot_{36}&\ \;a\;\ &\ \;r\;\ &\ \;s\;\ \\\hline
a &1rs	&ars	&ars	\\
r &ars	&ars	&\top	\\
s &ars	&\top	&ars	\\\hline
\end{array}$
&
$\begin{array}{|c|ccc|}\hline
\cdot_{37}&\ \;a\;\ &\ \;r\;\ &\ \;s\;\ \\\hline
a &\top	&ars	&ars	\\
r &ars	&ars	&\top	\\
s &ars	&\top	&ars	\\\hline
\end{array}$&&&
\end{tabular}

\vskip10pt

\caption{Atom structures (= frames) for the 37 nonsymmetric RAs of cardinality 16. The identity atom 1 is not shown, a string of elements denotes the join of them, and ${\sim}a=1rs$, ${\sim}r=1ar$, ${\sim}s=1as$.}
\label{tab:atom-structures}

\vspace{-0.6cm}
\end{table}

We consider small relation algebras with up to 16 elements, with the aim to compute all $\{\vee,\cdot,{\sim}\}$-subreducts of these algebras that are proper quasi relation algebras. Since symmetric relation algebras satisfy $x=x^\smallsmile$, it follows that ${\sim}x=\neg x$, hence the subreducts of such algebras only produce relation algebras. This means we only need to consider nonsymmetric relation algebras. In Roger Maddux's book \cite{Mad2006}, there are lists of all finite integral relation algebras with up to 5 atoms. Recall that a relation algebra is \emph{integral} if the identity element is an atom. For nonsymmetric ones with up to 4 atoms, these relation algebras are denoted $1_3, 2_3, 3_3,1_{37}-37_{37}$.

We begin by recalling some information on the smallest nonrepresentable relation algebras. Lyndon \cite{Lyn56} showed that all relation algebras with 8 elements or less are representable, and
McKenzie~\cite{McK1966} found a 16-element relation algebra (now referred to as $14_{37}$ in Roger Maddux's list \cite{Mad2006}) that is nonrepresentable.
There are 10 further algebras in the list of 37 that are nonrepresentable: $16_{37}, 21_{37}, 24_{37}-29_{37}, 32_{37}, 34_{37}$. The representations of the remaining 26 relation algebras were found by Steven Comer and Roger Maddux (see \cite{Mad2006}).

We now describe the maximal subreducts of these 37 relation algebras that are proper qRAs.
When they occur as a subreduct of a representable relation algebra, they are themselves representable (indicated by bold names below). The representation will use the same set $X$ and equivalence relation $E$ as the relation algebra. The additional structure required by Theorem~\ref{thm:Dq(E)} will be ${\leqslant} = \alpha = \beta = \mathrm{id}_X$. 

The frames for the first type of subreducts are based on the poset $1{+}1{+}2$, and there are 
20 (nonisomorphic) frames of this kind. The corresponding DqRAs have 12 elements forming the lattice $\m2{\times}\m2{\times}\m3$ and occur as subreducts of the relation algebras 
$\m1_{37}$, $\m2_{37}$, $\m5_{37}$, $\m6_{37}$, $\m7_{37}$, $\m8_{37}$, $\m{11}_{37}$, $\m{12}_{37}$, $14_{37}$, $\m{15}_{37}$,
$16_{37}$, $\m{17}_{37}$, $\m{20}_{37}$, $21_{37}$, $\m{22}_{37}$, $\m{31}_{37}$, $32_{37}$, $\m{33}_{37}$, $\m{36}_{36}$, $\m{37}_{37}$.

In each case the 2-element chain in the frame is given by $s\prec r$ (or isomorphically by $r\prec s$). To see that this frame corresponds to a subreduct of the listed relation algebras, it suffices to check that $s\leqslant x\cdot y\implies r\leqslant x\cdot y$ for all $x,y\in \{a,r,r\vee s\}$, while this formula fails for the other 17 nonsymmetric integral relation algebras with 16 elements (see Table~\ref{tab:atom-structures}). Hence there are at least 16 representable DqRA with poset $1{+}1{+}2$ as their frame. Using the representation game in \cite{JS2023} it has been checked that the DqRA-subreduct of McKenzie's algebra $14_{37}$ is not representable as a weakening relation algebra. For this algebra and $16_{37}$, $21_{37}$, $32_{37}$ it has not (yet) been determined if the DqRA-subreduct is representable.

Ten of the remaining 17 relation algebras in the list have a maximal DqRA-subreduct with $1{+}3$ as poset: $\m{13}_{37}$, $\m{19}_{37}$, $ \m{23}_{37}$, $ 24_{37}$, $ 25_{37}$, $ 26_{37}$, $ 27_{37}$, $ 28_{37}$, $ 29_{37}$, $ \m{30}_{37}.$

In this case, the poset of the frame satisfies $s\prec a\prec r$ (or isomorphically $r\prec a\prec s$), and such a frame corresponds to an 8-element subreduct of a relation algebras if it satisfies $s\leqslant x\cdot y\implies a\vee r\leqslant x\cdot y$ and $a\leqslant x\cdot y\implies r\leqslant x\cdot y$ for all $x,y\in \{r,a\vee r,a\vee r\vee s\}$.

Note that the algebras $\m{13}_{37}, 27_{37}-\m{30}_{37}$ are noncommutative, but the DqRA-subreducts are commutative, hence they can be expanded to DqRAs. Four of the relation algebras in this list are representable, but the DqRA-subreducts of $\m{19}_{37}$ and $\m{30}_{37}$ are isomorphic, so this gives representations for three 8-element DqRAs. Other representable 8-element DqRAs can be found as subalgebras of the sixteen 12-element DqRAs described above.

Finally, 7 algebras do not have subreducts that produce proper quasi relation algebras: $\m{3}_{37}, \m{4}_{37}, \m{9}_{37}$, $\m{10}_{37}, \m{18}_{37}, 34_{37}, \m{35}_{37}$.

\subsection{List of DInFL-algebras and DqRAs up to cardinality six}\label{sec:list-of-algebras}

There are $1+1+2+9+8+43=64$ distributive involutive residuated lattices with $\le 6$ elements.
In the table below, each algebra is named $D^{n}_{m,i(,k)}$ where $n$ is the 
cardinality and $m$ enumerates nonisomorphic involutive lattices of size $n$, in order
of decreasing height (an \emph{involutive lattice} is a lattice with order isomorphisms $\sim,-$ that are inverses of each other). The index $i$ enumerates nonisomorphic algebras with the same involutive lattice
reduct. The value $k$ indicates the number of nonisomorphic DqRA algebras that have this DInRL as reduct (by default, $k=1$ is not shown). For $k=2$, if the algebra is commutative, $\neg x={\sim}x=-x$ is one of the DqRAs,
and the other one uses an order-reversing isomorphism for $\neg x\ne {\sim}x$. 
The linear negations $\sim,-$ are determined by the element labeled $0$ (bottom for integral algebras). The elements of each algebra are denoted by one of four different symbols:
\begin{center}
\begin{tabular}{ll}
{\Large$\bullet$} = central idempotent
&{\Large$\circ$} = central nonidempotent\\
{\scriptsize$\blacksquare$} = noncentral idempotent
&{\scriptsize$\square$} = noncentral nonidempotent
\end{tabular}
\end{center}
Algebras with more central elements (round 
circles) are listed earlier, hence commutative algebras precede 
noncommutative ones. Finally, algebras are listed in decreasing order of 
number of idempotents (black nodes).

The monoid operation is indicated by labels. If a nonobvious product $xy$ 
is not listed, then it can be deduced from the given information: either
it follows from idempotence ($x^2=x$) indicated by a black node
or from commutativity or there are products $uv=wz$ such that
$u\leqslant x\leqslant w$ and $v\leqslant y\leqslant z$ (possibly $uv=\bot\bot$ or $wz=\top\top$).
\tikzstyle{every picture} = [scale=0.45]
\tikzstyle{every node} = [draw, fill=white, circle, inner sep=0pt, minimum size=5pt]
\tikzstyle{n} = [draw=none, rectangle, inner sep=0pt,minimum size=1pt] 
\tikzstyle{i} = [draw, fill=black, circle, inner sep=0pt, minimum size=5pt] 
\tikzstyle{e} = [draw=none, rectangle, inner sep=0pt]
\tikzstyle{nci} = [draw, fill=black, rectangle, inner sep=0pt, minimum size=5pt]
\tikzstyle{nc} = [draw, fill=white, rectangle, inner sep=0pt, minimum size=5pt]
\sloppy

\begin{tikzpicture}[baseline=0pt]
\node at (0,-1)[n]{$D^{1}_{1,1}$};
\node at (0,-1.5)[n]{\phantom{.}};
\node(0) at (0,0)[i,label=right:$ $]{};
\end{tikzpicture}
\ \ 
\begin{tikzpicture}[baseline=0pt]
\node at (0,-1)[n]{$D^{2}_{1,1}$};
\node at (0,-1.5)[n]{\phantom{.}};
\node(1) at (0,1)[i,label=right:$1$]{};
\node(0) at (0,0)[i,label=right:$ $]{} edge (1);
\end{tikzpicture}
\ \ 
\begin{tikzpicture}[baseline=0pt]
\node at (0,-1)[n]{$D^{3}_{1,1}$};
\node at (0,-1.5)[n]{\phantom{.}};
\node(2) at (0,2)[i,label=right:$1$]{};
\node(1) at (0,1)[,label=right:$a$]{} edge (2);
\node(0) at (0,0)[i,label=right:$ a^2$]{} edge (1);
\end{tikzpicture}
\ \ 
\begin{tikzpicture}[baseline=0pt]
\node at (0,-1)[n]{$D^{3}_{1,2}$};
\node at (0,-1.5)[n]{\phantom{.}};
\node(2) at (0,2)[i,label=right:$ $]{};
\node(1) at (0,1)[i,label=right:$1{=}0$]{} edge (2);
\node(0) at (0,0)[i,label=right:$ $]{} edge (1);
\end{tikzpicture}
\ \ 
\begin{tikzpicture}[baseline=0pt]
\node at (0,-1)[n]{$D^{4}_{1,1}$};
\node at (0,-1.5)[n]{\phantom{.}};
\node(3) at (0,3)[i,label=right:$1$]{};
\node(2) at (0,2)[i,label=right:$a$]{} edge (3);
\node(1) at (0,1)[,label=right:$b$]{} edge (2);
\node(0) at (0,0)[i,label=right:$ ab$]{} edge (1);
\end{tikzpicture}
\ \ 
\begin{tikzpicture}[baseline=0pt]
\node at (0,-1)[n]{$D^{4}_{1,2}$};
\node at (0,-1.5)[n]{\phantom{.}};
\node(3) at (0,3)[i,label=right:$1$]{};
\node(2) at (0,2)[,label=right:$a$]{} edge (3);
\node(1) at (0,1)[,label=right:$b{=}a^2$]{} edge (2);
\node(0) at (0,0)[i,label=right:$ ab$]{} edge (1);
\end{tikzpicture}
\ \ 
\begin{tikzpicture}[baseline=0pt]
\node at (0,-1)[n]{$D^{4}_{1,3}$};
\node at (0,-1.5)[n]{\phantom{.}};
\node(3) at (0,3)[i,label=right:$ \top 0$]{};
\node(2) at (0,2)[i,label=right:$1$]{} edge (3);
\node(1) at (0,1)[i,label=right:$0$]{} edge (2);
\node(0) at (0,0)[i,label=right:$ $]{} edge (1);
\end{tikzpicture}
\ \ 
\begin{tikzpicture}[baseline=0pt]
\node at (0,-1)[n]{$D^{4}_{1,4}$};
\node at (0,-1.5)[n]{\phantom{.}};
\node(3) at (0,3)[i,label=right:$ 0^2$]{};
\node(2) at (0,2)[,label=right:$0$]{} edge (3);
\node(1) at (0,1)[i,label=right:$1$]{} edge (2);
\node(0) at (0,0)[i,label=right:$ $]{} edge (1);
\end{tikzpicture}
\ \ 
\begin{tikzpicture}[baseline=0pt]
\node at (0,-1)[n]{$D^{4}_{2,1,2}$};
\node at (0,-1.5)[n]{\phantom{.}};
\node(3) at (0,2)[i,label=right:$1$]{};
\node(2) at (1,1)[i,label=right:$b$]{} edge (3);
\node(1) at (-1,1)[i,label=left:$a$]{} edge (3);
\node(0) at (0,0)[i,label=right:$ ab$]{} edge (1) edge (2);
\end{tikzpicture}
\ \ 
\begin{tikzpicture}[baseline=0pt]
\node at (0,-1)[n]{$D^{4}_{2,2}$};
\node at (0,-1.5)[n]{\phantom{.}};
\node(3) at (0,2)[i,label=right:$ \top 0$]{};
\node(2) at (1,1)[,label=right:$0$]{} edge (3);
\node(1) at (-1,1)[i,label=left:$1{=}0^2$]{} edge (3);
\node(0) at (0,0)[i,label=right:$ $]{} edge (1) edge (2);
\end{tikzpicture}

\ \ 

\begin{tikzpicture}[baseline=0pt]
\node at (0,-1)[n]{$D^{4}_{2,3}$};
\node at (0,-1.5)[n]{\phantom{.}};
\node(3) at (0,2)[i,label=right:$ 0^2$]{};
\node(2) at (1,1)[,label=right:$0$]{} edge (3);
\node(1) at (-1,1)[i,label=left:$1$]{} edge (3);
\node(0) at (0,0)[i,label=right:$ $]{} edge (1) edge (2);
\end{tikzpicture}
\ \ 
\begin{tikzpicture}[baseline=0pt]
\node at (0,-1)[n]{$D^{4}_{3,1}$};
\node at (0,-1.5)[n]{\phantom{.}};
\node(3) at (0,2)[i,label=right:$ $]{};
\node(2) at (1,1)[,label=right:$a{=}\top a$]{} edge (3);
\node(1) at (-1,1)[i,label=left:$1{=}0$]{} edge (3);
\node(0) at (0,0)[i,label=right:$ a^2$]{} edge (1) edge (2);
\end{tikzpicture}
\ \ 
\begin{tikzpicture}[baseline=0pt]
\node at (0,-1)[n]{$D^{4}_{3,2}$};
\node at (0,-1.5)[n]{\phantom{.}};
\node(3) at (0,2)[i,label=right:$ \top a$]{};
\node(2) at (1,1)[,label=right:$a$]{} edge (3);
\node(1) at (-1,1)[i,label=left:$1{=}0{=}a^2$]{} edge (3);
\node(0) at (0,0)[i,label=right:$ $]{} edge (1) edge (2);
\end{tikzpicture}
\ \ 
\begin{tikzpicture}[baseline=0pt]
\node at (0,-1)[n]{$D^{5}_{1,1}$};
\node at (0,-1.5)[n]{\phantom{.}};
\node(4) at (0,4)[i,label=right:$1$]{};
\node(3) at (0,3)[i,label=right:$a$]{} edge (4);
\node(2) at (0,2)[,label=right:$b{=}ab$]{} edge (3);
\node(1) at (0,1)[,label=right:$c$]{} edge (2);
\node(0) at (0,0)[i,label=right:$ b^2{=}ac$]{} edge (1);
\end{tikzpicture}
\ \ 
\begin{tikzpicture}[baseline=0pt]
\node at (0,-1)[n]{$D^{5}_{1,2}$};
\node at (0,-1.5)[n]{\phantom{.}};
\node(4) at (0,4)[i,label=right:$1$]{};
\node(3) at (0,3)[,label=right:$a$]{} edge (4);
\node(2) at (0,2)[,label=right:$b$]{} edge (3);
\node(1) at (0,1)[,label=right:$c{=}ab{=}a^2$]{} edge (2);
\node(0) at (0,0)[i,label=right:$ b^2{=}ac$]{} edge (1);
\end{tikzpicture}
\ \ 
\begin{tikzpicture}[baseline=0pt]
\node at (0,-1)[n]{$D^{5}_{1,3}$};
\node at (0,-1.5)[n]{\phantom{.}};
\node(4) at (0,4)[i,label=right:$1$]{};
\node(3) at (0,3)[,label=right:$a$]{} edge (4);
\node(2) at (0,2)[,label=right:$b{=}a^2$]{} edge (3);
\node(1) at (0,1)[,label=right:$c{=}ab$]{} edge (2);
\node(0) at (0,0)[i,label=right:$ b^2{=}ac$]{} edge (1);
\end{tikzpicture}

\ \ 

\begin{tikzpicture}[baseline=0pt]
\node at (0,-1)[n]{$D^{5}_{1,4}$};
\node at (0,-1.5)[n]{\phantom{.}};
\node(4) at (0,4)[i,label=right:$ \top 0$]{};
\node(3) at (0,3)[i,label=right:$1$]{} edge (4);
\node(2) at (0,2)[,label=right:$a$]{} edge (3);
\node(1) at (0,1)[i,label=right:$0{=}a^2$]{} edge (2);
\node(0) at (0,0)[i,label=right:$ $]{} edge (1);
\end{tikzpicture}
\ \ 
\begin{tikzpicture}[baseline=0pt]
\node at (0,-1)[n]{$D^{5}_{1,5}$};
\node at (0,-1.5)[n]{\phantom{.}};
\node(4) at (0,4)[i,label=right:$ $]{};
\node(3) at (0,3)[i,label=right:$1$]{} edge (4);
\node(2) at (0,2)[,label=right:$a{=}\top 0{=}\top a$]{} edge (3);
\node(1) at (0,1)[,label=right:$0$]{} edge (2);
\node(0) at (0,0)[i,label=right:$ a^2$]{} edge (1);
\end{tikzpicture}
\ \ 
\begin{tikzpicture}[baseline=0pt]
\node at (0,-1)[n]{$D^{5}_{1,6}$};
\node at (0,-1.5)[n]{\phantom{.}};
\node(4) at (0,4)[i,label=right:$ \top b$]{};
\node(3) at (0,3)[i,label=right:$a$]{} edge (4);
\node(2) at (0,2)[i,label=right:$1{=}0$]{} edge (3);
\node(1) at (0,1)[i,label=right:$b{=}ab$]{} edge (2);
\node(0) at (0,0)[i,label=right:$ $]{} edge (1);
\end{tikzpicture}
\ \ 
\begin{tikzpicture}[baseline=0pt]
\node at (0,-1)[n]{$D^{5}_{1,7}$};
\node at (0,-1.5)[n]{\phantom{.}};
\node(4) at (0,4)[i,label=right:$ 0a$]{};
\node(3) at (0,3)[,label=right:$0$]{} edge (4);
\node(2) at (0,2)[i,label=right:$a$]{} edge (3);
\node(1) at (0,1)[i,label=right:$1$]{} edge (2);
\node(0) at (0,0)[i,label=right:$ $]{} edge (1);
\end{tikzpicture}
\ \ 
\begin{tikzpicture}[baseline=0pt]
\node at (0,-1)[n]{$D^{5}_{1,8}$};
\node at (0,-1.5)[n]{\phantom{.}};
\node(4) at (0,4)[i,label=right:$ 0a$]{};
\node(3) at (0,3)[,label=right:$0{=}a^2$]{} edge (4);
\node(2) at (0,2)[,label=right:$a$]{} edge (3);
\node(1) at (0,1)[i,label=right:$1$]{} edge (2);
\node(0) at (0,0)[i,label=right:$ $]{} edge (1);
\end{tikzpicture}
\ \ 
\begin{tikzpicture}[baseline=0pt]
\node at (0,-1)[n]{$D^{6}_{1,1}$};
\node at (0,-1.5)[n]{\phantom{.}};
\node(5) at (0,5)[i,label=right:$1$]{};
\node(4) at (0,4)[i,label=right:$a$]{} edge (5);
\node(3) at (0,3)[i,label=right:$b$]{} edge (4);
\node(2) at (0,2)[,label=right:$c{=}ac$]{} edge (3);
\node(1) at (0,1)[,label=right:$d$]{} edge (2);
\node(0) at (0,0)[i,label=right:$ bc{=}ad$]{} edge (1);
\end{tikzpicture}
\ \ 
\begin{tikzpicture}[baseline=0pt]
\node at (0,-1)[n]{$D^{6}_{1,2}$};
\node at (0,-1.5)[n]{\phantom{.}};
\node(5) at (0,5)[i,label=right:$1$]{};
\node(4) at (0,4)[i,label=right:$a$]{} edge (5);
\node(3) at (0,3)[,label=right:$b$]{} edge (4);
\node(2) at (0,2)[,label=right:$c{=}ac{=}ab$]{} edge (3);
\node(1) at (0,1)[,label=right:$d{=}b^2$]{} edge (2);
\node(0) at (0,0)[i,label=right:$ bc{=}ad$]{} edge (1);
\end{tikzpicture}
\ \ 
\begin{tikzpicture}[baseline=0pt]
\node at (0,-1)[n]{$D^{6}_{1,3}$};
\node at (0,-1.5)[n]{\phantom{.}};
\node(5) at (0,5)[i,label=right:$1$]{};
\node(4) at (0,4)[i,label=right:$a$]{} edge (5);
\node(3) at (0,3)[,label=right:$b{=}ab$]{} edge (4);
\node(2) at (0,2)[,label=right:$c{=}b^2{=}ac$]{} edge (3);
\node(1) at (0,1)[,label=right:$d$]{} edge (2);
\node(0) at (0,0)[i,label=right:$ bc{=}ad$]{} edge (1);
\end{tikzpicture}
\ \ 
\begin{tikzpicture}[baseline=0pt]
\node at (0,-1)[n]{$D^{6}_{1,4}$};
\node at (0,-1.5)[n]{\phantom{.}};
\node(5) at (0,5)[i,label=right:$1$]{};
\node(4) at (0,4)[,label=right:$a$]{} edge (5);
\node(3) at (0,3)[i,label=right:$b{=}a^2$]{} edge (4);
\node(2) at (0,2)[,label=right:$c$]{} edge (3);
\node(1) at (0,1)[,label=right:$d{=}ac$]{} edge (2);
\node(0) at (0,0)[i,label=right:$ bc{=}ad$]{} edge (1);
\end{tikzpicture}
\ \ 
\begin{tikzpicture}[baseline=0pt]
\node at (0,-1)[n]{$D^{6}_{1,5}$};
\node at (0,-1.5)[n]{\phantom{.}};
\node(5) at (0,5)[i,label=right:$1$]{};
\node(4) at (0,4)[,label=right:$a$]{} edge (5);
\node(3) at (0,3)[,label=right:$b$]{} edge (4);
\node(2) at (0,2)[,label=right:$c$]{} edge (3);
\node(1) at (0,1)[,label=right:$d{=}b^2{=}ac{=}a^2$]{} edge (2);
\node(0) at (0,0)[i,label=right:$ bc{=}ad$]{} edge (1);
\end{tikzpicture}
\ \ 
\begin{tikzpicture}[baseline=0pt]
\node at (0,-1)[n]{$D^{6}_{1,6}$};
\node at (0,-1.5)[n]{\phantom{.}};
\node(5) at (0,5)[i,label=right:$1$]{};
\node(4) at (0,4)[,label=right:$a$]{} edge (5);
\node(3) at (0,3)[,label=right:$b$]{} edge (4);
\node(2) at (0,2)[,label=right:$c{=}a^2$]{} edge (3);
\node(1) at (0,1)[,label=right:$d{=}b^2{=}ac{=}ab$]{} edge (2);
\node(0) at (0,0)[i,label=right:$ bc{=}ad$]{} edge (1);
\end{tikzpicture}
\ \ 
\begin{tikzpicture}[baseline=0pt]
\node at (0,-1)[n]{$D^{6}_{1,7}$};
\node at (0,-1.5)[n]{\phantom{.}};
\node(5) at (0,5)[i,label=right:$1$]{};
\node(4) at (0,4)[,label=right:$a$]{} edge (5);
\node(3) at (0,3)[,label=right:$b{=}a^2$]{} edge (4);
\node(2) at (0,2)[,label=right:$c{=}ab$]{} edge (3);
\node(1) at (0,1)[,label=right:$d{=}b^2{=}ac$]{} edge (2);
\node(0) at (0,0)[i,label=right:$ bc{=}ad$]{} edge (1);
\end{tikzpicture}
\ \ 
\begin{tikzpicture}[baseline=0pt]
\node at (0,-1)[n]{$D^{6}_{1,8}$};
\node at (0,-1.5)[n]{\phantom{.}};
\node(5) at (0,5)[i,label=right:$ \top 0$]{};
\node(4) at (0,4)[i,label=right:$1$]{} edge (5);
\node(3) at (0,3)[i,label=right:$a$]{} edge (4);
\node(2) at (0,2)[,label=right:$b$]{} edge (3);
\node(1) at (0,1)[i,label=right:$0{=}ab$]{} edge (2);
\node(0) at (0,0)[i,label=right:$ $]{} edge (1);
\end{tikzpicture}
\ \ 
\begin{tikzpicture}[baseline=0pt]
\node at (0,-1)[n]{$D^{6}_{1,9}$};
\node at (0,-1.5)[n]{\phantom{.}};
\node(5) at (0,5)[i,label=right:$ $]{};
\node(4) at (0,4)[i,label=right:$1$]{} edge (5);
\node(3) at (0,3)[i,label=right:$a{=}\top a$]{} edge (4);
\node(2) at (0,2)[,label=right:$b{=}\top 0{=}\top b$]{} edge (3);
\node(1) at (0,1)[,label=right:$0$]{} edge (2);
\node(0) at (0,0)[i,label=right:$ ab$]{} edge (1);
\end{tikzpicture}
\ \ 
\begin{tikzpicture}[baseline=0pt]
\node at (0,-1)[n]{$D^{6}_{1,10}$};
\node at (0,-1.5)[n]{\phantom{.}};
\node(5) at (0,5)[i,label=right:$ \top 0$]{};
\node(4) at (0,4)[i,label=right:$1$]{} edge (5);
\node(3) at (0,3)[,label=right:$a$]{} edge (4);
\node(2) at (0,2)[,label=right:$b{=}a^2$]{} edge (3);
\node(1) at (0,1)[i,label=right:$0{=}ab$]{} edge (2);
\node(0) at (0,0)[i,label=right:$ $]{} edge (1);
\end{tikzpicture}
\ \ 
\begin{tikzpicture}[baseline=0pt]
\node at (0,-1)[n]{$D^{6}_{1,11}$};
\node at (0,-1.5)[n]{\phantom{.}};
\node(5) at (0,5)[i,label=right:$ $]{};
\node(4) at (0,4)[i,label=right:$1$]{} edge (5);
\node(3) at (0,3)[,label=right:$a{=}\top a$]{} edge (4);
\node(2) at (0,2)[,label=right:$b{=}a^2{=}\top 0{=}\top b$]{} edge (3);
\node(1) at (0,1)[,label=right:$0$]{} edge (2);
\node(0) at (0,0)[i,label=right:$ ab$]{} edge (1);
\end{tikzpicture}
\ \ 
\begin{tikzpicture}[baseline=0pt]
\node at (0,-1)[n]{$D^{6}_{1,12}$};
\node at (0,-1.5)[n]{\phantom{.}};
\node(5) at (0,5)[i,label=right:$ \top b$]{};
\node(4) at (0,4)[i,label=right:$a{=}a0$]{} edge (5);
\node(3) at (0,3)[i,label=right:$1$]{} edge (4);
\node(2) at (0,2)[i,label=right:$0$]{} edge (3);
\node(1) at (0,1)[i,label=right:$b{=}ab$]{} edge (2);
\node(0) at (0,0)[i,label=right:$ $]{} edge (1);
\end{tikzpicture}
\ \ 
\begin{tikzpicture}[baseline=0pt]
\node at (0,-1)[n]{$D^{6}_{1,13}$};
\node at (0,-1.5)[n]{\phantom{.}};
\node(5) at (0,5)[i,label=right:$ \top b$]{};
\node(4) at (0,4)[i,label=right:$a{=}0^2$]{} edge (5);
\node(3) at (0,3)[,label=right:$0$]{} edge (4);
\node(2) at (0,2)[i,label=right:$1$]{} edge (3);
\node(1) at (0,1)[i,label=right:$b{=}0b{=}ab$]{} edge (2);
\node(0) at (0,0)[i,label=right:$ $]{} edge (1);
\end{tikzpicture}
\ \ 
\begin{tikzpicture}[baseline=0pt]
\node at (0,-1)[n]{$D^{6}_{1,14}$};
\node at (0,-1.5)[n]{\phantom{.}};
\node(5) at (0,5)[i,label=right:$ 0^2{=}\top b$]{};
\node(4) at (0,4)[,label=right:$a$]{} edge (5);
\node(3) at (0,3)[,label=right:$0{=}0b{=}ab$]{} edge (4);
\node(2) at (0,2)[i,label=right:$1$]{} edge (3);
\node(1) at (0,1)[i,label=right:$b$]{} edge (2);
\node(0) at (0,0)[i,label=right:$ $]{} edge (1);
\end{tikzpicture}
\ \ 
\begin{tikzpicture}[baseline=0pt]
\node at (0,-1)[n]{$D^{6}_{1,15}$};
\node at (0,-1.5)[n]{\phantom{.}};
\node(5) at (0,5)[i,label=right:$ a^2{=}0b$]{};
\node(4) at (0,4)[,label=right:$0$]{} edge (5);
\node(3) at (0,3)[,label=right:$a{=}ab$]{} edge (4);
\node(2) at (0,2)[i,label=right:$b$]{} edge (3);
\node(1) at (0,1)[i,label=right:$1$]{} edge (2);
\node(0) at (0,0)[i,label=right:$ $]{} edge (1);
\end{tikzpicture}
\ \ 
\begin{tikzpicture}[baseline=0pt]
\node at (0,-1)[n]{$D^{6}_{1,16}$};
\node at (0,-1.5)[n]{\phantom{.}};
\node(5) at (0,5)[i,label=right:$ a^2{=}0b$]{};
\node(4) at (0,4)[,label=right:$0{=}ab$]{} edge (5);
\node(3) at (0,3)[,label=right:$a{=}b^2$]{} edge (4);
\node(2) at (0,2)[,label=right:$b$]{} edge (3);
\node(1) at (0,1)[i,label=right:$1$]{} edge (2);
\node(0) at (0,0)[i,label=right:$ $]{} edge (1);
\end{tikzpicture}
\ \ 
\begin{tikzpicture}[baseline=0pt]
\node at (0,-1)[n]{$D^{6}_{1,17}$};
\node at (0,-1.5)[n]{\phantom{.}};
\node(5) at (0,5)[i,label=right:$ a^2{=}0b$]{};
\node(4) at (0,4)[,label=right:$0{=}b^2{=}ab$]{} edge (5);
\node(3) at (0,3)[,label=right:$a$]{} edge (4);
\node(2) at (0,2)[,label=right:$b$]{} edge (3);
\node(1) at (0,1)[i,label=right:$1$]{} edge (2);
\node(0) at (0,0)[i,label=right:$ $]{} edge (1);
\end{tikzpicture}
\ \ 
\begin{tikzpicture}[baseline=0pt]
\node at (0,-1)[n]{$D^{6}_{2,1,2}$};
\node at (0,-1.5)[n]{\phantom{.}};
\node(5) at (0,4)[i,label=right:$1$]{};
\node(4) at (0,3)[i,label=right:$a$]{} edge (5);
\node(3) at (1,2)[i,label=right:$c$]{} edge (4);
\node(2) at (-1,2)[i,label=left:$b$]{} edge (4);
\node(1) at (0,1)[,label=right:$d$]{} edge (2) edge (3);
\node(0) at (0,0)[i,label=right:$ bc{=}ad$]{} edge (1);
\end{tikzpicture}
\ \ 
\begin{tikzpicture}[baseline=0pt]
\node at (0,-1)[n]{$D^{6}_{2,2}$};
\node at (0,-1.5)[n]{\phantom{.}};
\node(5) at (0,4)[i,label=right:$1$]{};
\node(4) at (0,3)[,label=right:$a$]{} edge (5);
\node(3) at (1,2)[i,label=right:$c{=}a^2$]{} edge (4);
\node(2) at (-1,2)[,label=left:$b$]{} edge (4);
\node(1) at (0,1)[,label=right:$d{=}b^2{=}ab$]{} edge (2) edge (3);
\node(0) at (0,0)[i,label=right:$ bc{=}ad$]{} edge (1);
\end{tikzpicture}
\ \ 
\begin{tikzpicture}[baseline=0pt]
\node at (0,-1)[n]{$D^{6}_{2,3,2}$};
\node at (0,-1.5)[n]{\phantom{.}};
\node(5) at (0,4)[i,label=right:$1$]{};
\node(4) at (0,3)[,label=right:$a$]{} edge (5);
\node(3) at (1,2)[,label=right:$c$]{} edge (4);
\node(2) at (-1,2)[,label=left:$b$]{} edge (4);
\node(1) at (0,1)[,label=right:$d{=}b^2{=}c^2{=}a^2$]{} edge (2) edge (3);
\node(0) at (0,0)[i,label=right:$ bc{=}ad$]{} edge (1);
\end{tikzpicture}
\ \ 
\begin{tikzpicture}[baseline=0pt]
\node at (0,-1)[n]{$D^{6}_{2,4,2}$};
\node at (0,-1.5)[n]{\phantom{.}};
\node(5) at (0,4)[i,label=right:$ \top 0$]{};
\node(4) at (0,3)[i,label=right:$1$]{} edge (5);
\node(3) at (1,2)[i,label=right:$b$]{} edge (4);
\node(2) at (-1,2)[i,label=left:$a$]{} edge (4);
\node(1) at (0,1)[i,label=right:$0{=}ab$]{} edge (2) edge (3);
\node(0) at (0,0)[i,label=right:$ $]{} edge (1);
\end{tikzpicture}
\ \ 
\begin{tikzpicture}[baseline=0pt]
\node at (0,-1)[n]{$D^{6}_{2,5}$};
\node at (0,-1.5)[n]{\phantom{.}};
\node(5) at (0,4)[nci,label=right:$ b\top {=}\top a$]{};
\node(4) at (0,3)[i,label=right:$1$]{} edge (5);
\node(3) at (1,2)[nci,label=right:$b{=}\top 0{=}\top b$]{} edge (4);
\node(2) at (-1,2)[nci,label=left:$a{=}0\top {=}a\top $]{} edge (4);
\node(1) at (0,1)[nc,label=right:$0{=}0b{=}a0{=}ab$]{} edge (2) edge (3);
\node(0) at (0,0)[i,label=right:$ ba$]{} edge (1);
\end{tikzpicture}
\ \ 
\begin{tikzpicture}[baseline=0pt]
\node at (0,-1)[n]{$D^{6}_{2,6}$};
\node at (0,-1.5)[n]{\phantom{.}};
\node(5) at (0,4)[i,label=right:$ \top b$]{};
\node(4) at (0,3)[i,label=right:$a{=}a0$]{} edge (5);
\node(3) at (1,2)[,label=right:$0$]{} edge (4);
\node(2) at (-1,2)[i,label=left:$1{=}0^2$]{} edge (4);
\node(1) at (0,1)[i,label=right:$b{=}0b{=}ab$]{} edge (2) edge (3);
\node(0) at (0,0)[i,label=right:$ $]{} edge (1);
\end{tikzpicture}
\ \ 
\begin{tikzpicture}[baseline=0pt]
\node at (0,-1)[n]{$D^{6}_{2,7}$};
\node at (0,-1.5)[n]{\phantom{.}};
\node(5) at (0,4)[i,label=right:$ \top b$]{};
\node(4) at (0,3)[i,label=right:$a{=}0^2{=}a0$]{} edge (5);
\node(3) at (1,2)[,label=right:$0$]{} edge (4);
\node(2) at (-1,2)[i,label=left:$1$]{} edge (4);
\node(1) at (0,1)[i,label=right:$b{=}0b{=}ab$]{} edge (2) edge (3);
\node(0) at (0,0)[i,label=right:$ $]{} edge (1);
\end{tikzpicture}
\ \ 
\begin{tikzpicture}[baseline=0pt]
\node at (0,-1)[n]{$D^{6}_{2,8}$};
\node at (0,-1.5)[n]{\phantom{.}};
\node(5) at (0,4)[i,label=right:$ 0^2{=}\top b$]{};
\node(4) at (0,3)[,label=right:$a$]{} edge (5);
\node(3) at (1,2)[,label=right:$0{=}0b{=}ab$]{} edge (4);
\node(2) at (-1,2)[i,label=left:$1$]{} edge (4);
\node(1) at (0,1)[i,label=right:$b$]{} edge (2) edge (3);
\node(0) at (0,0)[i,label=right:$ $]{} edge (1);
\end{tikzpicture}
\ \ 
\begin{tikzpicture}[baseline=0pt]
\node at (0,-1)[n]{$D^{6}_{2,9,2}$};
\node at (0,-1.5)[n]{\phantom{.}};
\node(5) at (0,4)[i,label=right:$ a^2{=}b^2$]{};
\node(4) at (0,3)[,label=right:$0{=}ab$]{} edge (5);
\node(3) at (1,2)[,label=right:$b$]{} edge (4);
\node(2) at (-1,2)[,label=left:$a$]{} edge (4);
\node(1) at (0,1)[i,label=right:$1$]{} edge (2) edge (3);
\node(0) at (0,0)[i,label=right:$ $]{} edge (1);
\end{tikzpicture}
\ \ 
\begin{tikzpicture}[baseline=0pt]
\node at (0,-1)[n]{$D^{6}_{3,1,2}$};
\node at (0,-1.5)[n]{\phantom{.}};
\node(5) at (0,4)[i,label=right:$1$]{};
\node(4) at (0,3)[,label=right:$a$]{} edge (5);
\node(3) at (1,2)[,label=right:$c$]{} edge (4);
\node(2) at (-1,2)[,label=left:$b$]{} edge (4);
\node(1) at (0,1)[,label=right:$d{=}bc{=}ab{=}a^2$]{} edge (2) edge (3);
\node(0) at (0,0)[i,label=right:$ b^2{=}c^2{=}ad$]{} edge (1);
\end{tikzpicture}
\ \ 
\begin{tikzpicture}[baseline=0pt]
\node at (0,-1)[n]{$D^{6}_{3,2}$};
\node at (0,-1.5)[n]{\phantom{.}};
\node(5) at (0,4)[i,label=right:$ \top c$]{};
\node(4) at (0,3)[i,label=right:$a$]{} edge (5);
\node(3) at (1,2)[,label=right:$b{=}ab$]{} edge (4);
\node(2) at (-1,2)[i,label=left:$1{=}0$]{} edge (4);
\node(1) at (0,1)[i,label=right:$c{=}bc{=}b^2{=}ac$]{} edge (2) edge (3);
\node(0) at (0,0)[i,label=right:$ $]{} edge (1);
\end{tikzpicture}
\ \ 
\begin{tikzpicture}[baseline=0pt]
\node at (0,-1)[n]{$D^{6}_{3,3}$};
\node at (0,-1.5)[n]{\phantom{.}};
\node(5) at (0,4)[i,label=right:$ \top c$]{};
\node(4) at (0,3)[i,label=right:$a{=}ab$]{} edge (5);
\node(3) at (1,2)[,label=right:$b$]{} edge (4);
\node(2) at (-1,2)[i,label=left:$1{=}0{=}b^2$]{} edge (4);
\node(1) at (0,1)[i,label=right:$c{=}bc{=}ac$]{} edge (2) edge (3);
\node(0) at (0,0)[i,label=right:$ $]{} edge (1);
\end{tikzpicture}
\ \ 
\begin{tikzpicture}[baseline=0pt]
\node at (0,-1)[n]{$D^{6}_{3,4}$};
\node at (0,-1.5)[n]{\phantom{.}};
\node(5) at (0,4)[i,label=right:$ $]{};
\node(4) at (0,3)[i,label=right:$a$]{} edge (5);
\node(3) at (1,2)[,label=right:$b{=}ab{=}\top c{=}\top b$]{} edge (4);
\node(2) at (-1,2)[i,label=left:$1{=}0$]{} edge (4);
\node(1) at (0,1)[,label=right:$c{=}ac$]{} edge (2) edge (3);
\node(0) at (0,0)[i,label=right:$ b^2$]{} edge (1);
\end{tikzpicture}
\ \ 
\begin{tikzpicture}[baseline=0pt]
\node at (0,-1)[n]{$D^{6}_{3,5,2}$};
\node at (0,-1.5)[n]{\phantom{.}};
\node(5) at (0,4)[i,label=right:$ ab{=}0a$]{};
\node(4) at (0,3)[,label=right:$0$]{} edge (5);
\node(3) at (1,2)[i,label=right:$b$]{} edge (4);
\node(2) at (-1,2)[i,label=left:$a$]{} edge (4);
\node(1) at (0,1)[i,label=right:$1$]{} edge (2) edge (3);
\node(0) at (0,0)[i,label=right:$ $]{} edge (1);
\end{tikzpicture}
\ \ 
\begin{tikzpicture}[baseline=0pt]
\node at (0,-1)[n]{$D^{6}_{3,6}$};
\node at (0,-1.5)[n]{\phantom{.}};
\node(5) at (0,4)[i,label=right:$ ab{=}0a$]{};
\node(4) at (0,3)[,label=right:$0{=}b^2$]{} edge (5);
\node(3) at (1,2)[,label=right:$b$]{} edge (4);
\node(2) at (-1,2)[i,label=left:$a$]{} edge (4);
\node(1) at (0,1)[i,label=right:$1$]{} edge (2) edge (3);
\node(0) at (0,0)[i,label=right:$ $]{} edge (1);
\end{tikzpicture}
\ \ 
\begin{tikzpicture}[baseline=0pt]
\node at (0,-1)[n]{$D^{6}_{3,7,2}$};
\node at (0,-1.5)[n]{\phantom{.}};
\node(5) at (0,4)[i,label=right:$\ ab{=}0a$]{};
\node(4) at (0,3)[,label=right:$\ 0{=}a^2{=}b^2$]{} edge (5);
\node(3) at (1,2)[,label=right:$b$]{} edge (4);
\node(2) at (-1,2)[,label=left:$a$]{} edge (4);
\node(1) at (0,1)[i,label=right:$1$]{} edge (2) edge (3);
\node(0) at (0,0)[i,label=right:$ $]{} edge (1);
\end{tikzpicture}
\ \ 
\begin{tikzpicture}[baseline=0pt]
\node at (0,-1)[n]{$D^{6}_{4,1}$};
\node at (0,-1.5)[n]{\phantom{.}};
\node(5) at (0.5,3)[i,label=right:$1$]{};
\node(4) at (1.5,2)[i,label=right:$b$]{} edge (5);
\node(3) at (-0.5,2)[,label=left:$a$]{} edge (5);
\node(2) at (0.5,1)[,label=right:$d{=}ab{=}bd$]{} edge (3) edge (4);
\node(1) at (-1.5,1)[i,label=left:$c{=}a^2$]{} edge (3);
\node(0) at (-0.5,0)[i,label=left:$ ad{=}bc$]{} edge (1) edge (2);
\end{tikzpicture}
\ \ 
\begin{tikzpicture}[baseline=0pt]
\node at (0,-1)[n]{$D^{6}_{4,2}$};
\node at (0,-1.5)[n]{\phantom{.}};
\node(5) at (0.5,3)[i,label=right:$ $]{};
\node(4) at (1.5,2)[i,label=right:$a{=}a0{=}\top a$]{} edge (5);
\node(3) at (-0.5,2)[i,label=left:$1$]{} edge (5);
\node(2) at (0.5,1)[i,label=right:$0$]{} edge (3) edge (4);
\node(1) at (-1.5,1)[i,label=left:$b{=}\top b$]{} edge (3);
\node(0) at (-0.5,0)[i,label=left:$ b0{=}ab$]{} edge (1) edge (2);
\end{tikzpicture}
\ \ 
\begin{tikzpicture}[baseline=0pt]
\node at (0,-1)[n]{$D^{6}_{4,3}$};
\node at (0,-1.5)[n]{\phantom{.}};
\node(5) at (0.5,3)[i,label=right:$ $]{};
\node(4) at (1.5,2)[i,label=right:$1$]{} edge (5);
\node(3) at (-0.5,2)[i,label=left:$a{=}0^2{=}\top a$]{} edge (5);
 \node(2) at (0.5,1)[,label=right:$b{=}\top b$]{} edge (3) edge (4);
\node(1) at (-1.5,1)[,label=left:$0$]{} edge (3);
\node(0) at (-0.5,0)[i,label=left:$ ab$]{} edge (1) edge (2);
\end{tikzpicture}
\ \ 
\begin{tikzpicture}[baseline=0pt]
\node at (0,-1)[n]{$D^{6}_{4,4}$};
\node at (0,-1.5)[n]{\phantom{.}};
\node(5) at (0.5,3)[i,label=right:$ $]{};
\node(4) at (1.5,2)[i,label=right:$1$]{} edge (5);
\node(3) at (-0.5,2)[,label=left:$a{=}\top 0{=}\top a$]{} edge (5);
\node(2) at (0.5,1)[,label=right:$b{=}0^2{=}a^2{=}\top b$]{} edge (3) edge (4);
\node(1) at (-1.5,1)[,label=left:$0$]{} edge (3);
\node(0) at (-0.5,0)[i,label=left:$ ab$]{} edge (1) edge (2);
\end{tikzpicture}
\ \ 
\begin{tikzpicture}[baseline=0pt]
\node at (0,-1)[n]{$D^{6}_{4,5}$};
\node at (0,-1.5)[n]{\phantom{.}};
\node(5) at (0.5,3)[i,label=right:$ a0{=}0b$]{};
\node(4) at (1.5,2)[,label=right:$0$]{} edge (5);
\node(3) at (-0.5,2)[i,label=left:$a$]{} edge (5);
\node(2) at (0.5,1)[i,label=right:$b{=}ab$]{} edge (3) edge (4);
\node(1) at (-1.5,1)[i,label=left:$1$]{} edge (3);
\node(0) at (-0.5,0)[i,label=left:$ $]{} edge (1) edge (2);
\end{tikzpicture}
\ \ 
\begin{tikzpicture}[baseline=0pt]
\node at (0,-1)[n]{$D^{6}_{4,6}$};
\node at (0,-1.5)[n]{\phantom{.}};
\node(5) at (0.5,3)[i,label=right:$ 0b{=}0a$]{};
\node(4) at (1.5,2)[i,label=right:$a{=}b^2$]{} edge (5);
\node(3) at (-0.5,2)[,label=left:$0$]{} edge (5);
\node(2) at (0.5,1)[i,label=right:$1$]{} edge (3) edge (4);
\node(1) at (-1.5,1)[,label=left:$b{=}ab$]{} edge (3);
\node(0) at (-0.5,0)[i,label=left:$ $]{} edge (1) edge (2);
\end{tikzpicture}
\ \ 
\begin{tikzpicture}[baseline=0pt]
\node at (0,-1)[n]{$D^{6}_{4,7}$};
\node at (0,-1.5)[n]{\phantom{.}};
\node(5) at (0.5,3)[i,label=right:$ b^2{=}0a$]{};
\node(4) at (1.5,2)[i,label=right:$a$]{} edge (5);
\node(3) at (-0.5,2)[,label=left:$0$]{} edge (5);
\node(2) at (0.5,1)[i,label=right:$1$]{} edge (3) edge (4);
\node(1) at (-1.5,1)[,label=left:$b{=}ab$]{} edge (3);
\node(0) at (-0.5,0)[i,label=left:$ $]{} edge (1) edge (2);
\end{tikzpicture}
\ \ 
\begin{tikzpicture}[baseline=0pt]
\node at (0,-1)[n]{$D^{6}_{4,8}$};
\node at (0,-1.5)[n]{\phantom{.}};
\node(5) at (0.5,3)[i,label=right:$\ a^2{=}a0{=}0b$]{};
\node(4) at (1.5,2)[,label=right:$\,0{=}b^2{=}ab$]{} edge (5);
\node(3) at (-0.5,2)[,label=left:$a$]{} edge (5);
\node(2) at (0.5,1)[,label=right:$b$]{} edge (3) edge (4);
\node(1) at (-1.5,1)[i,label=left:$1$]{} edge (3);
\node(0) at (-0.5,0)[i,label=left:$ $]{} edge (1) edge (2);
\end{tikzpicture}
\ \ 
\begin{tikzpicture}[baseline=0pt]
\node at (0,-1)[n]{$D^{6}_{4,9}$};
\node at (0,-1.5)[n]{\phantom{.}};
\node(5) at (0.5,3)[i,label=right:$\ 0b{=}0a{=}a^2$]{};
\node(4) at (1.5,2)[,label=right:$a{=}b^2$]{} edge (5);
\node(3) at (-0.5,2)[,label=left:$0{=}ab$]{} edge (5);
\node(2) at (0.5,1)[i,label=right:$1$]{} edge (3) edge (4);
\node(1) at (-1.5,1)[,label=left:$b$]{} edge (3);
\node(0) at (-0.5,0)[i,label=left:$ $]{} edge (1) edge (2);
\end{tikzpicture}
\ \ 
\begin{tikzpicture}[baseline=0pt]
\node at (0,-1)[n]{$D^{6}_{4,10}$};
\node at (0,-1.5)[n]{\phantom{.}};
\node(5) at (0.5,3)[i,label=right:$\ b^2{=}0a{=}a^2$]{};
\node(4) at (1.5,2)[,label=right:$a$]{} edge (5);
\node(3) at (-0.5,2)[,label=left:$0{=}ab$]{} edge (5);
\node(2) at (0.5,1)[i,label=right:$1$]{} edge (3) edge (4);
\node(1) at (-1.5,1)[,label=left:$b$]{} edge (3);
\node(0) at (-0.5,0)[i,label=left:$ $]{} edge (1) edge (2);
\end{tikzpicture}
\begin{center}\small
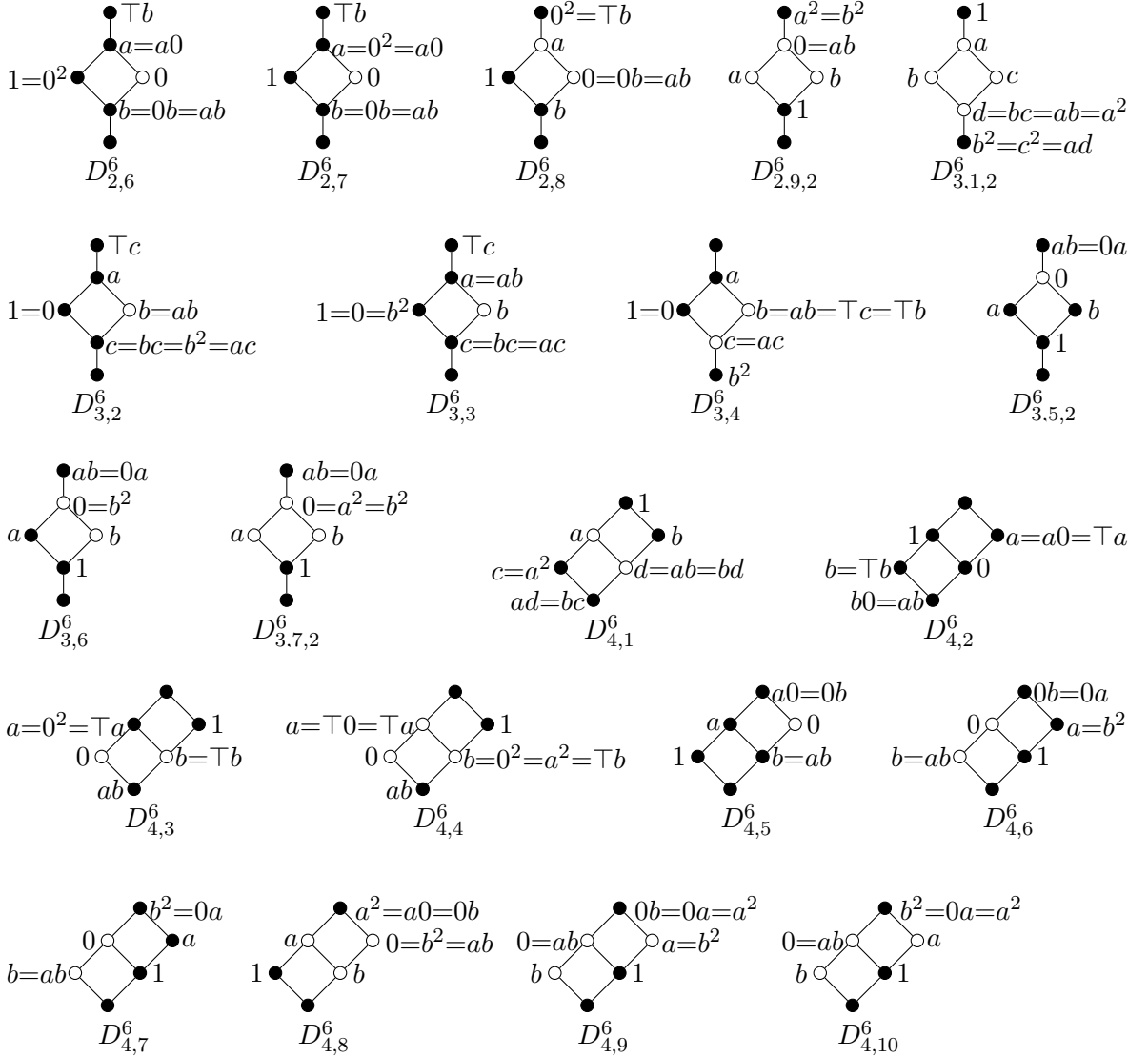
    Figure 2. DInFL-algebras and DqRAs up to cardinality six. A longer list up to cardinality\\ eight, together with the Python code used to generate this data can be found at \url{https://github.com/jipsen/Distributive-quasi-relation-algebras-and-DInFL}.
\end{center}

\subsection{Known representations of DInFL-algebras and DqRAs}\label{sec:reps-of-algebras}

In Table~\ref{tab:Rep-DqRA<=5} and~\ref{tab:Rep-DqRA=6} below, we list finite DqRAs up to size 6 and what is currently known about their representability. In each case, if there is a representation for the described DqRA, then the DInFL-reduct can be represented using the same partially ordered equivalence relation and the same $\alpha$ as for the DqRA. 

We briefly mention the smallest non-cyclic DInFL-algebra (and DqRA)  which has seven elements (see~\cite[Figure 1]{RDqRA25} or $D^7_{3,1,2}$). There are two non-cyclic DqRAs with that same DInFL-algebra reduct.  Its dual frame has the poset structure $\bowtie$. No representation is known for this algebra. Moreover, it is known from~\cite[Theorem 5.12]{RDqRA25} that if a representation exists, it will be infinite.

\begin{table}[ht!]
    \centering
\tabcolsep9pt
\begin{tabular}{|c|p{0.11\linewidth}|p{0.5\linewidth}|}\hline
DInFL-algebra &Description of  $\neg$& Representation/comment  \\ \hline 
$D^1_{1,1}$& $\neg={\sim}$& See~\cite[Example 5.1]{RDqRA25}.\\ \hline 
$D^2_{1,1}$& $\neg={\sim}$& See~\cite[Example 5.2]{RDqRA25}.  \\ \hline 
$D^3_{1,1}$& $\neg={\sim}$& No known representation. If a representation exists, it must be infinite (\cite[Example 5.14]{RDqRA25}). \\ \hline 
$D^3_{1,2}$& $\neg={\sim}$&  See~\cite[Example 5.1]{RDqRA25}. \\ \hline 
$D^4_{1,1}$ and $D^4_{1,2}$& $\neg= {\sim}$ &   No known representations. If representations exist, they must be infinite (\cite[Theorem 5.12]{RDqRA25}).\\ \hline 
$D^4_{1,3}$& $\neg = {\sim}$ &  Infinite representation~\cite{Mad2010,CR-Sugihara}, finite representation~\cite[Theorem 5.1]{CMR2025}\\ \hline 
$D^4_{1,4}$ & $\neg = {\sim}$ & See~\cite[Example 5.4]{RDqRA25}. \\ \hline 
$D^4_{2,1,2}$&$\neg={\sim}$ & See~\cite[Example 5.2]{RDqRA25}.   \\ \hline 
$D^{4,\neg a= a}_{2,1,2}$& $\neg a = a$ &
See~\cite[Example 5.6]{RDqRA25}.
\\ \hline 
$D^4_{2,2}$& $\neg = {\sim} $& See~\cite[Example 5.7]{RDqRA25}. \\ \hline 
$D^4_{2,3}$& $\neg = {\sim}$& See~\cite[Example 5.8]{RDqRA25}.  \\ \hline 
$D^4_{3,1}$& $\neg={\sim}$ &  No known representation. \\ \hline 
$D^4_{3,2}$& $\neg={\sim}$ & See~\cite[Example 5.1]{RDqRA25}. \\ \hline 
$D^5_{1,1}$ to $D^5_{1,5}$& $\neg={\sim}$ & No known representations. If  representations exist, they must be infinite (\cite[Theorem 5.12]{RDqRA25}).  \\ \hline 
$D^5_{1,6}$& $\neg={\sim}$ & Infinite representation~\cite{CR-Sugihara}, finite representation~\cite{CMR2025}  \\ \hline 
$D^5_{1,7}$& $\neg={\sim}$ &   See~\cite[Example 5.5]{RDqRA25}. \\ \hline 
$D^5_{1,8}$& $\neg={\sim}$ &  
See~\cite[Example 5.5]{RDqRA25}.\\ \hline 
\end{tabular}
\caption{Finite DqRAs of size $\leqslant 5$ and our knowledge of their representability.}
\label{tab:Rep-DqRA<=5}
\end{table}

\begin{table}[ht!]
    \centering
\tabcolsep9pt
\begin{tabular}{|c|p{0.10\linewidth}|p{0.51\linewidth}|}\hline
DInFL-algebra &Description of  $\neg$& Representation/comment  \\ \hline 
$D^6_{1,1}$ to $D^6_{1,11}$& $\neg  ={\sim}$&
No known representations. If representations exist, they must be infinite (\cite[Theorem 5.12]{RDqRA25}). 
\\ \hline 
$D^6_{1,12}$& $\neg  ={\sim}$&
Infinite representation~\cite{Mad2010, CR-Sugihara}, finite representation~\cite[Theorem 5.1]{CMR2025}
\\ \hline 
$D^6_{1,13}$& $\neg  ={\sim}$&
Apply~\cite[Theorem 4.1]{CMR2025} to 
$D^6_{1,13}\cong D^3_{1,2}[D^4_{1,4}]$.
\\ \hline 
$D^6_{1,14}$ to $D^6_{1,17}$& $\neg  ={\sim}$& No known representations. 
\\ \hline 
$D^6_{2,1}$ to $D^6_{2,3}$
& $\neg ={\sim}$ & No known representations. \\ \hline
$D^{6,\neg b=b}_{2,1,2}$, $D^{6,\neg b =b}_{2,3,2}$ 
& $\neg b =b$& No known representations. \\ \hline
$D^{6,\neg b =c}_{3,1,2}$
& $\neg b =c$& No known representation. \\ \hline
$D^6_{2,4,2}$ & $\neg={\sim}$ &
Apply~\cite[Theorem 4.1]{CMR2025} to 
$D^6_{2,4,2}\cong D^3_{1,2}[D^4_{2,1,2}]$.
\\ \hline
$D^{6,\neg a =a}_{2,4,2}$ & $\neg a =a$ &  
Apply~\cite[Theorem 4.1]{CMR2025} to 
$D^{6,\neg a =a}_{2,4,2}\cong D^3_{1,2}[D^{4,\neg a =a}_{2,1,2}]$.
\\ \hline
$D^6_{2,5}$ & $\neg a = a$ & 
See~\cite[Example 5.9]{RDqRA25}. \\ \hline 
$D^6_{2,6}$& $\neg  ={\sim}$& Apply~\cite[Theorem 4.1]{CMR2025} to $D^6_{2,6}\cong D^3_{1,2}[D^4_{2,2}]$.  \\ \hline 
$D^6_{2,7}$& $\neg  ={\sim}$&
Apply~\cite[Theorem 4.1]{CMR2025} to 
$D^6_{2,7}\cong D^3_{1,2}[D^4_{2,3}]$.
\\ \hline 
$D^6_{2,8}$, $D^6_{2,9,2}$, $D^6_{3,1,2}$  & $\neg = {\sim}$ & No known representations. \\ \hline
$D^{6,\neg a =b}_{2,9,2}$, $D^{6,\neg a =b}_{3,5,2}$, $D^{6,\neg a =b}_{3,7,2}$ & $\neg a =b$& No known representations. 
\\ \hline 
$D^6_{3,2}$& $\neg={\sim}$ & No known representation.\\ \hline 
$D^6_{3,3}$& $\neg  ={\sim}$&
Apply~\cite[Theorem 4.1]{CMR2025} to 
$D^6_{3,3}\cong D^3_{1,2}[D^4_{3,2}]$.
\\ \hline 
$D^6_{3,4}$ & $\neg  ={\sim}$&
No known representation. 
\\ \hline 
$D^6_{3,5,2}$ & $\neg  ={\sim}$&
See~\cite[Example 6]{CMR2026}. 
\\ \hline 
$D^6_{3,6}$, $D^6_{3,7}$& $\neg  ={\sim}$&
No known representations. \\
\hline
$D^6_{4,1}$ & $\neg = {\sim}$ & No known representation. \\ \hline 
$D^6_{4,2}$& $\neg ={\sim}$& See~\cite[Example 5.10]{RDqRA25}. Also, $D^6_{4,2} \cong D^2_{1,1} \times D^3_{1,2}$ \\ \hline
$D^6_{4,3}$ to $D^6_{4,7}$ & $\neg = {\sim}$ & No known representations. \\ \hline 
$D^6_{4,8}$ & $\neg = {\sim}$ & See~\cite{CR-pregroup,JS2023} \\ \hline 
$D^6_{4,9}$, $D^6_{4,10}$  & $\neg = {\sim}$ & No known representations. \\ \hline 
\end{tabular}
\caption{Finite DqRAs of size 6 and our knowledge of their representability. Algebras $D^6_{1,1}$ to $D^6_{1,17}$ are chains, the 23 algebras $D^6_{2,1}$ to $D^6_{3,7}$ have frames with poset $\mathbf{2} \times \mathbf{2}$, algebras $D^6_{4,1}$ to $D^6_{4,10}$ have lattices with order  structure $\mathbf{2}\times \mathbf{3}$.}
\label{tab:Rep-DqRA=6}
\end{table}

We note some additional consequences of~\cite[Theorem 4.1]{CMR2025}. Since $D^6_{1,8} \cong D^3_{1,2}[D^4_{1,1}]$, we have that if $D^6_{1,8}$ is not representable, then neither is $D^4_{1,1}$. Two similar instances of this arise from $D^6_{1,10} \cong D^3_{1,2} [ D^4_{1,2} ]$ and  $D^6_{3,2}\cong D^3_{1,2}[D^4_{3,1}]$. Since $\mathsf{RDqRA}$ is closed under products, if $D^6_{4,1}$ is not representable, then $D^3_{1,1}$ will not be representable, since $D^6_{4,1} \cong D^2_{1,1} \times D^3_{1,1}$.

\section{Conclusion}

We have presented  DInFL-frames and DqRA-frames that are dual to complete perfect distributive involutive FL-algebras and distributive quasi relation algebras respectively. We modified our presentation of the frames  from~\cite{CJR2024} in order to make the axioms more transparent and readable. Morphisms are described which provide a duality between the classes of complete perfect algebras and their dual frames. 
In order to obtain a duality for algebras which are not necessarily complete in Section~\ref{sec:dual-spaces} we use a modification of  Priestley duality for unbounded algebras. 

In Section~\ref{sec:fin-models},
we catalogued which finite DInFL-algebras and DqRAs are known to be representable. One of the most interesting algebras for which no representation is currently known is the three-element MV-algebra (named $D^3_{1,1}$ in our list). Another intriguing open problem is whether the DqRA term-subreduct of McKenzie's relation algebra $14_{37}$ can be represented using the construction from~\cite{RDqRA25} (outlined in  Section~\ref{sec:RDqRAs}). We expect the search for representations of DInFL-algebras and DqRAs to be an appealing area of research. 

\subsection*{Acknowledgements}
The first author would like to thank Chapman University for its hospitality during a visit in September/October 2023. 

\end{document}